\documentclass[journal,twoside,web]{ieeecolor}
\usepackage{lcsys}
\usepackage{cite}
\usepackage{algorithmic}
\usepackage{graphicx,nicefrac}
\usepackage{textcomp}
\usepackage{latexsym,amsfonts,amssymb,amsmath,amscd,mathrsfs}
\usepackage[utf8]{inputenc}
\usepackage{times} 
\usepackage{xargs} 
\usepackage{xspace}
\usepackage{setspace}
\usepackage{lineno} 
\modulolinenumbers[5] 
\usepackage{subcaption}
\usepackage{float}
\usepackage[x11names,table]{xcolor}
\usepackage{tikz}
\usetikzlibrary{arrows}
\usepackage{pgfplots}
\usepackage{pstricks}
\usepackage{pst-node}
\usepackage[mathscr]{euscript}
\usepackage[normalem]{ulem}
\usepackage[ruled,lined,linesnumbered,commentsnumbered]{algorithm2e}
\usepackage{booktabs}
\usepackage{longtable}
\usepackage{colortbl}
\usepackage{multirow}
\usepackage{pifont}
\usepackage{array}
\usepackage{xcolor}
\usepackage{mathtools}

\usepackage{soul}
\usepackage{lscape}

\usepackage[font=small]{caption}
\usepackage{changepage}

\usepackage[figuresright]{rotating}

\newtheorem{theorem}{Theorem}

\newtheorem{definition}{Definition}
\newtheorem{proposition}[theorem]{Proposition}

\newtheorem{lemma}[theorem]{Lemma}

\newtheorem{remark}{Remark}

\newtheorem{assumption}{Assumption}

\usetikzlibrary{positioning}

\title{\LARGE \bf Failure and Success in Single-Drug Control of Antimicrobial Resistance}
\author{Alejandro Anderson$^{1}$, Rami Katz$^{1}$,  Francesca Calà Campana$^{1}$, Giulia Giordano$^{1}$
\thanks{$^{1}$ Industrial Engineering Department, University of Trento, Italy. {Email: \{alejandro.anderson,ram.kats,f.calacampana,giulia.giordano\}@unitn.it}}
\thanks{Work supported by the European Union through the Next Generation EU, Mission 4, Component 2, PRIN 2022 grant PRIDE (2022LP77J4, CUP E53D23000720006) and the ERC INSPIRE grant (101076926).}
}

\pgfplotsset{compat=1.18}

\begin{document}
\date{}
\maketitle
\thispagestyle{empty}
\pagestyle{empty}

\begin{abstract}
We propose a mathematical model of Antimicrobial Resistance in the host to predict the failure of two antagonists of bacterial growth: the immune response and a single-antibiotic therapy. After characterising the initial bacterial load that cannot be cleared by the immune system alone, we define the success set of initial conditions for which an infection-free equilibrium can be reached by a viable single-antibiotic therapy, and we provide a rigorously defined inner approximation of the set.
For initial conditions within the success set, we propose an optimal control framework to design single-drug therapies.
\end{abstract}


\section{Introduction and Motivation}

\IEEEPARstart{A}{ntibiotics} are crucial in the treatment of bacterial infections \cite{hutchings2019antibiotics}. However, bacterial populations exposed to antibiotic selective pressure can develop mechanisms to become insensitive to their action, and thus avoid extinction. This leads to the emergence of Antimicrobial Resistance (AMR): the decline of sensitive strains is paralleled by the surge of resistant ones, for which the treatment is ineffective.
AMR is a major threat to human health worldwide and causes globally 4.95 million deaths/year \cite{murray2022global}.

Deterministic ODE-based models for AMR in a host (see for instance \cite{Garber1987,fujikawa2003new,Dagata08,Gehring2010,Gomes2013} and the recent surveys \cite{Blanquart2019,Tetteh2020}) that capture the evolution of bacterial populations under antibiotic pressure, emergence of resistance due to natural selection and immune system response \cite{Dagata08,Pugliese2008,smith2011mathematical,Gomes2013},
are fundamental to design optimal drug administration strategies to contrast or prevent AMR \cite{hernandez2021switching,tetteh2023scheduling,Katriel2024}.
Also pharmacokinetic (PK) and pharmacodynamic (PD) considerations \cite{regoes2004pharmacodynamic,Nielsen2013,udekwu2018pharmacodynamic,Katriel2024} need to be embedded in the model, to faithfully describe the evolution of the drug concentration in the host.

Building upon this rich literature, we propose an in-host mathematical model of AMR that captures the evolution of bacterial populations and the emergence of resistance by including a PK/PD description of antibiotic effects, drug-induced mutations and the immune system response, and we rigorously define a clinically relevant set of viable therapies that capture the periodic nature of drug administration (Section~\ref{sec:KeyAssumptions}).
In the absence of antibiotics, we characterise initial conditions for which the immune response alone cannot clear the infection and a therapy is required (Section~\ref{sec:immunefailure}). Then, in Section~\ref{sec:antibiotic}, we study the asymptotic behavior of the system in the presence of antibiotics, and rigorously define a success set of initial conditions for which there exists a viable therapy that clears the infection, and its complementary failure set for which the use of combinations of more drugs \cite{anderson2025invariant,hernandez2021switching,tetteh2023scheduling} is recommended to contrast the emergence of resistance and PK/PD limitations; we show that both sets are non-empty and we propose a constructive inner approximation of the success set, which we then compute numerically (Section~\ref{sec:numeric}). For initial conditions within the success set, Section~\ref{sec:optimal} formulates an optimal control problem to look for viable single-drug therapies that clear the infection while minimising a desired performance metric.

\section{PK/PD Model for AMR in the
Host}\label{sec:KeyAssumptions}

Given a drug and a microorganism, the \textit{Minimum Inhibitory Concentration} (MIC) is the lowest drug concentration that inhibits the visible microorganism growth and can be determined in standard phenotype test systems \cite{kowalska2021minimum,kahlmeter2003european,andrews2001determination}.

We consider fixed $u_M$ and $u_{\max}$, where $u_M >0$ is the MIC for susceptible bacteria and $u_{\max} > u_M$ is the maximum tolerated drug concentration. The MIC for resistant bacteria is above $u_{\max}$, hence no tolerated drug concentration can suppress their growth.
We consider a therapy with a single antibiotic administered so that its concentration $u(t) \in \mathcal{U} = \{\upsilon \colon 0 \leq \upsilon \leq u_{\max}\}$ evolves as follows \cite{regoes2004pharmacodynamic}:
\begin{equation}\label{eq:udynamic}
\begin{cases}
\dot{u}(t) = -\delta u(t), ~\tau_j \leq t < \tau_{j+1} \\
u(\tau_j) = w_j + u(\tau_j^-), ~j \in \mathbb{N}_0,
\end{cases}
\end{equation}
where $u(\tau_j^-) = \lim_{t \to \tau_j^-} u(t)$.
The therapy starts at time $t=0$. Drug administration occurs at an increasing sequence of times $\tau_j$, $j \in \mathbb{N}_0$, and $w_j \in \mathcal{W}=\{\upsilon \colon 0 < \upsilon \le u_{\max}-u_M\}$ is the dose administered at time $\tau_j$. Between administration times, the drug concentration decays exponentially.

\begin{definition}\label{def:viableu} \textbf{(Viable therapy.)}
A \textit{viable therapy} $u(t)$ evolves according to \eqref{eq:udynamic}, with $w_j\in \mathcal{W}$ $\forall \tau_j$, so that $u(t)\in \{\upsilon \colon u_M \leq \upsilon \leq u_{\max}\}$ $\forall t$.
\end{definition}

In particular, we consider a class of therapies with constant dose $w_j \equiv w\in \mathcal{W}$ and constant administration period $\Delta \tau=\tau_{j+1}-\tau_j>0$, $\forall$ $j \in \mathbb{N}_0$; see e.g. Fig.~\ref{fig:s_r_udynamics}a.

\begin{definition}\label{def:periodicu} \textbf{(Periodic therapy.)}
A \textit{periodic therapy} $u(t)$, $t \geq 0$, with $0 \leq \tau_0 \leq \Delta \tau$, evolves according to \eqref{eq:udynamic} with $w_j \equiv w \in \mathcal{W}$ and $\tau_{j+1}-\tau_j=\Delta \tau$ $\forall j \in \mathbb{N}_0$ for $t \geq \tau_0$.
Moreover, if $\tau_0=0$, $u(\tau_0)= w+u_M$; if $\tau_0>0$, $u(t)=u_M e^{-\delta (t-\tau_0)}$ for $0 \leq t < \tau_0$.
\end{definition}

This amounts to considering a possible initial adjustment period $[0, \tau_0)$ in which an initial test dose larger than $u_M$ is given at time $t=0$ to check the patient's reaction and rule out adverse effects; then, when the drug concentration reaches $u_M$ (at $t=\tau_0$), a periodic therapy starts.

The set of viable therapies is non-empty. In fact, if we consider a periodic therapy, given $w$, we can always choose $\Delta \tau$ so that the therapy is viable, as shown next.

\begin{proposition}\label{propo:impulsive_condition}
Denote $\ell_{\text{max}} \doteq u_{\text{max}}/ u_M > 1$ and fix $\ell$, with $0<\ell \leq \ell_{\text{max}}-1$, so that $w=\ell \, u_M \in \mathcal{W}$.
A periodic therapy $u(t)$, with constant dose $w$ and period $\Delta \tau>0$, is viable if $\Delta \tau = \delta^{-1}\ln{(\ell+1)}$.
\end{proposition}
\begin{proof}
For $0 \leq t < \tau_0$ (assuming that $\tau_0>0$), the lower bound on $u(t)$ holds trivially. The upper bound holds since $u(0) = u_Me^{\delta \tau_0}\leq u_Me^{\delta \Delta\tau} = u_M (\ell+1) \leq u_{\text{max}}$. 

Now consider $t \geq \tau_0$. The drug concentration evolves as $u(t) = (w+u_M) e^{-\delta (t-\tau_j)}$ for $\tau_j \leq t< \tau_j+\Delta{\tau}$, $j \in \mathbb{N}_0$.
Since $\Delta{\tau} = \delta^{-1}\ln{(\ell+1)}$ and $w=\ell \, u_M$, we have $u((\tau_j+\Delta \tau)^-)= (\ell+1) u_M e^{-\ln (\ell+1)} = u_M$ for all $j \in \mathbb{N}_0$. 
Hence, $u(t)\ge u_M$ $\forall$ $t\ge \tau_0$. 
Since $w \leq u_{\max}-u_M$, the chosen $\Delta{\tau}$ also guarantees $u(t)\le u_{\max}$ $\forall$ $t\ge \tau_0$.
\end{proof}

We denote as $\mathcal{V}$ the \textbf{set of viable periodic therapies as in Proposition~\ref{propo:impulsive_condition}}, on which we focus throughout this work.

\begin{remark}\label{rem:viable}
Given two doses $u_{0,1} > u_{0,2}$, the corresponding therapies $u_{1},u_{2} \in \mathcal{V}$ do not satisfy $u_{1}(t) \geq u_{2}(t)$ $\forall t$, hence a standard order relation does not apply.
\end{remark}

For the death rate $D(u)$ of sensitive bacteria exposed to a drug concentration $u$, consider the Hill function \cite{udekwu2018pharmacodynamic,regoes2004pharmacodynamic}
\begin{equation}\label{eq:HillFunction}
\begin{array}{r}
D(u) = E_{\max} \frac{(u/EC_{50})^k}{1 + (u/EC_{50})^k},
\end{array}
\end{equation}
with Hill coefficient $k$, maximum drug-mediated death rate $E_{\max}$, and drug concentration $EC_{50}$ yielding half-maximal death rate, $D(EC_{50})=E_{\max}/2$;
see e.g. Fig.~\ref{fig:s_r_udynamics}b.

We now introduce the ordinary-differential-equation model capturing the evolution of bacterial populations in the host, considering both susceptible bacteria $s$ and the total bacterial population $b = s+r$, where $r$ denotes resistant bacteria:
\begin{equation}\label{eq:originalModel}
\begin{cases}
\dot b = \alpha b \left(1 -\frac{b}{N}\right) -  I(b) b - D(u)s\\
\dot s = \alpha s \left(1 -\frac{b}{N}\right) - I(b) s -  D(u) s - M(u) s
\end{cases}
\end{equation}
By \eqref{eq:originalModel}, $\dot r = \alpha r \left(1 -\frac{b}{N}\right) - I(b) r + M(u)  s$.
Both susceptible and resistant bacteria have a net growth rate $\alpha>0$, affected by a logistic term with carrying capacity $N>0$ to account for resource limitations \cite{fujikawa2003new}.
Term $I(b)$ captures the host immune response effect, modelled as \cite{Dagata08,Pugliese2008,Gomes2013}
\begin{equation}\label{eq:inmuneresponse}
\begin{array}{r}
    I(b) = \beta \frac{K}{K+b},
\end{array}
\end{equation}
where $\beta>\alpha$ is the maximum killing rate and $K<N$ is the total bacterial load yielding half-maximal rate. Since $\beta > \alpha$, without treatment, the immune system clears the infection ($\dot b<0$) if $b$ is sufficiently small.
Recall that $u_M$ is the MIC for sensitive bacteria, which undergo mutations conferring resistance to the drug at rate
\begin{equation}\label{eq:mu}
M(u) = \mu H(u-u_M)
\end{equation}
where $\mu>0$ is a constant mutation rate \cite{Drake1998} and $H$ is a Heaviside step function.
The drug-mediated death rate $D(u)$, with $u\in \mathcal{U}$, only affects sensitive bacteria and is such that
\begin{equation}\label{eq:suscepCondition}
D(u) \ge \alpha \quad \forall u \ge u_M.
\end{equation}
Hence, if $s>0$, $\dot s < 0$ when $u \geq u_M$ (which is true $\forall u \in \mathcal{V}$).

Fig.~\ref{fig:s_r_udynamics}c shows an example of time evolution of $s(t)$, $r(t)$, $b(t)$ according to system~\eqref{eq:originalModel}, with $u \in \mathcal{V}$. 

\begin{remark}\label{Rem:DiscontofRHS}
Since the evolution of $u$ is discontinuous, system \eqref{eq:originalModel} has a discontinuous right-hand side in $t\in \mathbb{R}$. Still, due to \eqref{eq:udynamic}, for any initial time $t_0$ and any initial condition $(b',s')$ of \eqref{eq:originalModel}, there exists $\mathcal{I}_\varepsilon=[t_0,t_0+\varepsilon)$ with $\varepsilon>0$ on which $u$ is continuous. On $\mathcal{I}_\varepsilon$ the right-hand side of \eqref{eq:originalModel} is continuous in $t$ and Lipschitz continuous in $(b,s)$.
By the Picard-Lindelöf theorem \cite{teschl2012ordinary}, unique solutions of \eqref{eq:originalModel} exist on $\mathcal{I}_\varepsilon$, and are extended in time to $+\infty$ as follows: given $\tau_j$, one has a unique solution of \eqref{eq:originalModel} with initial condition $(s(\tau_j),b(\tau_j))$ on $[\tau_j,\tau_{j+1})$ and defines 
$(b(\tau_{j+1}),s(\tau_{j+1}))=\lim_{t\to \tau_{j+1}^-}(b(t),s(t))$. These solutions are continuous and piecewise continuously-differentiable on $[0,\infty)$.
\end{remark}

\begin{figure}[tb]
    \centering  \includegraphics[width=0.4\textwidth]{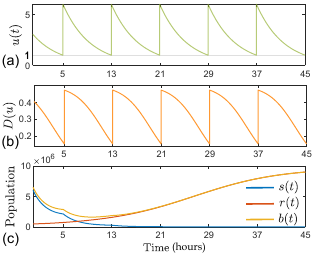}
    \caption{Single-antibiotic dynamics. Time evolution of: \textbf{(a)} $u \in \mathcal{V}$ with $u_M = 1$, $u_{\max} = 6$, $w=\ell \, u_M$, $\ell=5$, $\delta =  0.224$, $\tau_0=5 h$, $\Delta{\tau}=8h \approx \delta^{-1} \ln(\ell+1)$; \textbf{(b)} $D(u)$ as in \eqref{eq:HillFunction} with $E_{\max} = 0.5$, $EC_{50} = 1.5$, $k = 2$; \textbf{(c)} $b(t)$, $s(t)$ as in \eqref{eq:originalModel} and $r(t)=b(t)-s(t)$ with $N=10^7$, $\alpha = 0.15$, $\beta = 50$, $K = 1000$, $\mu = 10^{-5}$, for the given $D(u)$, with $b(0) = 6.5\cdot 10^6$, $s(0) = 6\cdot 10^6$, $r(0) = 0.5\cdot 10^6$.}
    \label{fig:s_r_udynamics}
\vspace{-5mm}
\end{figure}

\section{Failure and Success in Infection Clearance}\label{sec:main}

For suitable parameters, system~\eqref{eq:originalModel} admits not only an infection-free equilibrium, but also pathogenic equilibria with a non-zero bacterial load. We define as \textit{failure} the convergence of the total bacterial population to a pathogenic equilibrium, and we identify initial conditions from which failure is avoided, first by the immune response alone, in the absence of antibiotics ($u \equiv 0$), and then by a therapy $u \in \mathcal{V}$.

System~\eqref{eq:originalModel} evolves in the domain
\begin{equation}\label{eq:domain}
\mathbb{D} = \{(b,s)\in\mathbb{R}^2 \colon 0 \leq b\leq N \mbox{ and } 0 \leq s\le b\},
\end{equation}
which is invariant. In fact, $(0,0)$ is an equilibrium; when $s=0$ and $0 < b \leq N$, $\dot s =0$; when $b=N$ and $0 \leq s \leq N$, $\dot b < 0$; when $0 < s=b \leq N$, in light of \eqref{eq:suscepCondition}, $\dot s<\dot b<0$, and hence the vector field points strictly inside $\mathbb{D}$.

\subsection{Failure and success of the immune response}\label{sec:immunefailure}

Consider system~\eqref{eq:originalModel} with $u \equiv 0$; since $\frac{d}{dt}{(\frac{s}{b})}=\frac{\dot s b - \dot b s}{b^2}=0$, $s(t)=\frac{s(0)}{b(0)}b(t)$, and it is enough to focus on the dynamics of $b$. Define the parameters $q = \frac{N}{K}>1$ (since $N>K$) and $n = \frac{\beta}{\alpha}>1$ (since $\beta > \alpha$). Then,
\begin{equation}\label{eq:btotal}
\begin{array}{r}
\dot b =  \alpha b \left(1 -\frac{b}{N}\right) -  \frac{n\alpha b}{1 + \frac{q}{N} b}=\alpha b \left(1 -\frac{b}{N} - \frac{n}{1 + \frac{q}{N} b}\right).
\end{array}
\end{equation}
\begin{proposition}\label{prop:equilibria}
Let $\mathcal{D}=(q-1)^2-4q(n-1)$. System \eqref{eq:btotal} always admits the asymptotically stable equilibrium $b_0^*=0$, which is unique if $\mathcal{D}<0$. If $\mathcal{D}=0$, the system admits one nonzero equilibrium $b^* < N$, which is unstable. If $\mathcal{D}>0$, the system admits two nonzero equilibria $0<b^*_1<b^*_2<N$, with $b_1^*$ unstable and $b_2^*$ asymptotically stable.
\end{proposition}
\begin{proof}
From the equilibrium condition $\dot b = 0$, $b_0^* = 0$ is always an equilibrium, and there can be at most other two equilibria given by the positive real solutions to the equation $-\frac{1}{K}b^2 + (q-1)b-N(n-1)=0$, with discriminant $\mathcal{D}$.

If $\mathcal{D}<0$, $b_0^*=0$ is the only equilibrium and it is globally asymptotically stable, since $\dot b < 0$ for all $b>0$.

If $\mathcal{D}=0$, there is one nonzero equilibrium $b^* = \frac{K(q-1)}{2}$. Equilibrium $b_0^*=0$ is asymptotically stable and $b^*<N$ is unstable, since $\dot b <0$ for both $0 < b < b^*$ and $b > b^*$.

If $\mathcal{D}>0$, there are two nonzero equilibria
\begin{equation}\label{eq:b12}
\begin{array}{r}
    b^*_{1,2} = \frac{K}{2} \left[(q-1)\pm \sqrt{(q-1)^2-4q(n-1)}\right],
\end{array}
\end{equation}
with $0<b^*_1<b^*_2<N$. Equilibria $b_0^* = 0$ and $b^*_2$ are asymptotically stable, while $b^*_1$ is unstable, because $\dot b <0$ for $0 < b < b_1^*$, while $\dot b >0$ for $b_1^* < b < b_2^*$ and $\dot b < 0$ for $b > b_2^*$, as is also visualised in Fig.~\ref{fig:dotbvsb}a.
\end{proof}

We are interested in the existence of a stable pathogenic equilibrium. Thus, henceforth we always assume $\mathcal{D}>0$. Equivalently, we have the following standing assumption.

\begin{assumption}\label{ass:onsetinfection}
Parameters $q = \frac{N}{K}>1$ and $n = \frac{\beta}{\alpha}>1$ are such that $n < 1+ \frac{(q-1)^2}{4q}$.
\end{assumption}

Then, the failure set of all initial conditions from which the solution of \eqref{eq:btotal} does not converge to the infection-free equilibrium can be fully characterised.

\begin{theorem}\label{propos:Thresholdforbeta}
Consider the solution $b(t)$ of system \eqref{eq:btotal} under Assumption~\ref{ass:onsetinfection} emanating from the initial condition $b_0>0$.
We have $\lim_{t\to\infty} b(t) \neq 0$ if and only if $b_0 \geq b_1^*$, with $b_1^*=\frac{K}{2}[(q-1)-\sqrt{(q-1)^2-4q(n-1)}]$.
\end{theorem}
\begin{proof}
It follows from the proof of Proposition~\ref{prop:equilibria}. For $0<b<b^*_1$ (green segment in Fig.~\ref{fig:dotbvsb}a), $\dot b < 0$, and hence $\lim_{t\to\infty}b(t)=b_0^*=0$.
For $b=b^*_i$, with $i=1,2$, $\dot b = 0$, hence $\lim_{t\to\infty} b(t)=b^*_i \neq 0$.
For $b^*_1< b <b^*_2$ (yellow segment in Fig.~\ref{fig:dotbvsb}a), $\dot b > 0$ and, for $b > b^*_2$ (red segment in Fig.~\ref{fig:dotbvsb}a), $\dot b < 0$; so, for all $b>b^*_1$, $\lim_{t\to\infty}b(t)=b_2^* \neq 0$.
\end{proof}

Therefore, if the total bacterial load is $b_0 < b^*_1$, the immune system alone can clear the infection and no therapy is needed. Conversely, if $b_0\geq b^*_1$, the immune system alone cannot cope with the infection and antibiotics are required.

\begin{figure}[tb]
    \centering
    \includegraphics[width=0.48\textwidth]{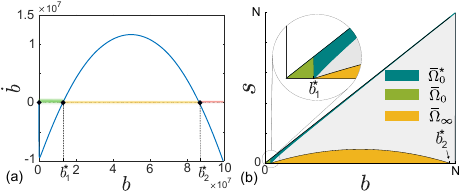}
    \caption{\textbf{(a)} Without antibiotics, $\dot b$ in \eqref{eq:btotal} as a function of $b$ under Assumption~\ref{ass:onsetinfection}, with $N=10^8$, $K = 10^4$, $\alpha = 0.88$ and $\beta = 100$. The system equilibria are $b^*_0=0$ (stable), $b^*_1$ (unstable), $b^*_2$ (stable). \textbf{(b)} Invariant domain $\mathbb{D}$ in \eqref{eq:domain} where we consider the solutions of system \eqref{eq:originalModel}, shaded in gray. Success subset $\bar\Omega_0$ in \eqref{eq:safesubset}, green; set $\bar\Omega_0^* = \Omega_0^* \setminus \bar \Omega_0$ in \eqref{eq:expandedsafesubset}, teal; failure subset $\bar\Omega_{\infty}$ in \eqref{eq:failuresubset}, yellow.}
    \label{fig:dotbvsb}
\vspace{-5mm}
\end{figure}

\subsection{Failure and success of single-antibiotic therapies}\label{sec:antibiotic}

Consider system~\eqref{eq:originalModel} evolving in $\mathbb{D}$ in \eqref{eq:domain} when $u \in \mathcal{V}$ is administered, and let $\phi(t;b_0,s_0,u) = (b(t),s(t))$ be the solution emanating from initial condition $(b_0,s_0) \in \mathbb{D}$.

For all $u \in \mathcal{V}$, $u(t) \geq u_M$ for all $t \geq 0$. This guarantees that $M(u)=\mu>0$ by \eqref{eq:mu} and the non-constant function $D(u) \ge \alpha>0$ by \eqref{eq:suscepCondition}. Therefore, $s(t)>0$ implies $\dot s(t) < 0$ $\forall$ $t$, and all the equilibria $(\bar b, \bar s)$ must be such that $\bar s=0$. Hence, the equilibrium values for $b$ are the ones we derived for \eqref{eq:btotal}, and the system~\eqref{eq:originalModel} has three equilibria: $(0,0)$, $(b^*_1,0)$ and $(b^*_2,0)$, with $b^*_1$ and $b^*_2$ as in \eqref{eq:b12}; see Fig.~\ref{fig:dotbvsb}b.

Given $(b_0,s_0)\in\mathbb{D}$, success of the single-antibiotic therapy occurs if there exists $u \in \mathcal{V}$ such that $\lim_{t \to \infty} \phi(t;b_0,s_0,u)= (0,0)$. Failure occurs otherwise.
Then, we can define a success set $\Omega_0\subset \mathbb{D}$ of initial conditions from which $\phi(t;b_0,s_0,u)$ converges to the infection-free equilibrium, and a complementary failure set $\Omega_{\infty}\subset \mathbb{D}$.

\begin{definition}\textbf{(Success/Failure sets.)}
The success set is
$\Omega_0 \doteq \{(b_0,s_0)\in\mathbb{D} \colon \lim_{t \to \infty} \phi(t;b_0,s_0,u)=(0,0) \mbox{ for some }u \in \mathcal{V}\}$.
The failure set is $\Omega_{\infty} \doteq \mathbb{D}\setminus \Omega_0$.
\end{definition}

We define a success subset $\bar\Omega_0$ and a failure subset $\bar\Omega_{\infty}$.

\begin{definition}\label{def:FailureSafeSubset}\textbf{(Success/Failure subsets.)} The set
\begin{equation}\label{eq:safesubset}
\bar\Omega_0 = \{(b,s)\in\mathbb{D} \colon b<b^*_1\}
\end{equation}
is denoted as success subset, while the failure subset is
\begin{equation}\label{eq:failuresubset}
\bar\Omega_{\infty} = \{(b,s)\in\mathbb{D} \colon \dot b > 0 \text{ for }u \equiv u_{\max}\}.
\end{equation}
\end{definition}

Fig.~\ref{fig:dotbvsb}b shows the set $\bar \Omega_0$ in green and the set $\bar \Omega_\infty$ in yellow, for a parameter choice satisfying Assumption~\ref{ass:onsetinfection}.

We study convergence of the system trajectories and prove that $\bar\Omega_0$ and $\bar\Omega_{\infty}$ are subsets of $\Omega_0$ and $\Omega_{\infty}$, respectively.

By Theorem~\ref{propos:Thresholdforbeta}, the trajectories emanating from initial conditions in $\bar\Omega_0$ converge to zero for $u \equiv 0$.

A characterisation of $\bar\Omega_{\infty}$ is given by the function
\begin{equation}
\begin{array}{r}
\psi(b) \doteq \frac{\alpha b(1 - \frac{b}{N}) - I(b)b}{D(u_{\max})},    
\end{array}
\end{equation}
for which we denote
\begin{equation*}
\begin{array}{lll}
&\operatorname{Gr}(\psi)\doteq \left\{(b,\psi(b))\colon \ b\in [b_1^*,b_2^*] \right\}.
\end{array}
\end{equation*}
$\operatorname{Gr}(\psi)$ is the set of points for which $\dot b =0$ for $u \equiv u_{\max}$, and $\bar{\Omega}_{\infty}$ consists of all points such that $s < \psi(b)$, i.e.,
\begin{equation}\label{eq:expressionbarOminf}  
\bar \Omega_\infty = \left\{(b,s)\in \mathbb{D}\colon \ b_1^*\leq b\leq b_2^*,\  0\leq s<\psi(b) \right\}.
\end{equation}
By the proof of Proposition~\ref{prop:equilibria}, $\psi(b)>0$ for $b_1^* < b < b_2^*$. 

\begin{lemma}\label{Lem:psiprop}
There exists $\varepsilon_*>0$ such that
the segment $\operatorname{Lin}(\varepsilon_*)\doteq\left\{(b,\varepsilon_*)\in \mathbb{D} \colon \ b\in [b_1^*,b_2^*] \right\}$ intersects $\operatorname{Gr}(\psi)$ at exactly two points. Then, for any $d_1\in(b_1^*,b_2^*)$ and $0<\varepsilon < \min(\varepsilon_*,\psi(d_1))$, the equation $\psi(b)=\varepsilon$ has a unique solution $d_2$ in $(d_1,b_2^*)$ and the rectangle $\mathcal{M}(d_1,\varepsilon) = \left([d_1,d_2]\times [0,\varepsilon]\right)\setminus \left\{(d_2,\varepsilon) \right\}$ is contained in the set $\bar \Omega_\infty$.
\end{lemma}
\begin{proof}
The function $\psi$ is smooth and positive on the interval $(b_1^*,b_2^*)$ with $\psi(b_1^*)=\psi(b_2^*)=0$, $\psi'(b_1^*)>0$ and $\psi'(b_2^*)<0$. Choose a small $\delta>0$ such that on $[b_1^*,b_1^*+2\delta]$ ($[b_2^*-2\delta,b_2^*]$) the function $\psi$ is increasing (decreasing). Since  
\begin{equation*}
\psi'(b) = 0 \Leftrightarrow \alpha N(K+b)^2-2\alpha b(K+b)^2-\beta K^2N = 0,
\end{equation*}
$\psi$ has at most one local minimum in $(b_1^*,b_2^*)$, on which $\psi$ is positive. Choose $0<\varepsilon_*<\min \left(\psi(b_1^*+\delta),\psi(b_2^*-\delta) \right)$, so that it is smaller than the value of $\psi$ at the local minimum (if any) in $(b_1^*,b_2^*)$. By continuity, $\operatorname{Lin}(\varepsilon_*)$ intersects $\operatorname{Gr}(\psi)$ once in each of the intervals $[b_1^*,b_1^*+\delta]$ and $[b_2^*-\delta,b_2^*]$.
If there were another intersection, we would obtain, by smoothness of $\psi$, an additional local minimum in $(b_1^*,b_2^*)$, with value smaller or equal to $\varepsilon_*$, which is a contradiction. 
The remaining claims follow trivially.
\end{proof}

Lemma~\ref{Lem:psiprop} is now used to guarantee convergence of the system trajectories to the equilibrium set.

\begin{proposition}\label{prop:convergence}
For any initial condition $(s_0,b_0)\in \mathbb{D}$ and any $u \in \mathcal{V}$, the trajectory $\phi(t;b_0,s_0,u)$ converges to one of the system equilibria: $(0,0)$, $(0,b_1^*)$ and $(0,b_2^*)$.
\end{proposition}
\begin{proof}
Let $\phi(t;b_0,s_0,u)=(b(t),s(t))$.  From \eqref{eq:mu} and \eqref{eq:suscepCondition}, $s>0$ implies $\dot s < -\mu s < 0$, whence $\lim_{t \to \infty} s(t)=0$. 

If $\lim_{t \to \infty} b(t) = b_1^*$, then $\phi(t;b_0,s_0,u)$ converges to the equilibrium $(b^*_1,0)$. Otherwise, there must exist an arbitrarily small constant $\eta >0$ and an increasing unbounded sequence of times $l_k$ such that $|b(l_k)-b_1^*| \geq 2\eta$ $\forall k$.

\underline{Case 1}: There exists a $k$ such that $b(l_k) \leq b_1^*-2\eta$, i.e., $b(l_k) \in \{(b,s) \in \mathbb{D} \colon 0 \leq b \leq b_1^*-2\eta\} \subset \bar \Omega_0$. For all $(b,s) \in \bar \Omega_0$, we have $\gamma(b) \doteq \alpha b (1 - b/N) -  I(b) b<0$, as shown in Section~\ref{sec:immunefailure}; thus $\dot{b}=\gamma(b) - D(u)s < 0$. This implies that $\lim_{t \to \infty} b(t) = 0$, and therefore $\phi(t;b_0,s_0,u)$ converges to the equilibrium $(0,0)$.

\underline{Case 2}: For all $k$, $b(l_k) \geq  b_1^*+2\eta$. Choose $\varepsilon_*>0$
as in Lemma~\ref{Lem:psiprop}, $\varpi>0$ arbitrarily small and $h < \min(\psi(b_2^*-\varpi)/2, \psi(b_1^*+\eta),\varepsilon_*)$, and define the closed rectangles 
\begin{equation*}
\begin{array}{lll}
&\hspace{-2mm}\mathcal{Y}_1 = \{(b,s) \in \mathbb{D} \colon \ b_1^*+2\eta \leq b \leq b_2^*-\varpi, \, 0 \leq s \leq h\}\\
&\hspace{-2mm}\mathcal{Y}_2 = \{(b,s) \in \mathbb{D} \colon \ b_2^*-\varpi \leq b \leq b_2^*+\varpi, \, 0 \leq s \leq  h\}\\
&\hspace{-2mm}\mathcal{Y}_3 = \{(b,s) \in \mathbb{D} \colon \ b_2^*+\varpi \leq b \leq N, \, 0 \leq s \leq h\}.
\end{array}
\end{equation*}
By Lemma~\ref{Lem:psiprop}, $\mathcal{Y}_1\subset  \mathcal{M}(b_1^*+\eta,h) \subset \bar \Omega_{\infty}$. Indeed, there exists a unique $\hat b>b_1^*+\eta$ such that $\operatorname{Lin}(h)$ intersects $\operatorname{Gr}(\psi)$; since $\psi(b_2^*-\varpi)>h$, it must be $b_2^*-\varpi < \hat b < b_2^*$.
Since $\mathcal{Y}_1 \subset \bar \Omega_{\infty}$ is closed, there exists $\xi>0$ such that, for all $(b,s)\in \mathcal{Y}_1$, we have $\chi(b,s) \doteq \alpha b \left(1 -\frac{b}{N}\right) -  I(b) b - D(u_{\max})s\geq \xi$. In particular, for any $u \in \mathcal{V}$, since $D(u(t)) \leq D(u_{\max})$ $\forall t$, it is $\dot b = \alpha b \left(1 -\frac{b}{N}\right) -  I(b) b - D(u)s \geq \chi(b,s) \geq \xi$ for all $(b,s)\in \mathcal{Y}_1$. Analogously, one can show that there exists $\zeta>0$ such that for $(b,s)\in \mathcal{Y}_3$, it holds that $\dot{b}\leq -\zeta<0$.

Since $\lim_{t\to \infty}s(t)=0$, we can choose $k$ such that $(b(l_k),s(l_k)) \in \mathcal{Y}_1 \cup \mathcal{Y}_2 \cup \mathcal{Y}_3$. If $(b(l_k),s(l_k)) \in \mathcal{Y}_2$, the solution remains in $\mathcal{Y}_2$ for all future times, as on both vertical sides of $\mathcal{Y}_2$ the derivative $\dot{b}$ points strictly inside $\mathcal{Y}_2$, whereas on the upper side the derivative $\dot{s}$ points inside $\mathcal{Y}_2$. If $(b(l_k),s(l_k)) \in \mathcal{Y}_1$, then, since $\dot s <0$ $\forall$ $s>0$ and $\dot b \geq \xi$, the trajectory leaves $\mathcal{Y}_1$ through the side $b=b_2^*-\varpi$ in finite time, meaning that it enters $\mathcal{Y}_2$ and remains there for all future times. If $(b(l_k),s(l_k)) \in \mathcal{Y}_3$, then, since $\dot s <0$ $\forall$ $s>0$ and $\dot b \leq -\zeta$, the trajectory leaves $\mathcal{Y}_3$ through the side $b=b_2^*+\varpi$ in finite time, meaning that it enters $\mathcal{Y}_2$ and remains there for all future times. This implies that $b_2^*-\varpi \leq \liminf_{t\to\infty} b(t)\leq \limsup_{t\to\infty} b(t) \leq b_2^*+\varpi$. Since $\varpi>0$ is arbitrarily small, we conclude that $\lim_{t \to \infty} b(t) = b_2^*$. So, $\phi(t;b_0,s_0,u)$ converges to the equilibrium $(b_2^*,0)$.
\end{proof}

Both $\Omega_{\infty}$ and $\Omega_0$ are nonempty, as we show next.
\begin{proposition}\label{prop:failuresafe}
\textbf{(Success and Failure subsets.)}
Given $\bar\Omega_0$  in \eqref{eq:safesubset} and $\bar\Omega_{\infty}$ in \eqref{eq:failuresubset}, we have $\bar\Omega_0 \subseteq\Omega_0$ and $\bar\Omega_{\infty}\subseteq\Omega_{\infty}$.
\end{proposition}
\begin{proof}
The proof of Proposition~\ref{prop:convergence} shows that trajectories emanating from $\bar\Omega_0$ converge to zero, whence $\bar\Omega_0 \subseteq\Omega_0$.

For any $(b_0,s_0)\in \bar{\Omega}_{\infty}$, we have $s_0<\psi(b_0)$ by \eqref{eq:expressionbarOminf}. Therefore, by continuity of $\psi$, $(b_0,s_0)\in \mathcal{O}$ for some square $\mathcal{O}\doteq [b_0-\varpi,b_0+\varpi]\times [s_0-\varpi,s_0+\varpi],\ \varpi>0$ such that  $\mathcal{O} \subset \bar \Omega_\infty$. Consider the rectangles 
\begin{equation*}
\begin{array}{lll}
&\mathcal{Y}_4 = \left\{(b,s)\in \mathbb{D} \colon \ \ b_0-\varpi \leq b\leq b_0,\ 0\leq s\leq s_0  \right\},\\
& \mathcal{Y}_5 = \left\{(b,s)\in \mathbb{D} \colon \ b_0-\varpi\leq b\leq N,\ 0\leq s\leq s_0  \right\},
\end{array}
\end{equation*}
with $\mathcal{Y}_4 \subset \bar \Omega_{\infty}$ closed. Arguments similar to \underline{Case 2} in the proof of Proposition \ref{prop:convergence} show the existence of $\xi>0$ such that,  $\forall u \in \mathcal{V}$ and all $(b,s)\in \mathcal{Y}_4$, $\dot{b}\geq \xi$. Since $\dot{s}<0$ if $s>0$, $\phi(t;b_0,s_0,u)$ remains in $\mathcal{Y}_5$ $\forall t\geq 0$ and,
by Proposition \ref{prop:convergence}, $\lim_{t \to \infty} \phi(t;b_0,s_0,u) = (0,b_2^*)$, showing $\bar\Omega_{\infty}\subseteq\Omega_{\infty}$.
\end{proof}

We can now prove the stability properties of the equilibria.
\begin{theorem}\label{cor:stability}
For all $u \in \mathcal{V}$, the equilibrium $(b_1^*,0)$ is unstable, while $(0,0)$ and $(b_2^*,0)$ are locally asymptotically stable. Moreover, there exist invariant neighbourhoods $\mathcal{Q}_0$ of $(0,0)$ and $\mathcal{Q}_2$ of $(b_2^*,0)$ such that $\mathcal{Q}_0$ (resp. $\mathcal{Q}_2$) lies in the basin of attraction of $(0,0)$ (resp. $(b_2^*,0)$) for all $u \in \mathcal{V}$.
\end{theorem}
\begin{proof}
As shown in the proof of Proposition~\ref{prop:failuresafe}, every solution emanating from $\mathcal{Q}_0=\bar\Omega_0$ converges to $(0,0)$, irrespective of $u$. Similarly, in the proof of Proposition~\ref{prop:convergence}, we have presented a closed rectangle $\mathcal{R}=\bigcup_{i=1}^3\mathcal{Y}_i$ that contains $(0,b_2^*)$ such that any solution emanating from $\mathcal{R}$ converges to $(0,b_2^*)$, irrespective of $u$. Thus, one can take $\mathcal{Q}_2$ to be the (relative) interior of $\mathcal{R}$ in $\mathbb{D}$.
The instability of $(b_1^*,0)$ follows from Proposition~\ref{prop:equilibria}; also, note that $(b_1^*,0)\in \partial \bar{\Omega}_0$.
\end{proof}

Hence, under Assumption~\ref{ass:onsetinfection}, there exists an asymptotically stable pathogenic equilibrium with a nonempty basin of attraction, corresponding to bacterial loads that no therapy in $\mathcal{V}$ can clear.
We are interested in approximating the basin of attraction of the infection-free equilibrium.

For $b \geq b_1^*$, consider the bounded set 
\begin{equation}\label{eq:set}
\hspace{-2mm}\mathcal{S}(b) = \{ s\in[0,b] \colon \forall u\in\mathcal{V}, \, \lim_{t \to \infty} \phi(t;b,s,u) \neq (0,0)\}
\end{equation}
and define the function
\begin{equation}\label{eq:boundary}
g(b)\doteq \sup \mathcal{S}(b), \quad b \in [b_1^*, N].    
\end{equation}
\begin{proposition}\label{prop:gb}
Function $g(b)$ is well defined, meaning that for any $b\in [b_1^*,N]$, $\mathcal{S}(b)\neq \emptyset$. Furthermore, $g(b)>0$ for all $b\in ( b_1^*,N ]$, while $g(b_1^*)=0$.
\end{proposition}
\begin{proof}
In the proofs of Propositions \ref{prop:convergence} and \ref{prop:failuresafe}, we have shown that $\bar{\Omega}_{\infty}\cup \left( \bigcup_{i=1}^3\mathcal{Y}_i\right)\subset \Omega_{\infty}$. Hence, for $b\in [b_1^*,N]$ it is $g(b)\geq \max \left(\psi(b),h \right)\geq 0$, and $g(b)>0$ for $b\in (b_1^*,N]$.

Consider any initial condition $(b_1^*,s)$ with $s>0$; at $b_1^*$, $\alpha b_1^* (1 - b_1^*/N) -  I(b_1^*) b_1^*=0$ as shown in Section~\ref{sec:immunefailure}, and hence $\dot{b}_1=\alpha b_1^* (1 - b_1^*/N) -  I(b_1^*) b_1^* - D(u)s < 0$ at the initial time, irrespective of $u$; in particular, $\phi(t;b_1^*,s,u)$ enters $\bar \Omega_0$ and converges to $(0,0)$; hence, $\mathcal{S}(b_1^*)=\{0\}$.
\end{proof}

Function $g(b)$ provides an inner approximation of $\Omega_0$.

\begin{theorem}\label{thm:expandedsafesubset}
\textbf{{(Expanded success subset.)}}
Consider
\begin{equation}\label{eq:expandedsafesubset}
\bar \Omega_0^* \doteq \{(b,s) \colon b_1^* \leq b \leq N, \, g(b) < s \leq b\}.
\end{equation}
Then, $\Omega_0^* = \bar \Omega_0 \cup \bar \Omega_0^*$
is included in the success set: $\Omega_0^* \subseteq \Omega_0$.
\end{theorem}
\begin{proof}
Proposition~\ref{prop:failuresafe} ensures that $\bar \Omega_0 \subset \Omega_0$. For any initial condition $(b,s)$ with $b_1^* \leq b \leq N$ and $g(b) < s \leq b$, by the definition of $g(b)$ in \eqref{eq:set}-\eqref{eq:boundary}, $\phi(t;b,s,u)$ converges to $(0,0)$ for some $u \in \mathcal{V}$; hence, $\bar \Omega_0^* \subseteq \Omega_0$.
\end{proof}

An example of set $\bar \Omega_0^*$ is shown in teal in Fig.~\ref{fig:dotbvsb}b.

\subsection{Computing an approximation of the success set}\label{sec:numeric}
A numerical approximation of $\bar \Omega_0^*$ is computed as follows.
Take $m$ therapies $u_1,\dots,u_m \in \mathcal{V}$, $H \in \mathbb{N}$ and a grid of points $b_j = b_1^* + j \frac{N-b_1^*}{H}$, $j=1,\dots,H$. For each $u_i$, $i \in \{1,\dots,m\}$, and each $b_j$, $j\in \{1,\dots,H\}$, we apply the bisection method in $s$ over the interval $[0,b_j]$. Note that $(b_j,0) \in \Omega_\infty$, whereas $(b_j,b_j)$ can be either in $\Omega_0$ or in $\Omega_\infty$. If simulation of $\phi(t;b_j,b_j,u_i)$ does not show convergence to $(0,0)$, the process stops. Otherwise, we consider the point $(b_j,\frac{b_j}{2})$ and simulate $\phi(t;b_j,\frac{b_j}{2},u_i)$; if it converges to $(0,0)$, we repeat the bisection on the interval $s \in [0,\frac{b_j}{2}]$, otherwise on the interval $s \in [\frac{b_j}{2},b_j]$. The process continues until we obtain an interval $[\sigma_j^L, \sigma_j^H]$ for which $\sigma_j^H-\sigma_j^L < \textit{tol}$ for a prescribed tolerance $\textit{tol}$. Then, we build an approximation 
\begin{equation}\label{eq:hatOmega_i}
\hat \Omega_i = \{(b,s) \colon b_1^* \leq b \leq N, \, g_i(b) < s \leq b\},
\end{equation}
where $g_i(b) = \sigma_j^H$ on the interval $(b_{j-1},b_j]$, $j\in \{1,\dots,H\}$, and $b_0=b_1^*$.
The resulting approximation of $\bar \Omega_0^*$ is $\bigcup_{i=1}^m \hat \Omega_i$.
We simulate the system with a standard $4$-th order Runge-Kutta method and the parameter values
$\alpha=0.25$,
$N=10^7$,
$\beta=20$,
$K=1.4\cdot 10^3$,
$E_{\max}=0.6$,
$EC_{50}=1.3$,
$k=2$,
$\mu=10^{-5}$,
$\delta=0.22$,
$u_M=1.5$,
$u_{\max}=16 u_M$ \cite{regoes2004pharmacodynamic,smith2011mathematical,Drake1998},
for which $b_1^*=1.14\cdot 10^5$ and $b_2^*=9.88\cdot 10^6$.

Fig.~\ref{fig:upperboundg} shows $\bar\Omega_0$ (green) and $\bar\Omega_\infty$ (yellow) computed for $b \in [0,10^6]$, and simulation results yielding approximations of $\bar\Omega_0^*$, with $b \in [b_1^*,10^6]$, for $\textit{tol} = 50$.

In Fig.~\ref{fig:upperboundg}a we computed, with $H=50$, $\hat\Omega_i$ for $25$ therapies in $\mathcal{V}$ with different combinations of $w$ and $\tau_0$: $w = 2 u_M$ and $\tau_0 \in \{0, 2, 3, 4, 4.5\}$; $w = 5 u_M$ and $\tau_0 \in \{0, 2, 4, 6, 8\}$; $w = 7 u_M$ and $\tau_0 \in \{0, 3, 4, 7, 9\}$; $w = 10 u_M$ and $\tau_0 \in \{0, 3, 5, 7, 10\}$; $w = 15 u_M$ and $\tau_0 = \{0, 5, 6, 7, 8\}$.
Therapies with the same dose $w$ but different $\tau_0$ yield the same set, so we can only observe the effect of different $w$ in Fig.~\ref{fig:upperboundg}b; in all simulated cases, the set for a larger dose contains the set for a smaller dose.
Fig.~\ref{fig:upperboundg}c shows the values $\sigma_j^H$ for $w = 15 u_M$ computed on different grids with $H = 10$, $H = 50$, $H = 100$ in the interval $[b_1^*, 10^6]$.
Finally, Fig.~\ref{fig:upperboundg}d shows the approximation of $\bar\Omega_0^*$ by $\bigcup_{i=1}^m \hat \Omega_i$, determined by the points $\sigma_j^H$ for $w = 15 u_M$.

\begin{figure}[tb]
    \centering    \includegraphics[width=0.5\textwidth]{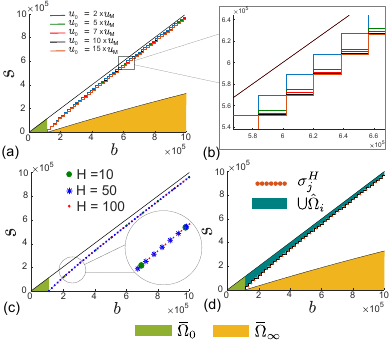}
    \caption{\textbf{(a)} $\hat\Omega_i$ computed for different therapies in $\mathcal{V}$ with $w = \ell u_M$, $\ell\in\{2,5,7,10,15\}$, each for five different values of $\tau_0$. \textbf{(b)} For a given $w$, the set is unchanged by varying $\tau_0$. The set for $w = 15 u_M$ contains all the others. \textbf{(c)} $\sigma_j^H$ for $w = 15 u_M$ computed on different grids with $H = 10$ (green), $H = 50$ (blue), $H = 100$ (red). \textbf{(d)} Approximation of $\bar\Omega_0^*$ by $\bigcup_{i=1}^m \hat \Omega_i$ determined by the points $\sigma^H_j$ for $w = 15 u_M$.}
    \label{fig:upperboundg}
\vspace{-5mm}
\end{figure}

Our exhaustive numerical simulation campaign in Section~\ref{Subsec:Conservaiveness_Succ_Set} suggests that $\Omega_0^*$ in Theorem~\ref{thm:expandedsafesubset} is the actual success set, i.e., function $g(b)$ provides an accurate separation between $\Omega_0$ and $\Omega_\infty$. A sensitivity analysis assessing how parameter variations affect $g(b)$ is also provided in Section~\ref{Sec:SuppMaterial}.

\section{Optimal Antibiotic Treatment}\label{sec:optimal}

Theorem~\ref{thm:expandedsafesubset} provides us with a set $\Omega_0^*$ of initial conditions from which the system trajectory can be steered to the infection-free equilibrium by \textit{some} therapy in $u \in \mathcal{V}$, which is uniquely characterised by a choice of $\ell$ as in Proposition~\ref{propo:impulsive_condition}. Hence, given an initial condition in $\Omega_0^*$, we seek a suitable antibiotic therapy that clears the infection while minimising the cost of the optimal control problem
\begin{equation}\label{eq:optimalcontrolV}
\min_{\ell \in (0, \ell_{\max}]} \;\; \int_0^T \left(\zeta_1 u(t) + \zeta_3 b(t)\right)dt \; + \zeta_2 b(T),
\end{equation}
over the finite horizon $T$, subject to \eqref{eq:udynamic}--\eqref{eq:mu}, $w_j \equiv w = \ell u_M$, $\Delta \tau = \delta^{-1} \ln(\ell+1)$.
We solve \eqref{eq:optimalcontrolV} over a discrete grid of step $0.01$ for various choices of $\zeta_i$. A thorough discussion of the optimal control framework is provided in Section~\ref{Sec:SuppMaterial}.

When $\zeta_2=1$ and $\zeta_1=\zeta_3=0$, the cost is the total number $b(T)$ of bacteria at the end of the treatment horizon and thus, given any initial condition in $\Omega_0^*$, the solution to the optimal control problem yields clearance of the infection, provided that $T$ is sufficiently large, as shown in the two left panels of Fig.~\ref{fig:OCP}. The solution does not provide clearance when the initial condition is outside $\Omega_0^*$, as shown in Section~\ref{Sec:SuppMaterial}; this further validates the hypothesis that the function $g(b)$ actually separates success and failure sets.

The two right panels of Fig.~\ref{fig:OCP} consider the cost with $\zeta_1=10^5$ and $\zeta_2=\zeta_3=1$, which also aims at minimising drug side effects (by penalising high antibiotic concentrations) and the bacterial load over the entire duration of treatment, leading to a reduction in the resistant bacterial population.

\begin{figure*}
    \centering
        \centering
    \includegraphics[width=0.92\textwidth]{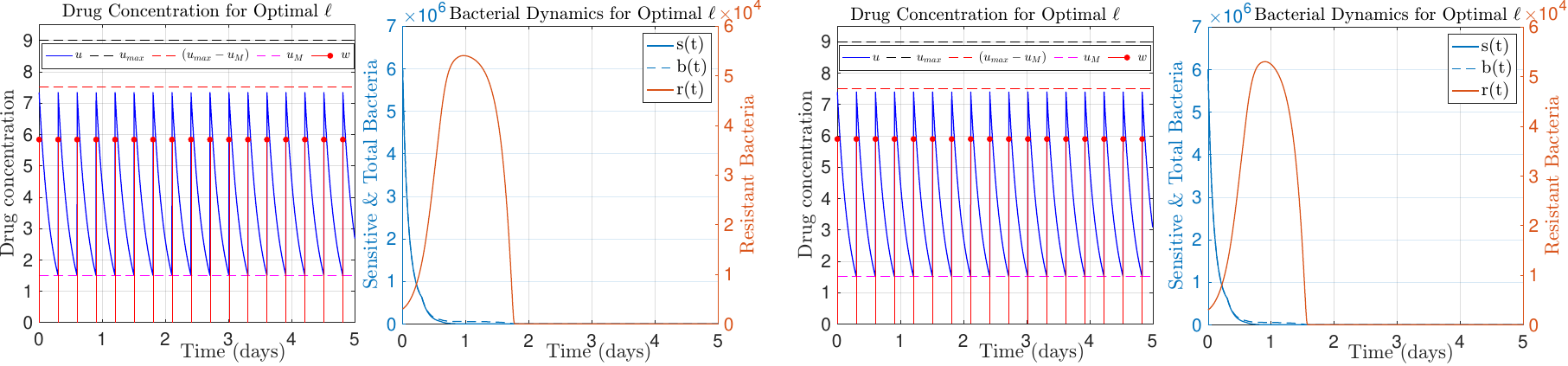}
    \caption{Time evolution of the drug concentration and administered drug doses and time evolution of the state variables $b$, $s$ and of $r=b-s$ obtained by solving the optimal control problem \eqref{eq:optimalcontrolV} with initial condition $b_0= 6.003\cdot 10^6$ and $s_0= 6\cdot 10^6$ in $\Omega_0^*$, with weights $\zeta_2=1$ and $\zeta_1=\zeta_3=0$ (two left panels) and $\zeta_1=10^5$ and $\zeta_2=\zeta_3=1$ (two right panels). The parameters are as in Section~\ref{sec:numeric}, apart from $\beta = 10$. On the left, the minimum $J=0$ is not unique and achieved for all $\ell \geq 3.90$; to reduce the use of antibiotic, we select $\ell = 3.90$. On the right, the unique minimum is obtained for $\ell = 3.93$.}
    \label{fig:OCP}
\vspace{-6mm}
\end{figure*}

A more general optimal control formulation that considers a larger class of admissible controls (including therapies in $\mathcal{V}$) with time-varying dose $w_j$ and allows termination of the treatment before $T$ is discussed in Section~\ref{Sec:SuppMaterial}: the results suggest that lower drug doses administered for shorter periods can clear bacterial infections while reducing side effects.

\section{Concluding Discussion}
We have proposed an in-host model of AMR that captures the evolution of a bacterial population including susceptible and resistant strains in the presence of antibiotics, of which we provide a PK/PD description. To support a sequential decision-making process, we have characterised initial bacterial loads (i) for which the immune response alone cannot clear the infection, and thus drug administration is required, and (ii) for which a single-drug therapy belonging to a clinically relevant class is guaranteed to be successful (success set). For all initial conditions that are not in the success set, the use of multiple different antibiotics is recommended. We have provided rigorous inner approximations of the success set and shown how to compute viable successful policies through optimal control.
Future directions include investigating how the success set expands when multiple drugs are used, fully characterising success and failure sets, and considering more complex models, including stochastic ones \cite{Blanquart2019,Tetteh2020} that overcome the limitations of deterministic models for small population sizes.



\newpage

\section{Supplementary Material}\label{Sec:SuppMaterial}

In the Main Manuscript, we consider the dynamical system
\begin{equation}\label{eq:originalModel}
\begin{cases}
\dot b = \alpha b \left(1 -\frac{b}{N}\right) -  I(b) b - D(u)s\\
\dot s = \alpha s \left(1 -\frac{b}{N}\right) - I(b) s -  D(u) s - M(u) s
\end{cases}
\end{equation}
where 
\begin{equation}\label{eq:inmuneresponse}
I(b) = \beta \frac{K}{K+b}
\end{equation}
and $\beta > \alpha$,
\begin{equation}\label{eq:HillFunction}
D(u) = E_{\max} \frac{(u/EC_{50})^k}{1 + (u/EC_{50})^k}
\end{equation}
and $D(u) \geq \alpha$ for all $u \geq u_M$, while
\begin{equation}\label{eq:mu}
M(u) = \begin{cases}
\mu \mbox{ if } u \ge u_M,\\
0 \mbox{ otherwise.}
\end{cases}
\end{equation}

Denote $u(\tau_j^-) = \lim_{t \to \tau_j^-} u(t)$. We consider the class $\mathcal{V}$ of \textbf{viable periodic therapies} such that $u(t)\in \{\upsilon \colon u_M \leq \upsilon \leq u_{\max}\}$ evolves according to
\begin{equation}\label{eq:udynamic}
\begin{cases}
\dot{u}(t) = -\delta u(t), ~\tau_j \leq t < \tau_{j+1} \\
u(\tau_j) = w_j + u(\tau_j^-), ~j \in \mathbb{N}_0,
\end{cases}
\end{equation}
for $t \geq \tau_0$, with $0 \leq \tau_0 \leq \Delta \tau$, where
$$
w_j \equiv w \in \mathcal{W}=\{\upsilon \colon 0 < \upsilon \le u_{\max}-u_M\}
$$
and $\tau_{j+1}-\tau_j=\Delta \tau$ $\forall j \in \mathbb{N}_0$. If $\tau_0=0$, $u(\tau_0)= w+u_M$; if $\tau_0>0$, $u(t)=u_M e^{-\delta (t-\tau_0)}$ for $0 \leq t < \tau_0$.

All the parameters are listed in Table~\ref{tb:parameters} along with their meaning. 

\begin{table}[htb]
\footnotesize
\centering
\rowcolors{1}{}{white!20}
\begin{tabular}{llll}
\rowcolor{gray!20}
\textsc{\footnotesize{Parameter}} & \textsc{\footnotesize{Description}} & \textsc{\footnotesize{Value}}  \\ \hline
$\alpha$ & Net Growth Rate & 0.25 \\
$N$ & Carrying Capacity & $10^7$ \\
 $\mu$ & Mutation Rate & $10^{-5}$ \\ 
$\beta$ & Immune Response Death Rate  & $20$ \\
$K$ & Immune Response Half-Effect & $1.4\cdot 10^3$ \\
$\delta$ & Antibiotic Decay Rate & $0.22$ \\
$E_{\max}$ & Maximum Antibiotic Death Rate & $0.6$ \\
$EC_{50}$ & Antibiotic Half-Effect  & $1.3$ \\
$k$ & Hill Coefficient & $2$ \\
$u_M$ & Antibiotic Minimum Inhibitory Concentration & $1.5$ \\
$u_{\max}$ & Maximum Antibiotic Concentration& $16 u_M$ \\
\hline
 \end{tabular}
 \caption{List of parameters, along with their description and value \cite{regoes2004pharmacodynamic,smith2011mathematical,Drake1998}.}
 \label{tb:parameters}
 \end{table}

The $b$-$s$ system evolves in the invariant domain
\begin{equation}\label{eq:domain}
\mathbb{D} = \{(b,s)\in\mathbb{R}^2 \colon 0 \leq b\leq N \mbox{ and } 0 \leq s\le b\}.
\end{equation}

In the Main Manuscript, we define and evaluate the success set $\Omega_0$ of initial conditions $(b_0,s_0) \in \mathbb{D}$ from which the system trajectories $\phi(t;b_0,s_0,u)$ converge to the origin, thanks to some single-drug viable periodic therapy $u \in \mathcal{V}$, and the complementary failure set $\Omega_\infty$ of initial conditions $(b_0,s_0) \in \mathbb{D}$ from which the system trajectories $\phi(t;b_0,s_0,u)$ do not converge to the origin regardless of the considered $u \in \mathcal{V}$.
We provide an inner approximation of the success set with rigorous theoretical guarantees.
For initial conditions within the success set, we also consider an optimal control framework to design optimal antibiotic treatment.

\newpage

This Section includes additional material that further elaborates on the results in the Main Manuscript by providing a thorough and detailed numerical analysis that addresses the following aspects:
\begin{enumerate}[noitemsep]
    \item We assess the conservativeness of the inner approximation of the success set in Section III-C of the Main Manuscript. A thorough campaign of numerical simulations consistently shows that our inner approximation of the success set is the actual success set. In fact, when computed with sufficient numerical precision, function $g(b)$ provides an accurate separation between the success and failure sets: trajectories emanating from all the sampled points above the function converge to the origin for at least one of the tested choices of $u$, while trajectories emanating from all the sampled points below the function do not converge to the origin for all the tested choices of $u$.
    \item We conduct a thorough numerical sensitivity analysis to assess how our estimate of the success set, given by the points above the function $g(b)$, varies when the system parameters are varied.
    \item Expanding on Section IV of the Main Manuscript, we provide an additional analysis of optimal antibiotic administration problems over a finite horizon $T$. First, we consider an optimal control problem formulation where the set of admissible therapies corresponds exactly to the set $\mathcal{V}$ considered in the Main Manuscript, with progressively more complex cost functionals, considering combinations of the total number of bacteria at time $T$, the integral of the drug concentration in $[0,T]$, and the integral of the total number of bacteria in $[0,T]$. In fact, \textit{for initial conditions in the success set}, where clearance of the infection (i.e., convergence of the system trajectory to the origin) with some viable periodic therapy is guaranteed, it makes sense to define a metric of interest, which can also take into account possible side effects of the drug, and seek an optimal therapy which clears the infection while minimising that metric.\\
    Moreover, to gain further insight into treatment regimes and the development of resistance from a broader perspective, we also consider an optimal control formulation where the set of admissible therapies is broader and includes $\mathcal{V}$ as a proper subset. Consequently, whenever the initial condition is in the success set, any periodic therapy that successfully eliminates infection is automatically admissible within the framework of the optimal control problem; in particular, the optimal control problem that considers as its cost functional the total number of bacteria at time $T$ is guaranteed to have a solution that achieves clearance, provided that $T$ is large enough.
    The solutions to the optimal control problem numerically validate the observation that the function $g(b)$ actually separates success and failure sets, based on their computation through numerical simulations: in fact, the optimal control steers the system trajectory to the origin when the initial condition is above $g(b)$ and hence in the success set, while the system trajectory converges to the high-bacterial-load equilibrium $b_2^*$ when the initial condition is below $g(b)$ (and hence in the numerically computed failure set).
\end{enumerate}

\subsection{Assessing the Conservativeness of the Approximation of the Success Set}\label{Subsec:Conservaiveness_Succ_Set}

We define the success set as $\Omega_0 \doteq \{(b_0,s_0)\in\mathbb{D} \colon \lim_{t \to \infty} \phi(t;b_0,s_0,u)=(0,0) \mbox{ for some }u \in \mathcal{V}\}$.
In Section III-B of the Main Manuscript, Theorem 9 provides a theoretically guaranteed inner approximation $\Omega_0^*$ of the success set:
\begin{equation}
\Omega_0^* = \bar \Omega_0 \cup \bar \Omega_0^*,
\end{equation}
where
\begin{equation}\label{eq:success_subset}
\bar\Omega_0 = \{(b,s)\in\mathbb{D} \colon b<b^*_1\},
\end{equation}
\begin{equation}\label{eq:success_expansion}
\bar \Omega_0^* = \{(b,s)\in\mathbb{D} \colon b_1^* \leq b \leq N, \, g(b) < s \leq b\},
\end{equation}
and the function $g(b)$ is defined as
\begin{equation}
g(b)\doteq \sup \mathcal{S}(b), \quad b \in [b_1^*, N],
\end{equation}
where
\begin{equation}
\begin{aligned}
\mathcal{S}(b) = &\{ s\in[0,b] \colon \forall u\in\mathcal{V}, \, \lim_{t \to \infty} \phi(t;b,s,u) \neq (0,0) \\ 
&\mbox{ for } b \geq b_1^*\}.
\end{aligned}
\end{equation}
In Proposition 8 in the Main Manuscript, we show that function $g(b)$ is well defined (for any $b\in [b_1^*,N]$, $\mathcal{S}(b)\neq \emptyset$), that $g(b)>0$ for all $b\in ( b_1^*,N ]$, and that $g(b_1^*)=0$.

A numerical approximation of $\bar{\Omega}_0^*$ can thus be obtained by numerically approximating the function $g(b)$.
In Section III-C of the Main Manuscript, we describe our algorithm to compute an approximation of $g(b)$:
\begin{itemize}[noitemsep]
\item fix the number $H \in \mathbb{N}$ of grid points in the interval $b \in (b_1^*, N]$;
\item set the tolerance  $\textit{tol}$;
\item select $m$ viable periodic therapies $u_1,\dots,u_m \in \mathcal{V}$;
\item construct a grid of points $b_j = b_1^* + j \frac{N-b_1^*}{H}$, $j=1,\dots,H$, in the interval $(b_1^*, N]$;
\item for each $u_i$, $i \in \{1,\dots,m\}$, and each $b_j$, $j\in \{1,\dots,H\}$, apply the bisection algorithm in $s$ over the interval $[0,b_j]$:
\begin{itemize}[noitemsep]
\item $(b_j,0) \in \Omega_\infty$, whereas $(b_j,b_j)$ can be either in $\Omega_0$ or in $\Omega_\infty$. If simulation of $\phi(t;b_j,b_j,u_i)$ does not show convergence to $(0,0)$, the process stops;
\item otherwise, consider the point $(b_j,\frac{b_j}{2})$ and simulate $\phi(t;b_j,\frac{b_j}{2},u_i)$; if it converges to $(0,0)$, repeat the bisection on the interval $s \in [0,\frac{b_j}{2}]$, otherwise on the interval $s \in [\frac{b_j}{2},b_j]$;
\item the process continues until an interval $[\sigma_j^L, \sigma_j^H]$ for which $\sigma_j^H-\sigma_j^L < \textit{tol}$ for a prescribed tolerance $\textit{tol}$ is obtained;
\end{itemize}
\item build the set 
\begin{equation*}
\hat \Omega_i = \{(b,s) \colon b_1^* \leq b \leq N, \, g_i(b) < s \leq b\},
\end{equation*}
where $g_i(b) = \sigma_j^H$ on the interval $(b_{j-1},b_j]$, $j\in \{1,\dots,H\}$, and $b_0=b_1^*$;
\item the resulting approximation of $\bar \Omega_0^*$ is the set $\bigcup_{i=1}^m \hat \Omega_i$.
\end{itemize}

In Section III-C of the the Main Manuscript, we consider the system with the parameter values listed in Table~\ref{tb:parameters} and we show (focusing on the interval $b \in [0, 10^6]$ in Figure 3, while the whole interval $b \in [0, 10^7]$ is visualised in Figure 2b) approximations of the success set obained by numerically approximating the values of function $g(b)$ through the above algorithm, with  $H = 10$, or $50$, or $100$, with $\textit{tol}=50$ and with the $25$ viable periodic therapies in $\mathcal{V}$ shown in Table~\ref{tab:tau_values}, which are characterised (cf. Proposition 1 in the Main Manuscript) by (i) $\ell \in (0, \frac{u_{\max}}{u_M}-1]$, with $\frac{u_{\max}}{u_M}-1=15$, which leads to doses $w = \ell u_M$ and corresponding administration periods $\Delta \tau = \frac{\ln{(\ell+1)}}{\delta}$; and (ii) initial administration time $0 \leq \tau_0 \leq \Delta \tau$ as per Definition 2 in the Main Manuscript.

\begin{table}[h!]
\centering
\begin{footnotesize}
\begin{tabular}{|c|c|c|c|c|c|c|}
\hline
\boldmath$w$ & \boldmath$\Delta \tau$ & \boldmath$\tau_0^{(1)}$ & \boldmath$\tau_0^{(2)}$ & \boldmath$\tau_0^{(3)}$ & \boldmath$\tau_0^{(4)}$ & \boldmath$\tau_0^{(5)}$ \\
\hline
$2u_M$  & 4.99 & 0   & 2   & 3   & 4   & 4.5 \\
$5u_M$  & 8.14 & 0   & 2   & 4   & 6   & 8   \\
$7u_M$  & 9.45 & 0   & 3   & 4   & 7   & 9   \\
$10u_M$ & 10.90 & 0   & 3   & 5   & 7   & 10  \\
$15u_M$ & 12.60 & 0   & 5   & 6   & 7   & 8   \\
\hline
\end{tabular}
\end{footnotesize}
\caption{The 25 \textbf{viable periodic therapies} considered in the numerical simulations in Section III-C of the Main Manuscript, as well as in this Supplementary Material: five different doses $w$ are considered and, for each, five different initial administration times $\tau_0$, given by $\tau_0^{(1)},\tau_0^{(2)},\tau_0^{(3)},\tau_0^{(4)}$ and $\tau_0^{(5)}$.}
\label{tab:tau_values}
\end{table}

To assess the conservativeness of our inner approximation, we perform additional numerical experiments aimed at verifying that our theoretically guaranteed inner approximation $\Omega_0^*$ of the success set is indeed a good and tight approximation: ideally, by densely sampling the domain $\mathbb{D}$, we expect to find that the solutions emanating from initial conditions $(b_0,s_0)$ that lie below $g(b)$, i.e. $s_0 < g(b_0)$, do not converge to the origin with any of the considered therapies, while the solutions emanating from initial conditions $(b_0,s_0)$ that lie above $g(b)$, i.e. $s_0 > g(b_0)$, and hence belong to our approximated success set, converge to the origin with at least one of the considered therapies.

First we consider just $m=1$ therapy, chosen as the one with the maximum possible dose $w$, since -- as we have shown in the Main Manuscript -- at least with the considered parameters, this therapy provides an approximation of the success set that contains all the other approximations that can be obtained with the other tested therapies. Then, we consider $m=5$ therapies with different values of $w$ as in Table~\ref{tab:tau_values} and with $\tau_0=0$, since -- as we have shown in the Main Manuscript -- at least with the considered parameters, choosing a different $0 \leq \tau_0 \leq \Delta \tau$ does not affect the numerical results.

\paragraph{Numerical grid analysis with the maximal-dose periodic therapy.}

We construct a dense grid of initial conditions for model \eqref{eq:originalModel}, using a resolution defined by a spacing of $10{,}000$ over the domain $\mathbb{D}$ in \eqref{eq:domain}, resulting in a total of $501{,}501$ points inside $\mathbb{D}$. Among them, the $78$ points with $b\in [0, b^*_1)$ are within the success subset $\bar{\Omega}_0$ in \eqref{eq:success_subset}, and hence any trajectory emanating from these points is guaranteed to converge to the origin regardless of $u$. For each point $(b_0,s_0)$, we numerically compute the trajectory $\phi(t;b_0,s_0,u)$ for a periodic therapy $u$ characterised by a constant maximum dose $w = 15 u_M$ and initial administration time $\tau_0 = 0$. For our simulations, we use MATLAB’s built-in \texttt{ode45} function, which implements a Runge-Kutta method of order 4 and 5 for numerical integration. The numerical solver tolerances are set to \texttt{RelTol = 1e-8} and \texttt{AbsTol = 1e-10}.
Then, each initial condition $(b_0,s_0)$ in $\mathbb{D}$ is classified as \textbf{successful (blue)} if the corresponding trajectory $\phi(t;b_0,s_0,u)$ enters the success subset $\bar{\Omega}_0$ during the simulation; otherwise, it is classified as \textbf{failed (red)}, if it enters the failure subset $\bar{\Omega}_\infty$ (see Eq. (11) in the Main Manuscript).
We also compute the function $g(b)$ over the interval $b \in [b_1^*, N]$ using the method described above (and in the Main Manuscript), considering a grid of $H = 1{,}100$ points and a tolerance of $\textit{tol} = 0.05$, for the periodic therapy with $w = 15 u_M$ and $\tau_0 = 0$. The value of $H$ is selected so that the values $b_k$ at which $g(b_k)$ is computed coincide with the values of $b$ in the grid of initial conditions.

Figures~\ref{fig:gb3000vsGrid10000x1}--\ref{fig:gb3000vsGrid10000x5} visualise the numerical simulation results at different resolutions.
Function $g(b)$ effectively separates the successful and failed initial conditions of the dense grid: all the failed points $(b_0,s_0)$ lie below $g(b)$, i.e. $s_0 < g(b_0)$, and all successful points $(b_0,s_0)$ lie above $g(b)$, i.e. $s_0 > g(b_0)$, and hence belong to our approximated success set.

Therefore, \textbf{our theoretically guaranteed inner approximation $\Omega_0^*$ of the success set appears to be the actual success set $\Omega_0$}: \textbf{function $g(b)$ provides an accurate separation between the success and failure sets, provided that it is computed with sufficient numerical precision}, and points that are not in our approximation $\Omega_0^*$ of the success set $\Omega_0$ appear to actually be in the failure set $\Omega_\infty$.

The only source of conservativeness is due to the \emph{numerical approximation} caused by the fact that the values of $g(b)$ can be computed numerically (with finite machine precision) only at a \emph{finite number of points} $b_k$, and the points $(b_k, g(b_k))$ are then visually connected by a “staircase” that provides an inner approximation of the success set, since the value of the function is kept constant and equal to $g(b_k)$ for all $b \in (b_{k-1}, b_k]$. The discrepancy between the numerically computed “staircase” $g(b)$ and the true function $g(b)$, due to the discretisation, converges to zero if the number of grid points $H \to \infty$.

\paragraph{Refined numerical grid analysis with five periodic therapies.}

To further validate the accuracy of function $g(b)$ in separating success and failure sets, we conduct additional simulations over a refined triangular grid defined on the domain $\mathbb{D}$ for $b \in [0, 5 \cdot 10^5]$, with a spacing of $1{,}000$. We simulate the system trajectories emanating from each point in the grid with five distinct periodic therapies, all with $\tau_0 = 0$ and impulsive doses $w \in \{2u_M, 5u_M, 7u_M, 10u_M, 15u_M\}$, as specified in Table~\ref{tab:tau_values}. A point $(b_0,s_0)$ is classified as successful (blue) if the corresponding trajectory enters the set $\bar{\Omega}_0$ with at least one of the five therapies, and failed (red) if the corresponding trajectory enters the subset $\bar{\Omega}_\infty$ with all five therapies.
Also, we compute the function $g(b)$ over the interval $b \in [0, 5 \cdot 10^5]$ using the method described above (and in the Main Manuscript), with $H = 501$ grid points and a stopping tolerance of $0.05$. The value of $H$ is selected so that the values $b_k$ at which $g(b_k)$ is computed coincide with the values of $b$ in the grid of initial conditions.
Figures~\ref{fig:gb5000vsGrid10000_Reducedx1}--\ref{fig:gb5000vsGrid10000_Reducedx4} visualise the numerical simulation results at different resolutions.
Also in this case, function $g(b)$ effectively separates the successful and failed initial conditions of the dense grid: all the failed points $(b_0,s_0)$ lie below $g(b)$, i.e. $s_0 < g(b_0)$, and all successful points $(b_0,s_0)$ lie above $g(b)$, i.e. $s_0 > g(b_0)$, and hence belong to our approximated success set.

\subsection{Parameter Sensitivity Analysis}\label{Subsec:Parameter_Sen_Analysis}

We conduct here a thorough numerical sensitivity analysis to see how our estimate of the success set $\Omega_0$ based on the computation of $g(b)$, which is
\begin{equation}\label{eq:successsetfullapprox}
\begin{aligned}
\Omega_0^* = \bar \Omega_0 \cup \bar \Omega_0^* = \{(b,s)\in\mathbb{D} \colon b<b^*_1\} \cup {} \\
\{(b,s)\in\mathbb{D} \colon b_1^* \leq b \leq N, \, g(b) < s \leq b\},
\end{aligned}
\end{equation}
is affected by variations in the system parameters.

The nominal parameter values are provided in Table~\ref{tb:parameters}. For our local sensitivity analysis, we keep all the parameters fixed to their nominal value apart from a single parameter at a time, for which we consider six additional values within an interval corresponding to $[\frac{1}{2}p, \frac{3}{2}p]$, where $p$ is the nominal parameter value, while ensuring that all the assumptions in the Main Manuscript -- specifically, Equation (6) and Assumption 1 -- remain satisfied. The function $g(b)$ is computed using a grid of $H = 100$ points uniformly distributed over the interval $[b_1^*, N]$, with $\textit{tol} = 50$, considering the periodic therapy given by $w=15u_M$ and $\tau_0=0$.

To quantify the variability of $g(b)$ across different parameter values, we define two indices, $I_1$ and $I_2$ for each parameter. Let $p$ be a model parameter with $m$ values ${p_1, p_2, \dots, p_m}$ considered in the sensitivity analysis ($m = 7$ in our case). For each value $p_i$, we evaluate function $g_{p_i}(b)$ on a common grid of points ${b_1, b_2, \dots, b_H}$ of size $H = 100$. Then,

\textbf{Index 1} quantifies the maximum variation in the equilibrium point $b_1^*$, and is defined as:
    \[
    I_1(p) = \max_{1 \leq i,j \leq n} \left| b^*_{1, p_i} - b^*_{1, p_j} \right|
    \]
    where $b^*_{1, p_i}$ denotes the smallest (nonzero) equilibrium point in Eq. (9) of the Main Manuscript. The value of $b_1^*$ is fundamental in the definition of the success set, as shown in \eqref{eq:successsetfullapprox}, and of the function $g(b)$.

\textbf{Index 2} measures the maximum pointwise deviation between any pair of threshold curves, and is defined as:
    \[
    I_2(p) = \max_{1 \leq i < j \leq n} \, \max_{1 \leq k \leq H} \left| g_{p_i}(b_k) - g_{p_j}(b_k) \right|,
    \]
    corresponding to the $\ell^\infty$ norm of the difference between all pairs of curves over the considered grid.

\medskip

Tables~\ref{tab:IndicesSideBySide} (a) and (b) report the values of Index 1 and Index 2 respectively, both ordered from the largest to the smallest. Note that for the parameter $N$, Index 2 is not computed, because variations in $N$ significantly alter the size of the domain over which $g(b)$ is defined, rendering Index 2 non-informative in this case. Parameters $\alpha$, $\beta$ and $K$, which characterise the spontaneous net growth rate of bacteria and the immune system response, turn out to have the largest impact on the shape of the success set approximation, according to both indices. The dependence of $I_1$ on the model parameters $\alpha$, $\beta$, $K$ and $N$ can be also analysed analytically using Eq.~(9) of the Main Manuscript, while $I_1$ is unaffected by all the other parameters.
The approximations of $g(b)$ computed for different parameters values  are shown in Figures~\ref{fig:changingalpha}--\ref{fig:changingN}.

\begin{table}[h!]
\centering
\begin{minipage}[t]{0.48\textwidth}
\centering
\begin{tabular}{lc}
\textbf{Parameter} & \textbf{Index} ${I_1}$ \\
\hline
$\alpha$ & $1.4625 \cdot 10^5$ \\
$\beta$  & $1.1463 \cdot 10^5$ \\
$K$      & $1.1317 \cdot 10^5$ \\
$N$      & $1.7574 \cdot 10^3$ \\
$\mu$    & $0$ \\
$u_M$      & $0$ \\
$E_{\max}$ & $0$ \\
$EC_{50}$ & $0$ \\
$k$      & $0$ \\
$\delta$ & $0$ \\
\end{tabular}
\caption*{(a) Index $I_1$: Difference in $b^*_1$.}
\end{minipage}%
\hfill
\begin{minipage}[t]{0.48\textwidth}
\centering
\begin{tabular}{lc}
\textbf{Parameter} & \textbf{Index} ${I_2}$ \\
\hline
$\alpha$  & $1.1696 \cdot 10^5$ \\
$\beta$   & $4.9108 \cdot 10^4$ \\
$K$       & $4.8308 \cdot 10^4$ \\
$E_{\max}$& $4.4231 \cdot 10^4$ \\
$k$       & $7.4364 \cdot 10^3$ \\
$EC_{50}$ & $4.9591 \cdot 10^3$ \\
$u_M$       & $3.7384 \cdot 10^3$ \\
$\delta$  & $3.6264 \cdot 10^3$ \\
$\mu$     & $2.2888 \cdot 10^2$ \\
$N$       & - \\
\end{tabular}
\caption*{(b) Index $I_2$: Difference in $g(b)$.}
\end{minipage}
\caption{Indices $I_1$ and $I_2$ quantifying how each of the model parameters affects the approximation of the success set. Parameters are listed in descending order according to the value of each index.}
\label{tab:IndicesSideBySide}
\end{table}

\subsection{Bacterial Minimisation and Optimal Antibiotic Administration:\\ Increasingly Complex Optimal Control Frameworks} \label{Subsec:Bacterial_Minimisation}

For initial conditions in the success set, where clearance of the infection (i.e., convergence of the system trajectory to the origin) with some viable periodic therapy is guaranteed, it makes sense to define a metric of interest, which can also take into account possible side effects of the drug (assumed to be proportional to the drug concentration), and look for an optimal therapy that clears the infection while minimising that metric.

Here, we discuss some progressively more complex optimal control formulations to solve the optimal antibiotic administration problem over a finite horizon $T$.

First, we consider as admissible controls those associated with the viable periodic therapies $u\in \mathcal{V}$ considered in the Main Manuscript, and we consider increasingly complex cost functionals that include:
\begin{itemize}[noitemsep]
\item only the total number of bacteria $b(T)$ at the end of the finite horizon;
\item a combination of $b(T)$ and of the integral of the drug concentration $\int_0^T u(t) dt$ over the horizon;
\item a combination of $b(T)$, of the integral of the drug concentration $\int_0^T u(t) dt$ and of the integral of the total number of bacteria $\int_0^T b(t) dt$.
\end{itemize}

Whenever the initial condition is in the success set, the optimal control problem that considers as its cost functional the total number of bacteria at time $T$ is guaranteed to have a solution that achieves clearance, provided that $T$ is large enough.
This is confirmed by our numerical results: the solutions to the optimal control problem numerically validate the observation that the function $g(b)$ actually separates success and failure sets (based on their computation through numerical simulations), because the optimal control steers the system trajectory to the origin when the initial condition is chosen above $g(b)$ (and hence in the success set), while the system trajectory converges to the high-bacterial-load equilibrium $b_2^*$ when the initial condition is below $g(b)$ (and hence outside the success set, i.e. in the numerically computed failure set).

Conversely, in the other two scenarios, certain choices of the cost weights may lead to undesired solutions that do not guarantee clearance of the infection even for initial conditions in the success set.

To gain further insight into treatment regimes and the development of resistance from a broader perspective, we also consider a larger set of admissible controls and thus a more flexible optimal control formulation. We consider both an optimal control problem formulation where the set of admissible therapies corresponds exactly to the set $\mathcal{V}$ considered in the Main Manuscript, and an optimal control formulation where the set of admissible therapies is broader and includes $\mathcal{V}$ as a proper subset. Consequently, whenever the initial condition is in the success set, any periodic therapy that successfully eliminates infection is automatically admissible within the framework of the optimal control problem.

In all our optimal control simulations, we consider the values of the system parameters listed in Table~\ref{tb:parameters}, with the exception of $\beta = 10$.

\subsubsection{Finding the Optimal Viable Periodic Therapy}

We first consider optimal control problems where the set of admissible controls corresponds to viable periodic therapies $u\in \mathcal{V}$ considered in the Main Manuscript.

As recalled earlier on in the Supplementary Material Section S1, for viable periodic therapies $u\in \mathcal{V}$, Proposition 1 in the Main Manuscript provides the mathematical relationship between the antibiotic dose $w=\ell u_M$ and its corresponding administration period $\Delta \tau=\frac{\ln(\ell+1)}{\delta}$ so as to guarantee that the constraint $u(t)\in \{\upsilon \colon u_M \leq \upsilon \leq u_{\max}\}$ is satisfied for all $t$. Moreover, $w \in \mathcal{W}=\{\upsilon \colon 0 < \upsilon \le u_{\max}-u_M\}$ is also guaranteed provided that $\ell \in (0, \ell_{\max}]$, where $\ell_{\max} = (u_{\max}-u_M)/u_M=\frac{u_{\max}}{u_M}-1$. Given $\ell$, one can immediately compute the corresponding pair $(\Delta \tau, w)$ and thus the corresponding viable periodic therapy $u\in \mathcal{V}$.

We can therefore choose $\ell$ as our decision variable and consider the optimal control problem

\begin{eqnarray}
    && \min_\ell \;\; \int_0^T \left(\zeta_1 u(t) + \zeta_3 b(t)\right)dt \; + \zeta_2 b(T) \label{eq:optimalcontrolV_SP}\\
    & \mbox{s.t. } & \eqref{eq:originalModel}-\eqref{eq:udynamic}, \nonumber\\
    & & w_j \equiv w = \ell u_M, \nonumber\\
    & & \tau_0=0,\;\; \tau_{j+1}-\tau_j \equiv \Delta\tau = \frac{\ln(\ell+1)}{\delta},\nonumber\\
    & & \ell \in (0, \ell_{\max}], \nonumber
\end{eqnarray}

which we solve over a discrete grid of step $0.01$ with combinations of $\zeta_i$ that reflect the scenarios outlined above.

\subsubsection{Aiming at bacterial eradication: $\zeta_2=1$ and $\zeta_1=\zeta_3=0$}

We consider the objective of minimising $b(T)$, the total number of bacteria at the end of the treatment horizon $T$.
This amounts to setting cost functional weights $\zeta_2=1$ and $\zeta_1=\zeta_3=0$ in the optimal control problem \eqref{eq:optimalcontrolV_SP}.

Theorem 9 in the Main Manuscript provides us with a set $\Omega_0^*$ of initial conditions from which the system trajectory can be steered to the infection-free equilibrium at the origin by \textit{some} therapy in $\mathcal{V}$. Given any initial condition in $\Omega_0^*$, the therapy obtained through the solution to the optimal control problem will thus achieve clearance, provided that the horizon $T$ is sufficiently large.

We consider three different initial conditions belonging to the success set:
\begin{itemize}[noitemsep]
\item IC1: $s_0= 6 \cdot 10^6$ and $r_0=3\cdot 10^3$, hence $b_0=6.003  \cdot 10^6$,
\item IC2: $s_0 = 6.002 \cdot 10^6$ and $r_0 = 10^3$, hence $b_0=6.003  \cdot 10^6$,
\item IC3: $s_0 = 6\cdot 10^6$ and $r_0 = 3\cdot 10^2$, hence $b_0=6.0003  \cdot 10^6$.
\end{itemize}

For IC1, we get the results depicted in Figure~\ref{fig:suppl_OCP1}. The cost functional is non-increasing as $\ell$ increases and the minimum $J=0$ is achieved for $\ell \geq 3.90$. Thus, to avoid unnecessary use of drug, the optimal solution corresponding to the smallest drug dose ($\ell = 3.90$) is selected.

For IC2, the results illustrated in Figure~\ref{fig:suppl_OCP3} are obtained. Due to the lower initial number of resistant bacteria, a smaller drug dose is enough to reduce the peak of resistant bacteria and the width of the lobe compared to the previous case.

For IC3, as a consequence of the much lower initial number of resistant bacteria and of the lower initial total number of bacteria, we observe a much smaller peak of the resistant bacteria even with a very small drug dose. Indeed, $J=0$ for $\ell \geq 0.24$, thus the dosage $w = 0.36$ with administration time of $0.98 \approx 1$ hour is chosen, as shown in Figure~\ref{fig:suppl_OCP2}.

We now consider an initial condition that does not belong to $\Omega_0^*$: $s_0 = 6\cdot 10^6$ and $r_0 = 3 \cdot 10^5$, so that $b_0 = 6.3\cdot 10^6$. The results are illustrated in Figure~\ref{fig:suppl_OCP4}. As expected, the therapy fails and the value of the cost functional increases. Thus, the minimum is achieved at the very beginning, corresponding to $\ell = 0.01$. However, the value of $J$ is practically constant. In fact, the relative error between the extreme values of $J$ achieved in the considered interval of values of $\ell$ is of the order of $10^{-12}$, which is close to machine precision; hence the differences in the values of $J$ can be ascribed to numerical errors.

The solutions to the optimal control problem presented in this section provide an additional numerical validation of the observation that function $g(b)$ actually separates success and failure sets. In all our considered simulations (only some of which are reported here), \textbf{the optimal control solution has been able to steer the system trajectory to the origin if, and only if, the initial condition was in our theoretically guaranteed approximation of the success set $\Omega_0^*$}.

\subsubsection{Mitigating side effects: $\zeta_3=0$}

To also take into account the side effects associated with the excessive use of antibiotics, we introduce an additional cost term that penalises high antibiotic concentrations. Hence, the optimal control problem is formulated to balance effective bacterial elimination with the need to reduce potential adverse effects.
This amounts to setting cost functional weights $\zeta_1 \neq 0$, $\zeta_2 \neq 0$ and $\zeta_3=0$ in the optimal control problem \eqref{eq:optimalcontrolV_SP}.

Again, we consider different initial conditions to represent different patient scenarios.

For IC1 ($s_0 = 6 \cdot 10^6$ and $r_0 = 3 \cdot 10^3$), we show the results for different weight choices. Specifically, when $\zeta_1=\zeta_2=1$, we get the (unique) optimal solution $\ell = 3.88$, see Figure~\ref{fig:suppl_OCP5}. Indeed, for $\ell > 3.88$, the cost function is slightly increasing (differently from the previous case, as a consequence of the penalisation of higher drug concentrations).
Increasing the weight on the drug concentration to $\zeta_1 = 10^4$ leads to the results shown in Figure~\ref{fig:suppl_OCP6}: the unique minimum is achieved for $\ell = 3.89$ and, for $\ell > 3.89$, the cost function is visibly increasing. 

When the initial condition is outside the success set ($s_0= 6\cdot 10^6$ and $r_0 = 3\cdot10^5$), the obatined optimal therapy is never able to clear the infection, for any choice of the cost weights; see e.g. Figure~\ref{fig:suppl_OCP11} where $\zeta_1 = \zeta_2 = 1$. Differently from the scenario with the cost $b(T)$ (i.e., $\zeta_2=1$ and $\zeta_1=\zeta_3=0$), now $J$ significantly increases when $\ell$ increases, as a consequence of the penalisation of higher antibiotic concentrations.

\subsubsection{Mitigating side effects and the bacterial load over the whole horizon}

We finally consider the most comprehensive cost functional in the optimal control problem \eqref{eq:optimalcontrolV_SP}, with $\zeta_1, \zeta_2, \zeta_3 \neq 0$, so as to also consider the goal of minimising the bacterial load over the entire duration of treatment, rather than only at the final time.

We focus on the initial condition IC1 in the success set ($s_0= 6\cdot 10^6$ and $r_0 = 3\cdot10^3$), and consider different weights in the cost functional. When $\zeta_i =1$ for $1=1,2,3$, the optimal solution is obtained for $\ell = 5$, which corresponds to the maximum allowed dose $w = u_{\max} - u_M$; see Figure~\ref{fig:suppl_OCP7}.

We thus increase the weight $\zeta_1$, related to the antibiotic concentration, to see whether smaller drug doses can clear the infection. For $\zeta_1= 10^5$, as shown in Figure~\ref{fig:suppl_OCP8}, the cost $J$ achieves its minimum for $\ell=3.93$, which leads to clearance of the infection, but with a larger peak and larger lobe in the resistant bacterial population.

When $\zeta_1= 10^5$, $\zeta_2 =1$ and $\zeta_3 = 10^5$, the results are illustrated in Figure~\ref{fig:suppl_OCP9}. Again, $\ell = 5$ provides the minimum of $J$ in order to clear the infection. Conversely, with  $\zeta_1= 10^5$, $\zeta_2 =10^5$ and $\zeta_3 = 1$, the results are depicted in Figure~\ref{fig:suppl_OCP10}, showing that $\ell = 3.93$ minimises $J$, although again with a taller and larger lobe in the resistant bacterial population.

Overall, we observe that, when the integral weights on the concentration $u$ and the total number of bacteria are the same, then the minimisation of $b(t)$ prevails, leading to a larger optimal antibiotic dose. When the weights on the integral of $u$ and the number of bacteria at the end of the treatment period are equal, the one on the concentration of the drug prevails and leads to a smaller optimum antibiotic dose. 

Finally, we consider an example of initial condition outside the success set ($s_0 = 6\cdot 10^6$ and $r_0 = 3\cdot 10^5$), with $\zeta_1 = \zeta_2 = \zeta_3 = 1$. No admissible therapy can clear the infection and the bacterial population increases, with $b \to b_2^*$; see Figure~\ref{fig:suppl_OCP12}.

\subsection{Optimal antibiotic treatment with a larger class of viable therapies}

Theorem 9 in the Main Manuscript provides us with a set $\Omega_0^*$ of initial conditions from which the system trajectory can be steered to the infection-free equilibrium by \textit{some} therapy in $\mathcal{V}$. Given an initial condition in $\Omega_0^*$, we seek a suitable antibiotic therapy within an optimal control framework.

In this section, we consider a more general optimal control problem, which allows for increased flexibility by allowing the choice of different drug doses at different administration times, while the administration period $\Delta \tau$ remains fixed, as well as early termination of the treatment (drug administration can be interrupted before the end of the treatment horizon), which is allowed to help further reduce the drug concentration in the patient and avoid unnecessary use of antibiotics. We then look for the optimal drug dosage that minimises a combination of the integral bacterial load and drug concentration in the bloodstream.

To this aim, we consider a larger class of viable therapies as per Definition 1 in the Main Manuscript, with $\tau_0=0$, constant $\tau_{j+1}-\tau_j = \Delta \tau$ $\forall j$, and doses
\begin{equation*}
w(t) = \begin{cases}
\ell(\tau_j) u_M \in \mathcal{W} \mbox{ for } t=\tau_j,\\
0 \mbox{ otherwise }
\end{cases}
\end{equation*}
as our control variables.
Yet, as mentioned, for more flexibility, we allow the therapy to stop before the end of our finite horizon $T$, and hence we do not enforce the constraint $u(t) \geq u_M$.

Given weights $\zeta_1, \zeta_3 \geq 0$, our optimal control problem formulation
\begin{eqnarray}
    & & \min_{w(t)} \int_0^T \left(\zeta_1 u(t) + \zeta_3 b(t)\right)dt \label{eq:ocp} \\
    & \mbox{s.t. } & \eqref{eq:originalModel}-\eqref{eq:udynamic}, \nonumber \\
    & & 0 \leq u \leq u_{\max}, \nonumber\\
    & & 0 \leq w_j \leq u_{\max}-u_M, \nonumber\\
    & & \tau_0=0,\;\; \tau_{j+1}-\tau_j \equiv \Delta\tau \nonumber
\end{eqnarray}
aims at minimising the total amount of bacteria and the drug-induced side effects over $[0,T]$. 

The set $\mathcal{V}$ of viable periodic therapies that we consider in our theoretical framework constitutes a proper subset of the admissible controls for the optimal control problem \eqref{eq:ocp}. Consequently, any viable periodic therapy that successfully clears the infection is automatically admissible within the framework of the optimal control problem \eqref{eq:ocp}. Therefore, for initial conditions within the success set, we are guaranteed the existence of at least one infection-clearing therapy that is available to the optimal control problem; if the optimal solution does not select this particular therapy, this necessarily indicates the existence of alternative therapies that achieve a superior performance according to the considered optimisation criterion (although they may not clear the infection).

For $\Delta \tau = 8$ h and $T=5$ days, we numerically solve the optimal control problem with a multiple shooting method using CasAdi \cite{casadi} in MATLAB. The chosen parameter values are reported in the caption of Fig.~\ref{fig:ocp_test}, which compares optimal therapies obtained with three different choices of the weights.
When $\zeta_3 \gg \zeta_1$ (Fig.~\ref{fig:ocp_test}a), higher doses are maintained until the end of treatment.
When $\zeta_3=\zeta_1$ (Fig.~\ref{fig:ocp_test}b), the optimal strategy yields the maximum antibiotic dose for the first day, after which the treatment stops. When $\zeta_1 \gg \zeta_3$ (Fig.~\ref{fig:ocp_test}c), the antibiotic doses are maximal in the first 16 hours, then a smaller dose is given. 
Despite these different antibiotic administration profiles, the time evolution of both susceptible and resistant bacteria is almost identical in all scenarios and is reported in the second row in Fig.~\ref{fig:ocp_test}: $r$ increases until $s$ is reduced enough to allow the immune system to target resistant bacteria with sufficient capacity to induce their decrease, leading to a complete eradication of the infection. The peaks of resistant bacteria vary slightly between cases: $40723$ in (a), $40725$ in (b), and $40732$ in (c). This finding suggests that lower drug doses administered for shorter periods can effectively clear bacterial infections while minimising patient exposure to medication.


\newpage

\onecolumn

\begin{sidewaysfigure}
    \centering
    \includegraphics[width=1\textwidth]{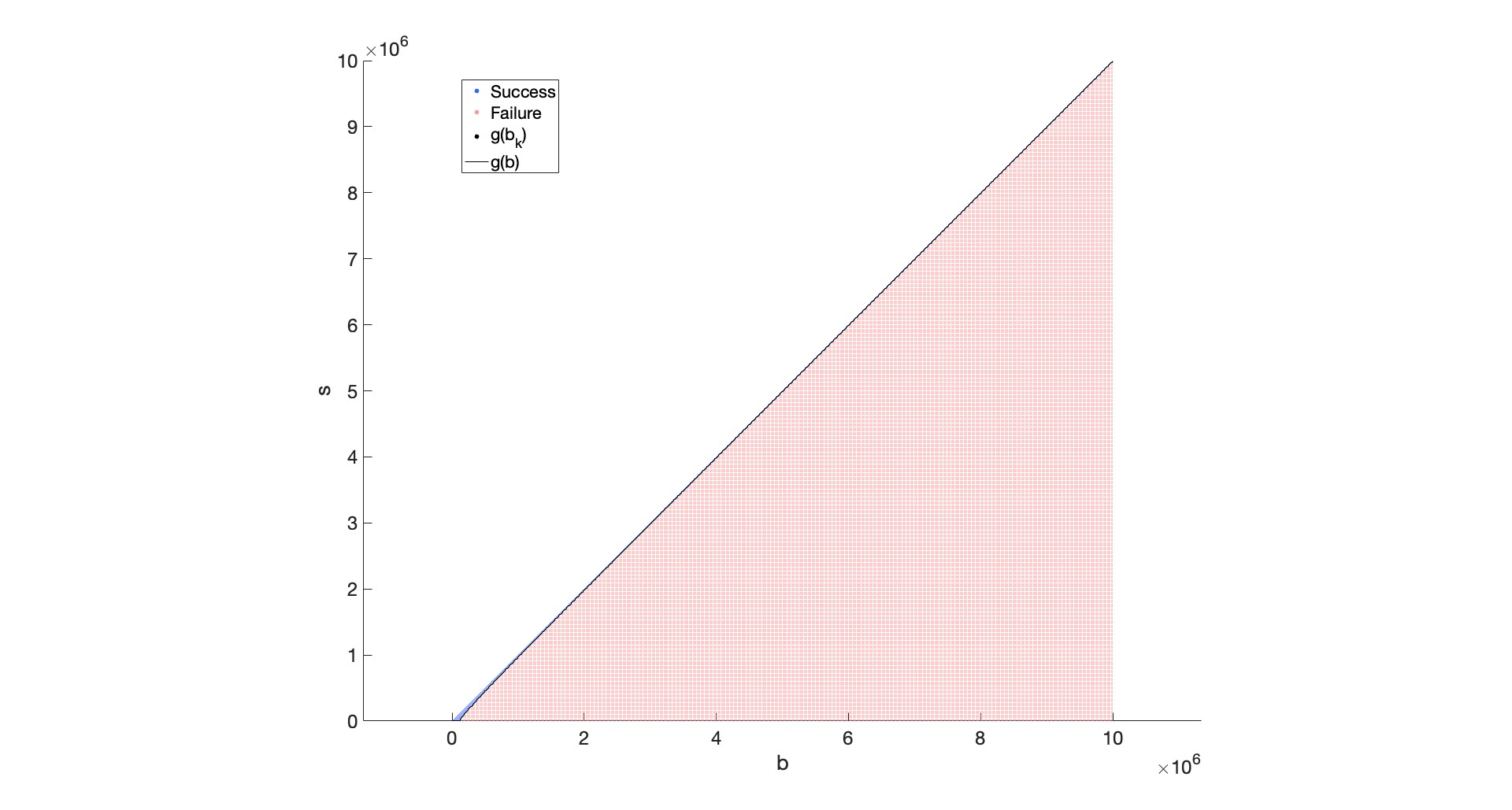}
    \caption{Numerical simulation results shown over the interval $b \in [0, N]$, with $N = 10^7$. Initial conditions from which the trajectories converge to the origin are marked in blue, initial conditions from which the trajectories do not converge to the origin are marked in red. The values $g(b_k)$ of $g(b)$ computed at the grid points $b_k$ are marked with black circles; the “staircase” numerical approximation of the function $g(b)$, obtained by keeping the value of the function constant and equal to $g(b_k)$ for all $b \in (b_{k-1}, b_k]$, is shown in black. All the red points $(b_k,s_k)$ lie below $g(b)$, i.e. $s_k < g(b_k)$, and all blue points $(b_k,s_k)$ lie above $g(b)$, i.e. $s_k > g(b_k)$, and hence belong to our approximated success set.}
    \label{fig:gb3000vsGrid10000x1}
\end{sidewaysfigure}
\begin{sidewaysfigure}
        \centering
        \includegraphics[width=1\textwidth]{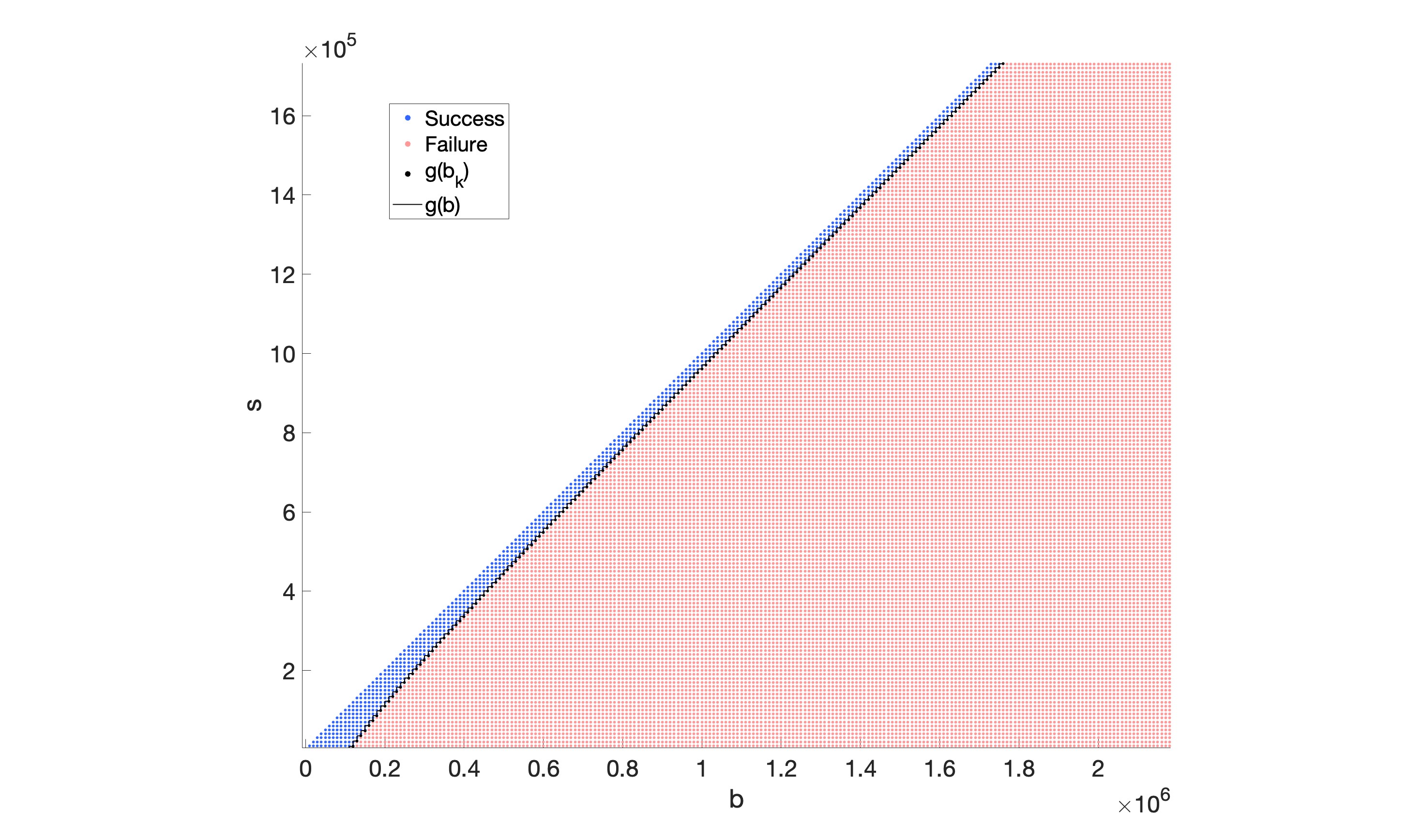}
    \caption{Numerical simulation results shown over the interval $b \in [0, 2\cdot 10^6]$. Initial conditions from which the trajectories converge to the origin are marked in blue, initial conditions from which the trajectories do not converge to the origin are marked in red. The values $g(b_k)$ of $g(b)$ computed at the grid points $b_k$ are marked with black circles; the “staircase” numerical approximation of the function $g(b)$, obtained by keeping the value of the function constant and equal to $g(b_k)$ for all $b \in (b_{k-1}, b_k]$, is shown in black. All the red points $(b_k,s_k)$ lie below $g(b)$, i.e. $s_k < g(b_k)$, and all blue points $(b_k,s_k)$ lie above $g(b)$, i.e. $s_k > g(b_k)$, and hence belong to our approximated success set.}
    \label{fig:gb3000vsGrid10000x2}
\end{sidewaysfigure}
\begin{sidewaysfigure}
        \centering
        \includegraphics[width=1\textwidth]{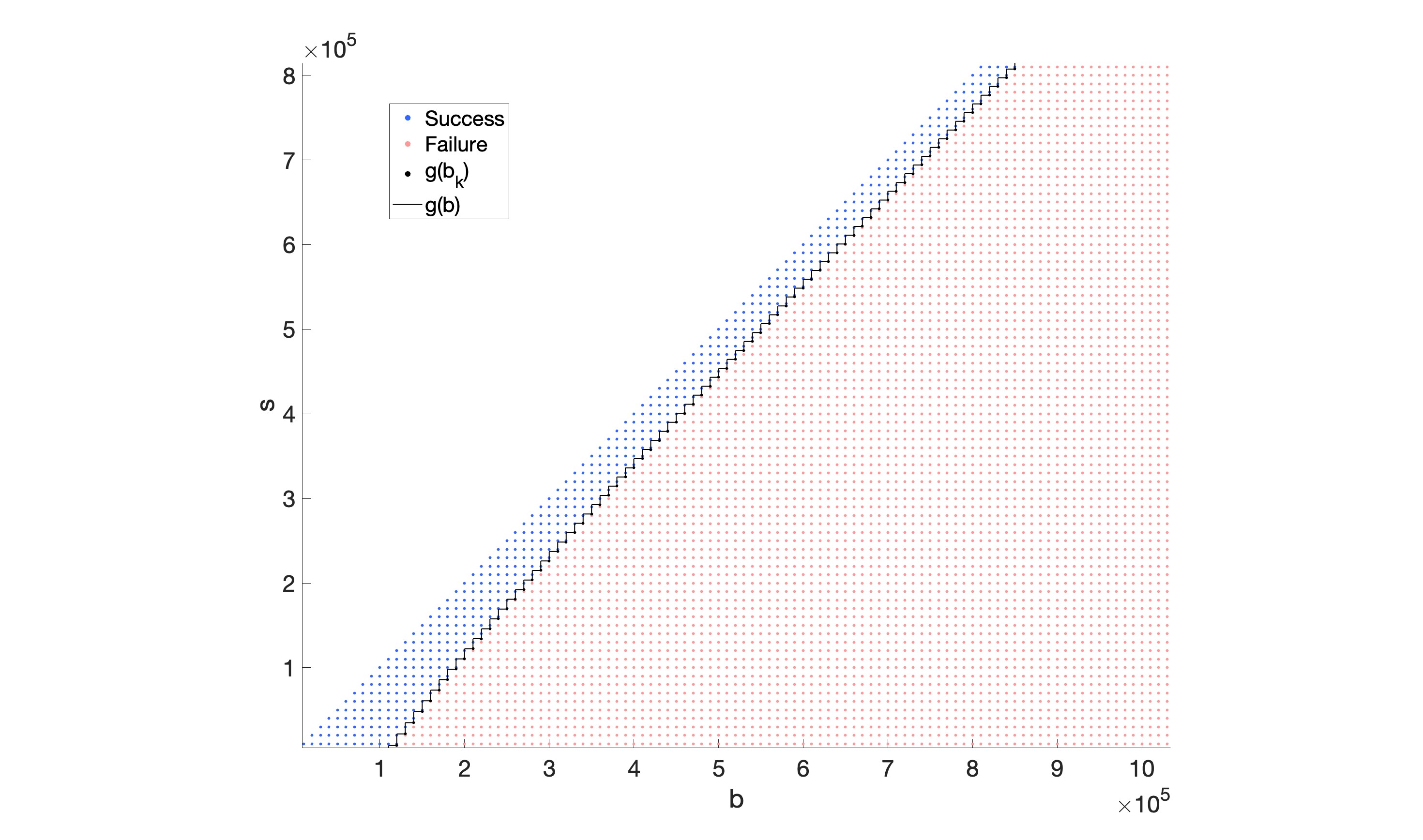}
    \caption{Numerical simulation results shown over the interval $b \in [0, 1\cdot 10^6]$. Initial conditions from which the trajectories converge to the origin are marked in blue, initial conditions from which the trajectories do not converge to the origin are marked in red. The values $g(b_k)$ of $g(b)$ computed at the grid points $b_k$ are marked with black circles; the “staircase” numerical approximation of the function $g(b)$, obtained by keeping the value of the function constant and equal to $g(b_k)$ for all $b \in (b_{k-1}, b_k]$, is shown in black. All the red points $(b_k,s_k)$ lie below $g(b)$, i.e. $s_k < g(b_k)$, and all blue points $(b_k,s_k)$ lie above $g(b)$, i.e. $s_k > g(b_k)$, and hence belong to our approximated success set.}
    \label{fig:gb3000vsGrid10000x3}
\end{sidewaysfigure}
%
%
\begin{sidewaysfigure}
    \centering
        \centering
        \includegraphics[width=1\textwidth]{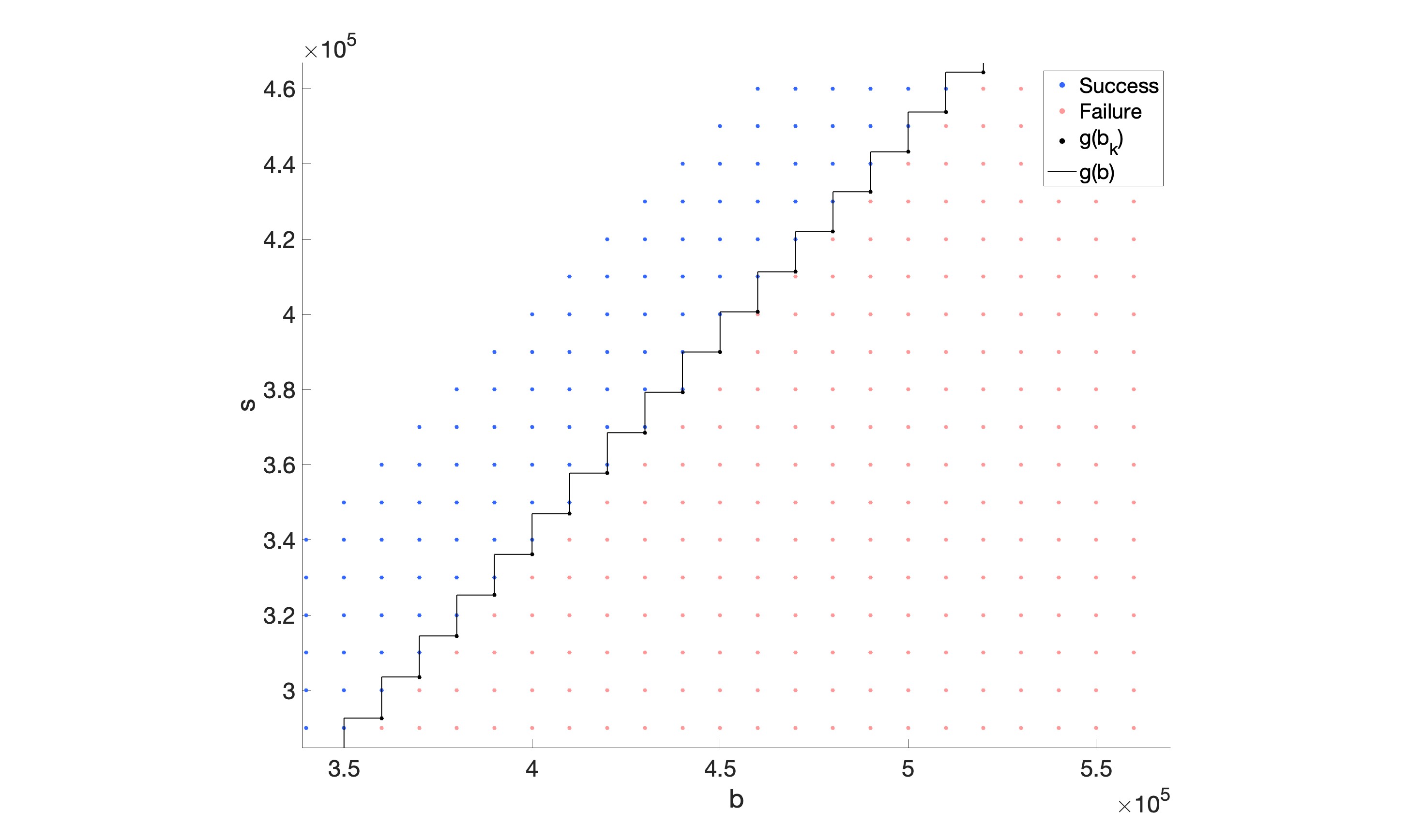}
    \caption{Numerical simulation results shown over the interval $b \in [3.5 \cdot 10^5, 5.5\cdot 10^5]$. Initial conditions from which the trajectories converge to the origin are marked in blue, initial conditions from which the trajectories do not converge to the origin are marked in red. The values $g(b_k)$ of $g(b)$ computed at the grid points $b_k$ are marked with black circles; the “staircase” numerical approximation of the function $g(b)$, obtained by keeping the value of the function constant and equal to $g(b_k)$ for all $b \in (b_{k-1}, b_k]$, is shown in black. All the red points $(b_k,s_k)$ lie below $g(b)$, i.e. $s_k < g(b_k)$, and all blue points $(b_k,s_k)$ lie above $g(b)$, i.e. $s_k > g(b_k)$, and hence belong to our approximated success set.}
    \label{fig:gb3000vsGrid10000x5}
\end{sidewaysfigure}


\begin{sidewaysfigure}
    \centering
    \includegraphics[width=1\textwidth]{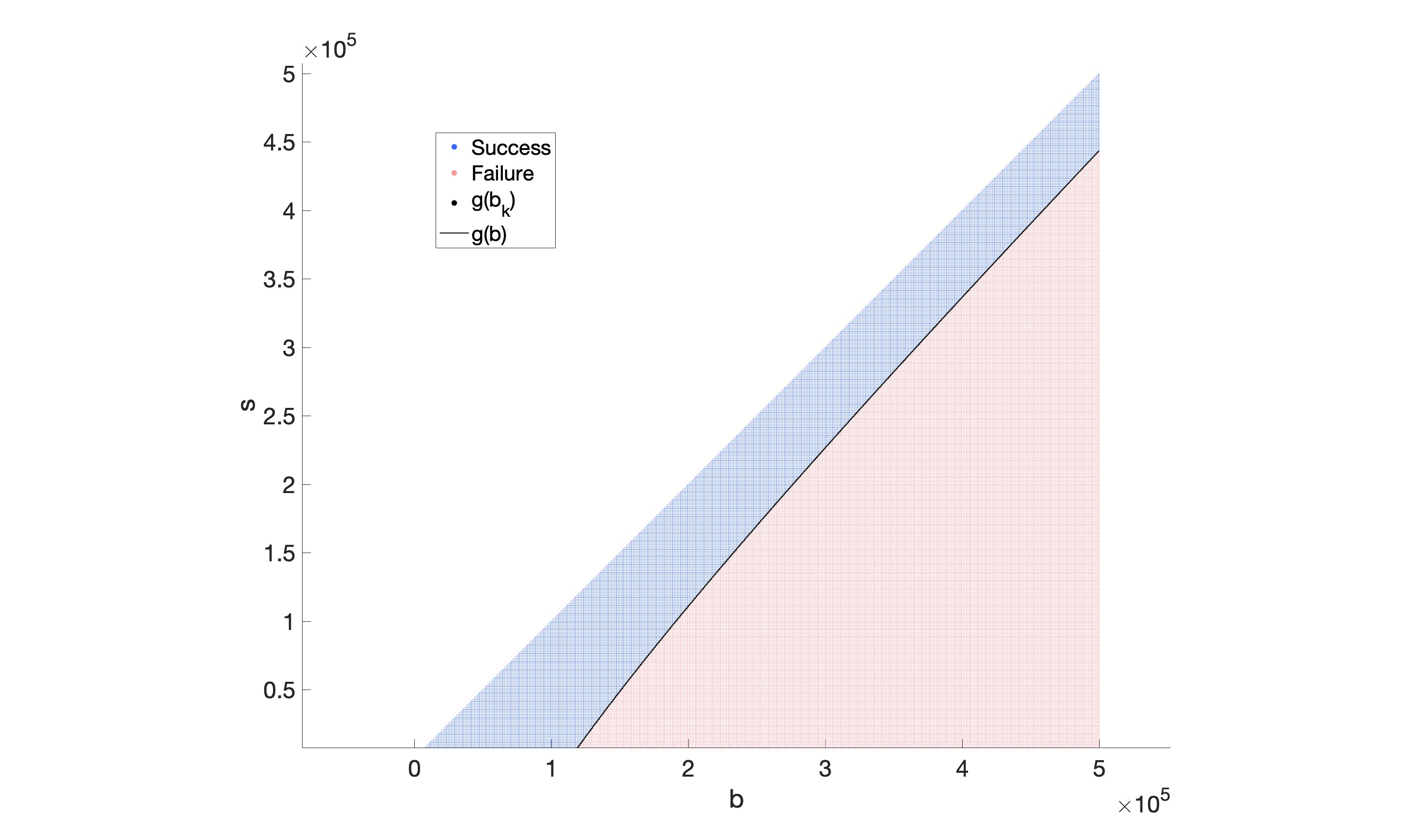}
    \caption{Numerical simulation results shown over the interval $b \in [0, 5\cdot 10^5]$. Initial conditions from which the trajectories converge to the origin are marked in blue, initial conditions from which the trajectories do not converge to the origin are marked in red. The values $g(b_k)$ of $g(b)$ computed at the grid points $b_k$ are marked with black circles; the “staircase” numerical approximation of the function $g(b)$, obtained by keeping the value of the function constant and equal to $g(b_k)$ for all $b \in (b_{k-1}, b_k]$, is shown in black. All the red points $(b_k,s_k)$ lie below $g(b)$, i.e. $s_k < g(b_k)$, and all blue points $(b_k,s_k)$ lie above $g(b)$, i.e. $s_k > g(b_k)$, and hence belong to our approximated success set.}
    \label{fig:gb5000vsGrid10000_Reducedx1}
\end{sidewaysfigure}
\begin{sidewaysfigure}
        \centering
        \includegraphics[width=1\textwidth]{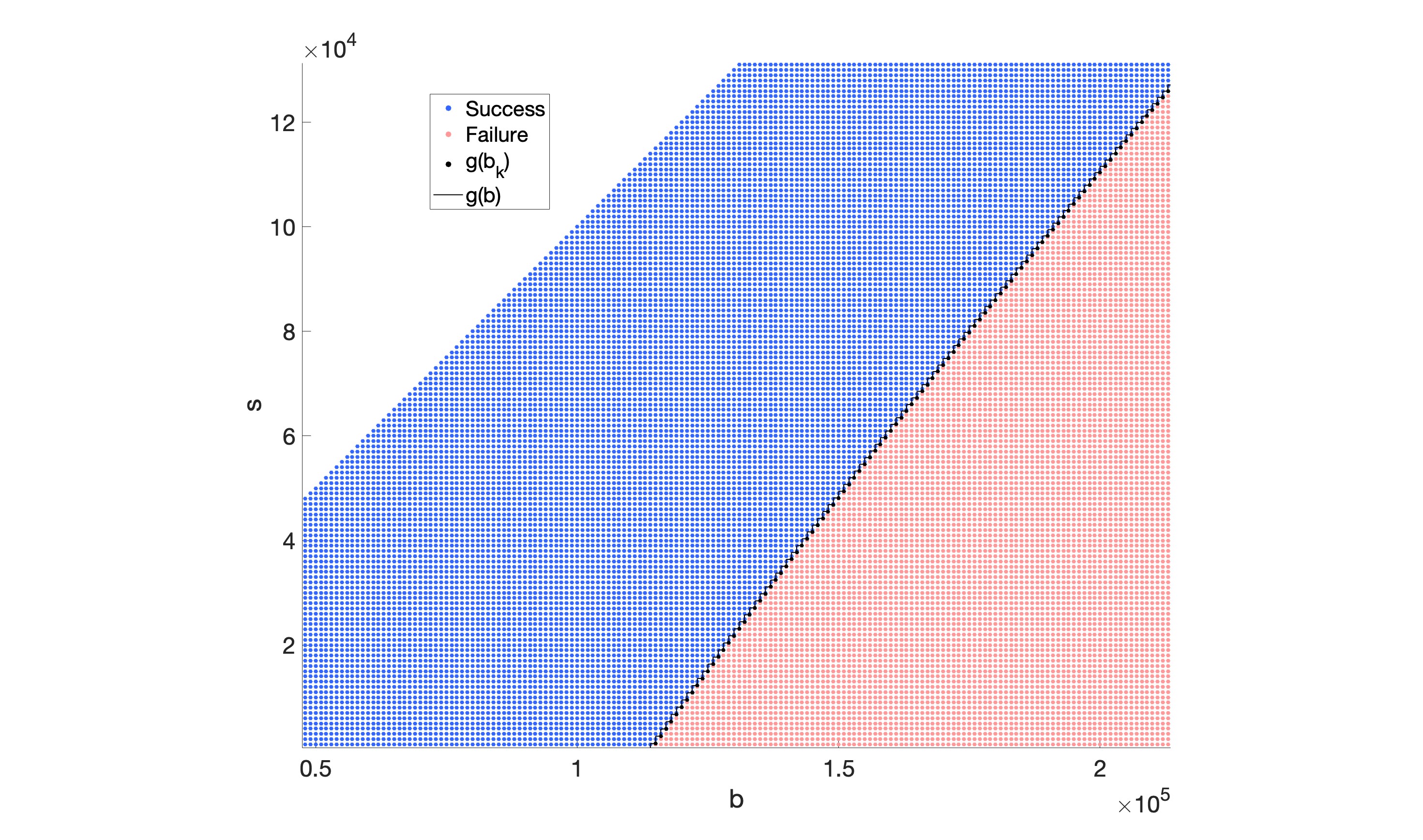}
    \caption{Numerical simulation results shown over the interval $b \in [0.5 \cdot 10^5, 2\cdot 10^5]$. Initial conditions from which the trajectories converge to the origin are marked in blue, initial conditions from which the trajectories do not converge to the origin are marked in red. The values $g(b_k)$ of $g(b)$ computed at the grid points $b_k$ are marked with black circles; the “staircase” numerical approximation of the function $g(b)$, obtained by keeping the value of the function constant and equal to $g(b_k)$ for all $b \in (b_{k-1}, b_k]$, is shown in black. All the red points $(b_k,s_k)$ lie below $g(b)$, i.e. $s_k < g(b_k)$, and all blue points $(b_k,s_k)$ lie above $g(b)$, i.e. $s_k > g(b_k)$, and hence belong to our approximated success set.}
    \label{fig:gb5000vsGrid10000_Reducedx2}
\end{sidewaysfigure}
\begin{sidewaysfigure}
        \centering
        \includegraphics[width=1\textwidth]{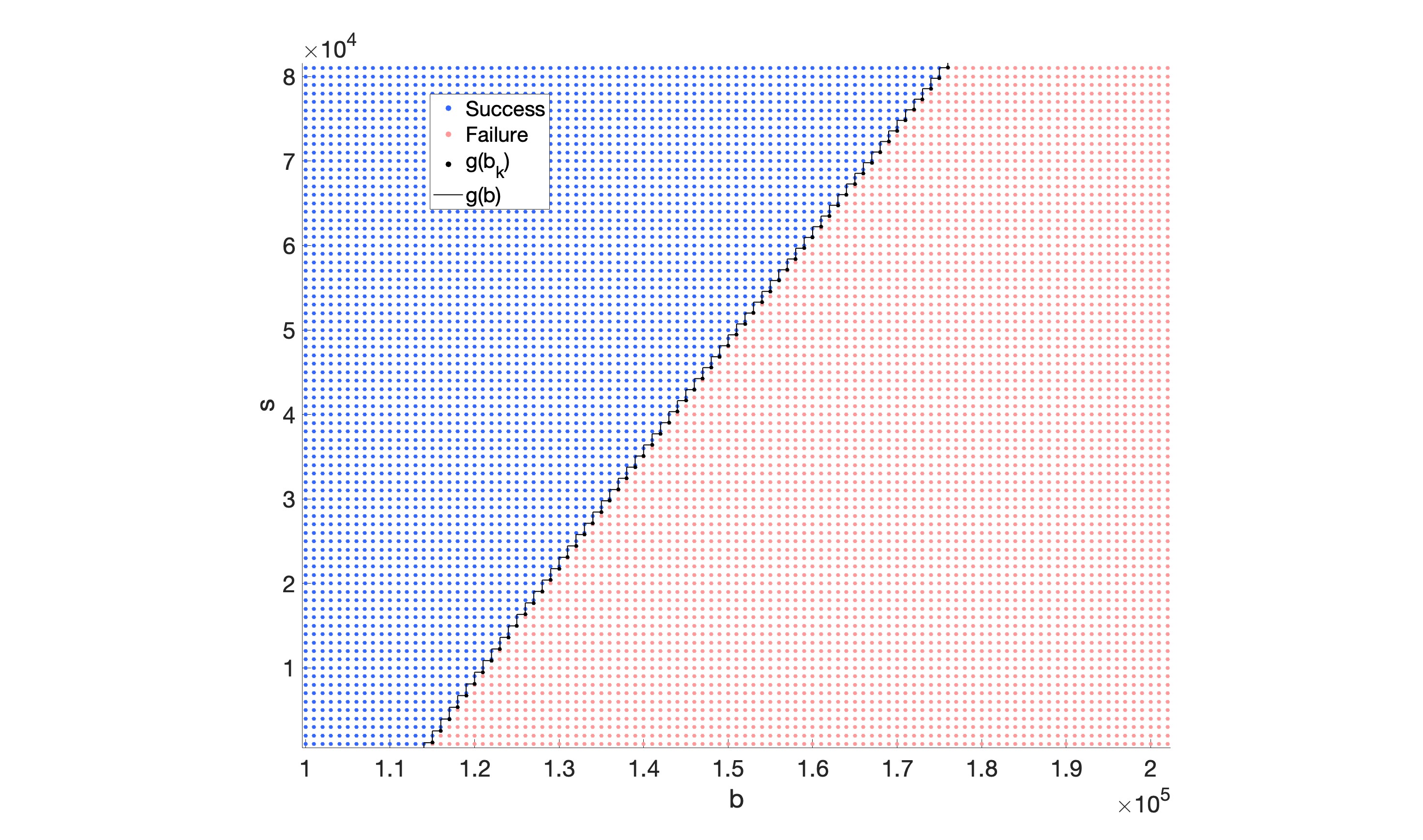}
    \caption{Numerical simulation results shown over the interval $b \in [1\cdot 10^5, 2\cdot 10^5]$. Initial conditions from which the trajectories converge to the origin are marked in blue, initial conditions from which the trajectories do not converge to the origin are marked in red. The values $g(b_k)$ of $g(b)$ computed at the grid points $b_k$ are marked with black circles; the “staircase” numerical approximation of the function $g(b)$, obtained by keeping the value of the function constant and equal to $g(b_k)$ for all $b \in (b_{k-1}, b_k]$, is shown in black. All the red points $(b_k,s_k)$ lie below $g(b)$, i.e. $s_k < g(b_k)$, and all blue points $(b_k,s_k)$ lie above $g(b)$, i.e. $s_k > g(b_k)$, and hence belong to our approximated success set.}
    \label{fig:gb5000vsGrid10000_Reducedx3}
\end{sidewaysfigure}
\begin{sidewaysfigure}
    \centering
        \centering
        \includegraphics[width=1\textwidth]{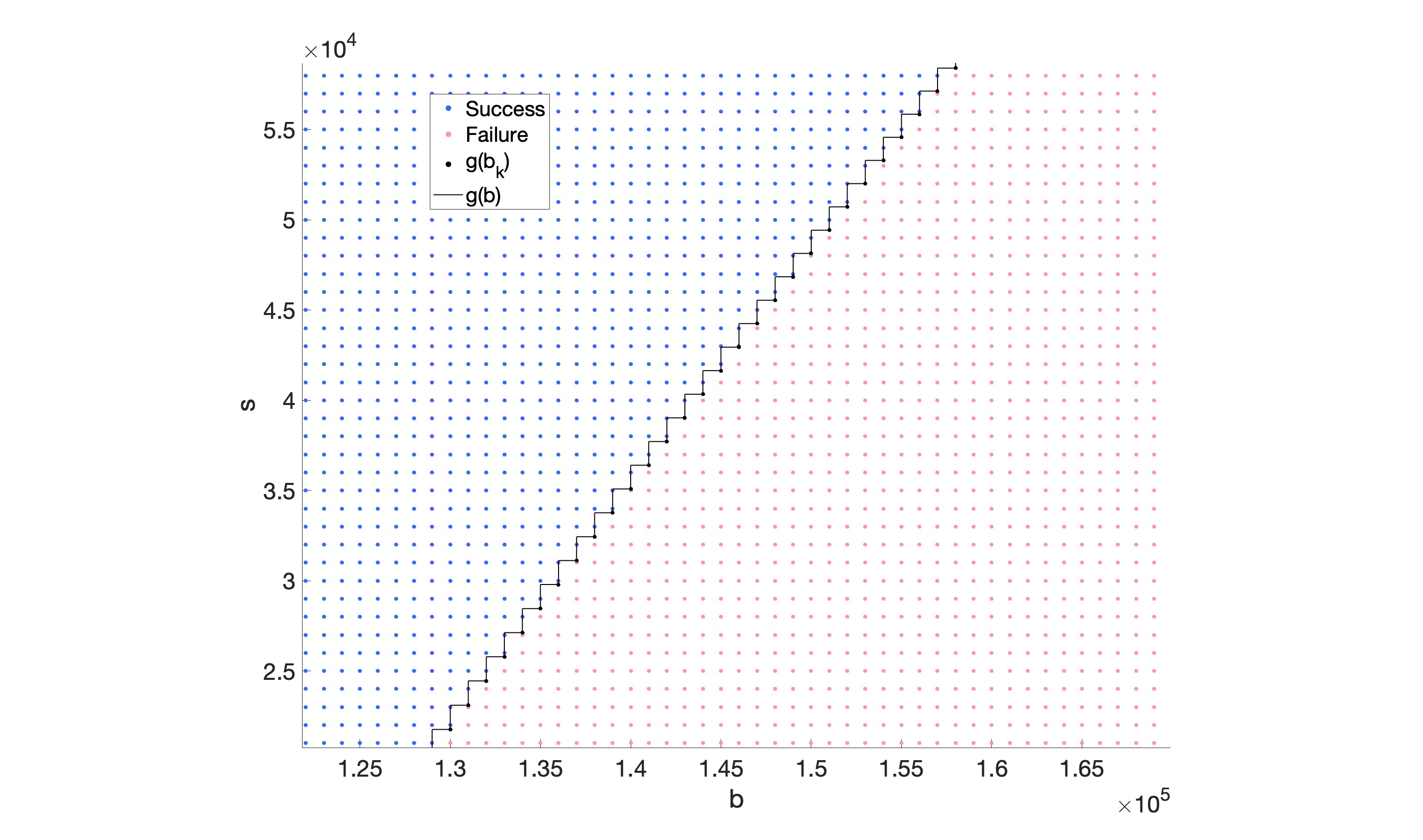}
    \caption{Numerical simulation results shown over the interval $b \in [1.25\cdot 10^5, 1.65\cdot 10^5]$. Initial conditions from which the trajectories converge to the origin are marked in blue, initial conditions from which the trajectories do not converge to the origin are marked in red. The values $g(b_k)$ of $g(b)$ computed at the grid points $b_k$ are marked with black circles; the “staircase” numerical approximation of the function $g(b)$, obtained by keeping the value of the function constant and equal to $g(b_k)$ for all $b \in (b_{k-1}, b_k]$, is shown in black. All the red points $(b_k,s_k)$ lie below $g(b)$, i.e. $s_k < g(b_k)$, and all blue points $(b_k,s_k)$ lie above $g(b)$, i.e. $s_k > g(b_k)$, and hence belong to our approximated success set.}
    \label{fig:gb5000vsGrid10000_Reducedx4}
\end{sidewaysfigure}
%

\begin{figure}[p]
    \centering
        \centering
        \includegraphics[width=0.9\textwidth]{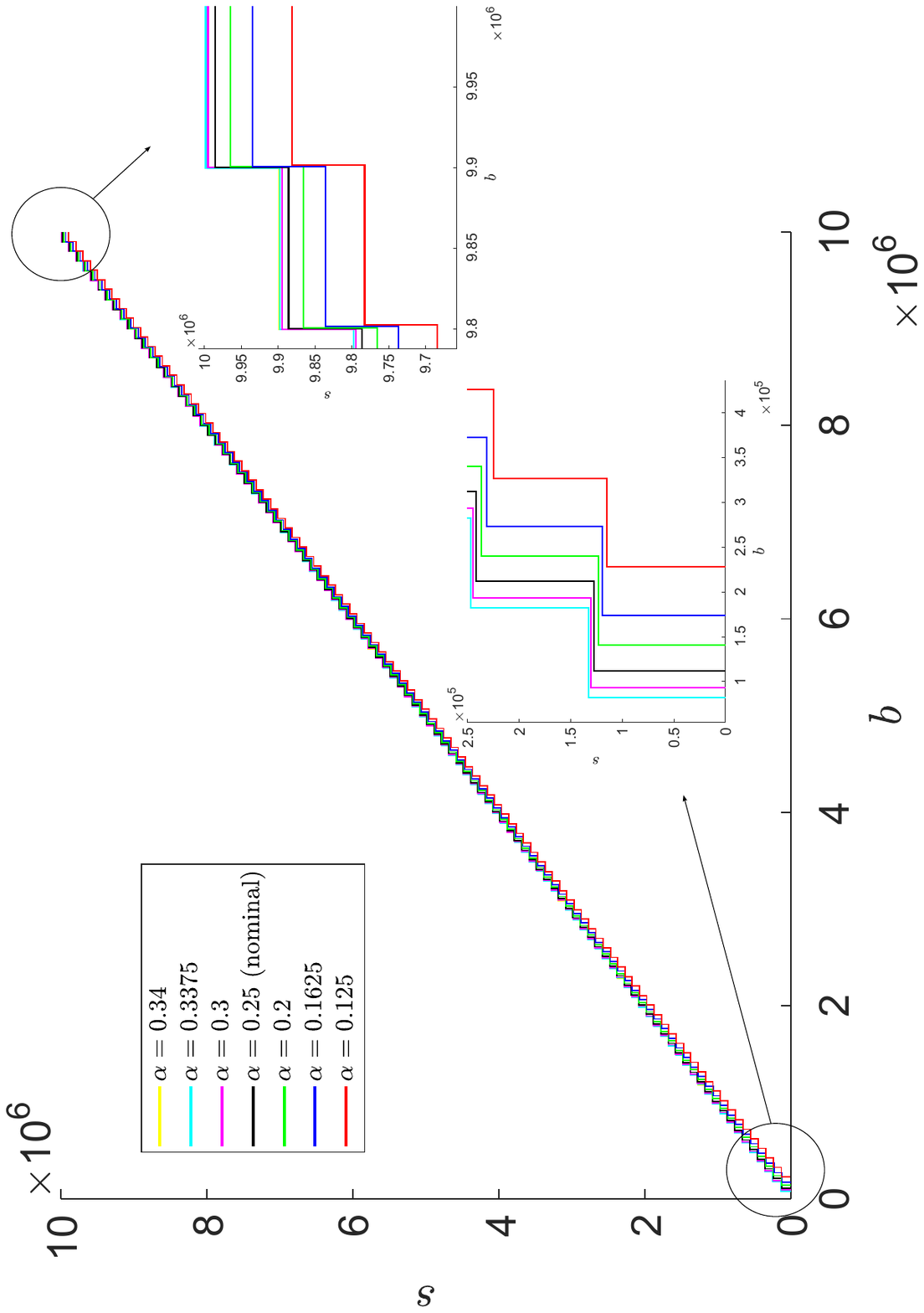}
    \caption{Numerical analysis of the sensitivity to changes in the parameter $\alpha$, bacterial net growth rate. In the insets, the yellow curve ($\alpha= 0.34$) lies behind the light blue one ($\alpha=0.3375$). Increasing $\alpha$ decreases $b_1^*$ and shrinks the success set, as expected: a larger bacterial growth rate makes clearance more difficult to achieve. $I_1 = 1.4625\cdot 10^5$ and $I_2 = 1.1696\cdot 10^5$.
}
    \label{fig:changingalpha}
\end{figure}

\begin{figure}[p]
    \centering
        \centering
        \includegraphics[width=0.9\textwidth]{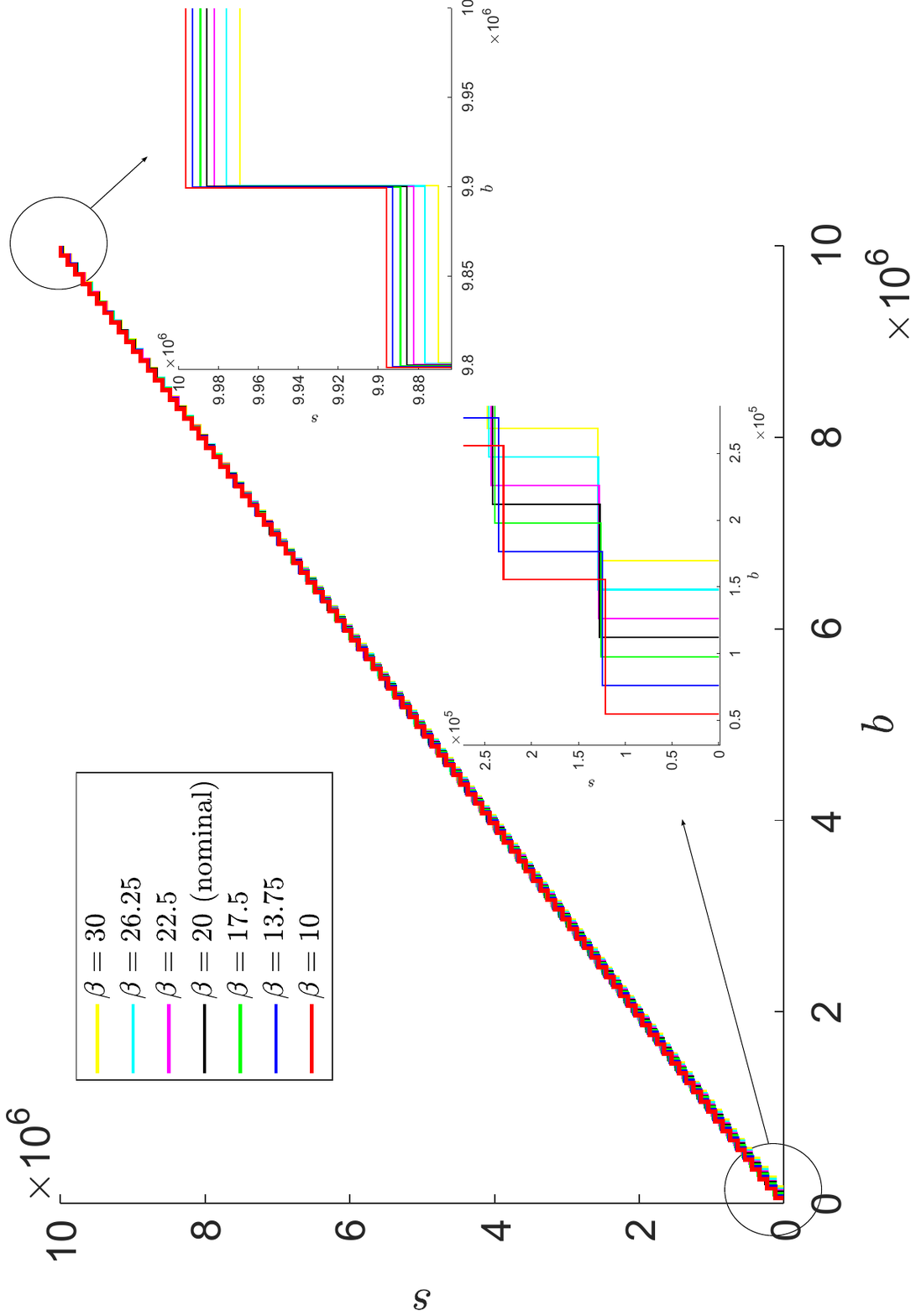}
    \caption{Numerical analysis of the sensitivity to changes in the parameter  $\beta$, maximum killing rate of the immune system. Increasing $\beta$ increases $b_1^*$ and expands the success set, as expected: a stronger immune system helps clear the infection. $I_1 = 1.1463\cdot 10^5$ and $I_2 = 4.9108\cdot 10^4$.}
    \label{fig:changingbeta}
\end{figure}

\begin{figure}[p]
    \centering
        \centering
        \includegraphics[width=0.9\textwidth]{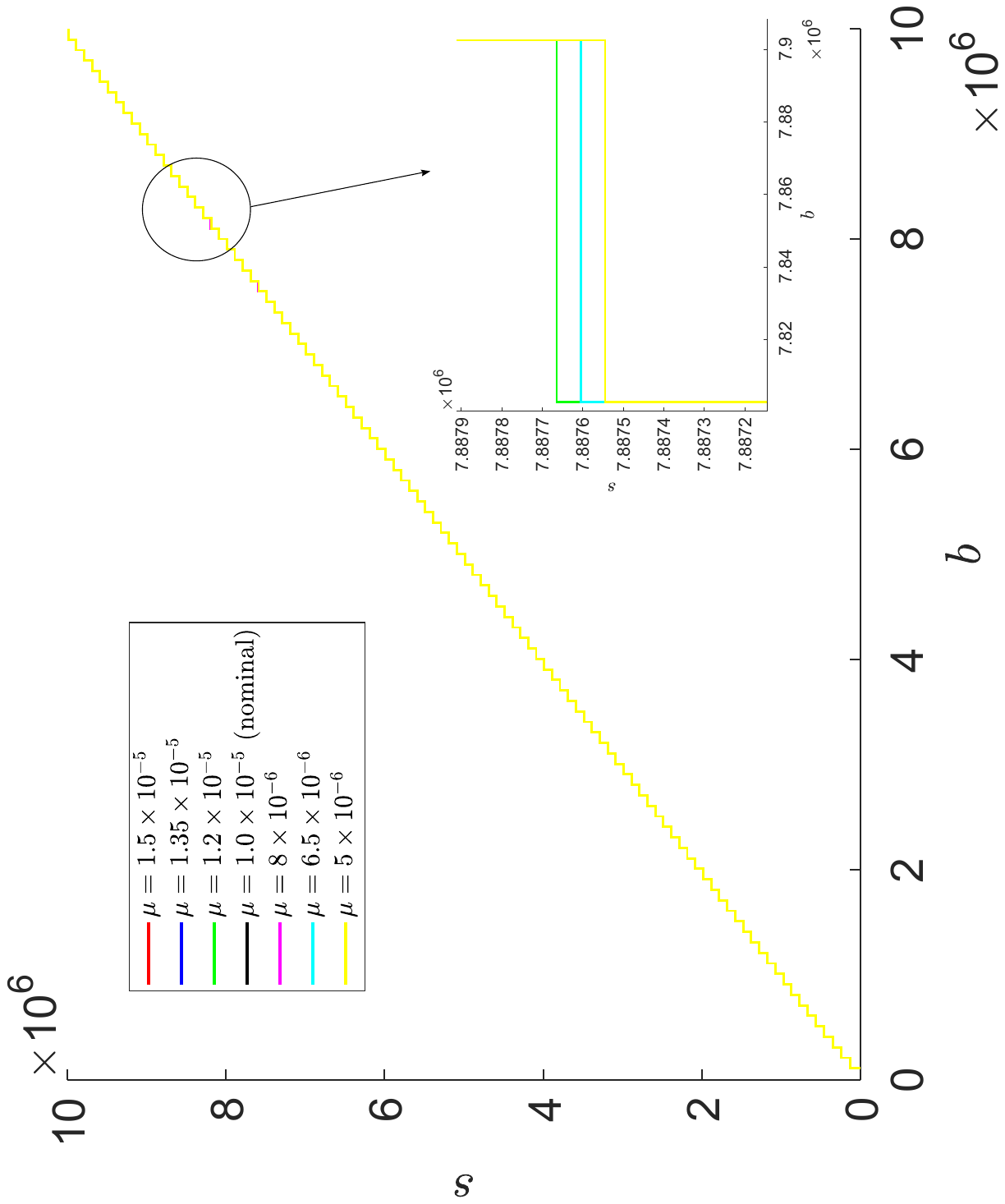}
    \caption{Numerical analysis of the sensitivity to changes in the parameter  $\mu$, mutation rate.  In the inset, the black and the purple curve lie behind the light blue one, and the blue and red curves lies behind green one. Increasing $\mu$ shrinks the success set, as expected: more mutations lead to more resistant bacteria and thus make clearance more difficult to achieve. $I_1 = 0$ and $I_2 = 2.2888\cdot 10^2$.}
    \label{fig:changingmu}
\end{figure}

\begin{figure}[p]
    \centering
        \centering
        \includegraphics[width=0.9\textwidth]{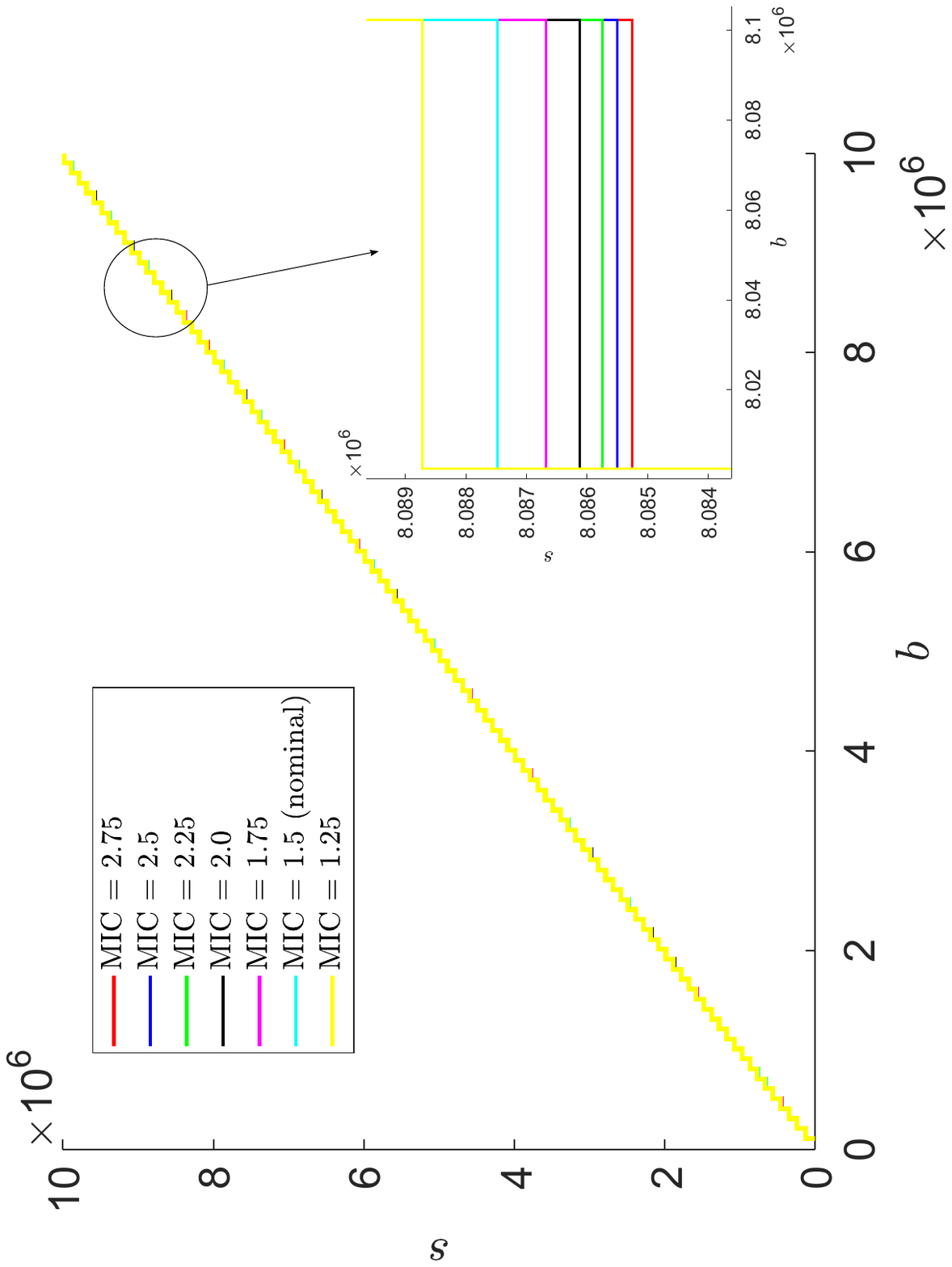}
    \caption{Numerical analysis of the sensitivity to changes in the parameter $u_M$, the antibiotic Minimum Inhibitory Concentration (MIC). With all other parameters held fixed, increasing $u_M$ expands the success set, because it leads to the use of larger antibiotic concentrations. $I_1 = 0$ and $I_2 = 3.7384\cdot 10^3$.}
    \label{fig:changingMIC}
\end{figure}

\begin{figure}[p]
    \centering
        \centering
        \includegraphics[width=0.9\textwidth]{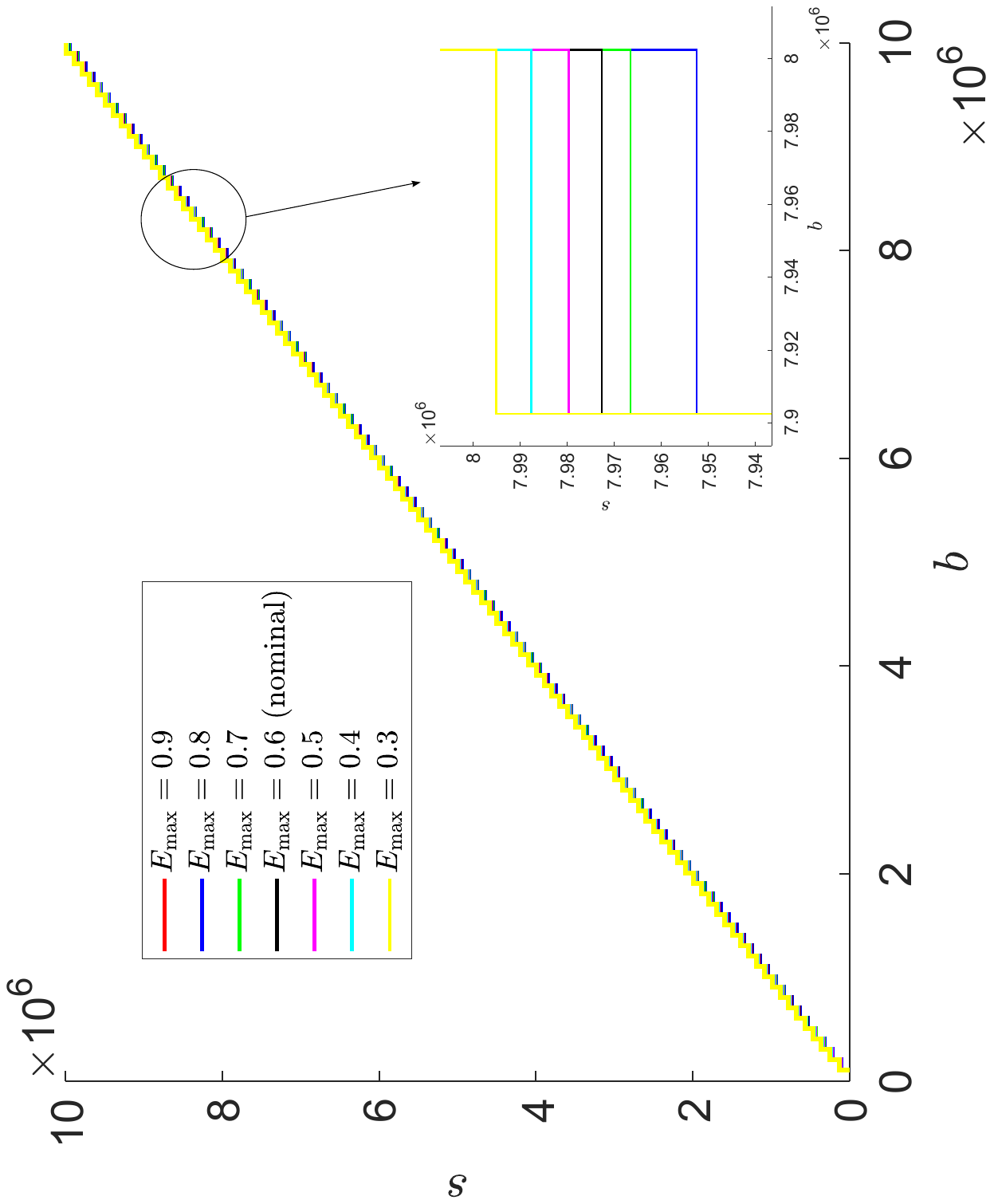}
    \caption{Numerical analysis of the sensitivity to changes in the parameter $E_{\max}$, maximum antibiotic killing rate. In the inset, the red curve lies behind the blue one. Increasing $E_{\max}$ expands the success set, as expected: a more effective antibiotic helps clear the infection. $I_1 = 0$ and $I_2 = 4.4231\cdot 10^4$.}
    \label{fig:changingEmax}
\end{figure}

\begin{figure}[p]
    \centering
        \centering
        \includegraphics[width=0.9\textwidth]{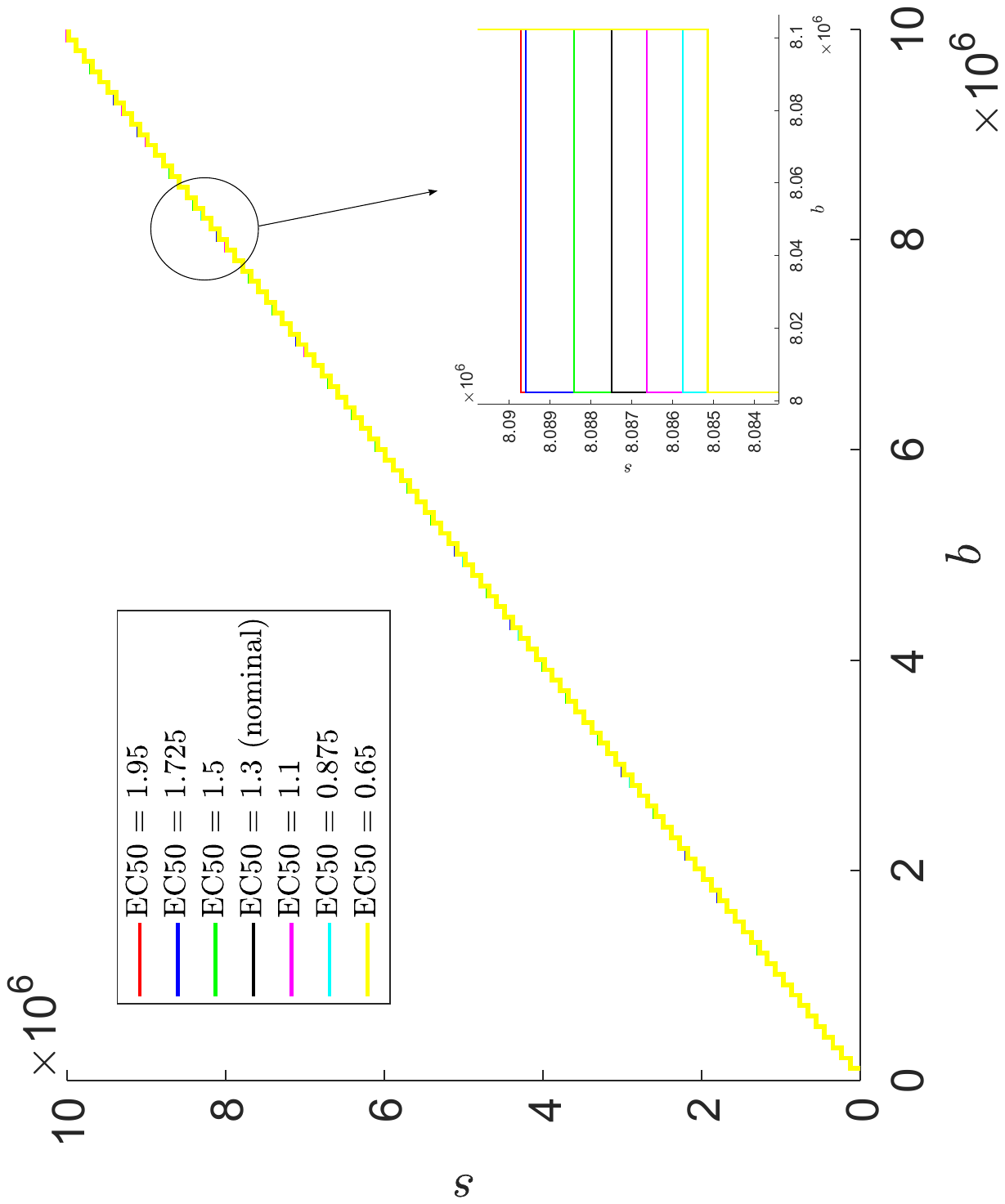}
    \caption{Numerical analysis of the sensitivity to changes in the parameter $EC_{50}$, which is the antibiotic concentration required to yield half-maximal killing rate. As expected, increasing $EC_{50}$ shrinks the success set, because it is associated with a reduction in the efficacy of the antibiotic. $I_1 = 0$ and $I_2 = 4.9591 \cdot 10^3$.}
    \label{fig:changingEC50}
\end{figure}

\begin{figure}[p]
    \centering
        \centering
        \includegraphics[width=0.9\textwidth]{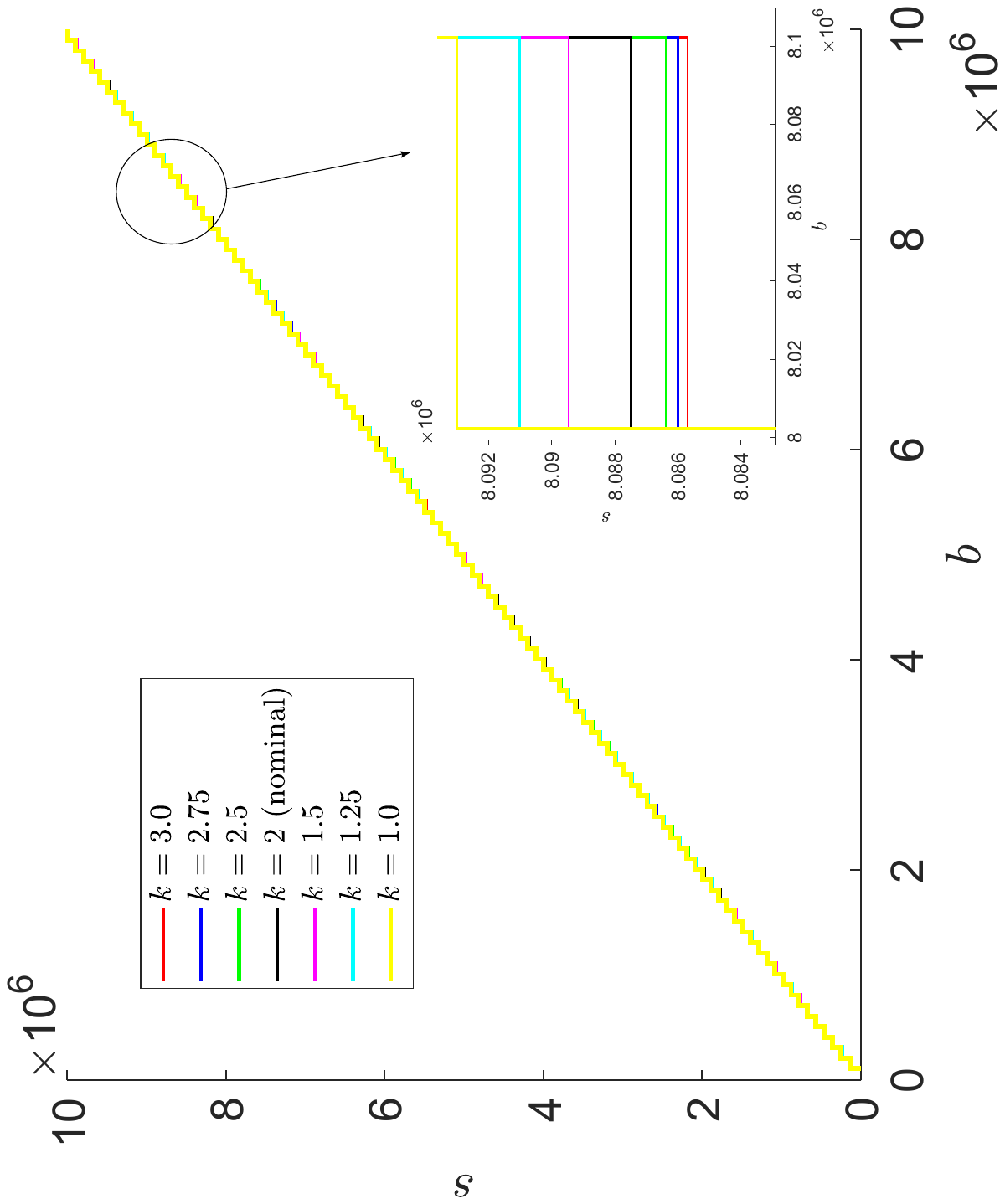}
    \caption{Numerical analysis of the sensitivity to changes in the parameter  $k$, Hill coefficient of the function providing the antibiotic killing rate depending on its concentration. Increasing $k$ expands the success set, because it is associated with an increased the efficacy of the antibiotic at the provided concentrations (above $u_M>EC_{50}$). $I_1 = 0$ and $I_2 = 7.4364 \cdot 10^3$.}
    \label{fig:changingk}
\end{figure}

\begin{figure}[p]
    \centering
        \centering
        \includegraphics[width=0.9\textwidth]{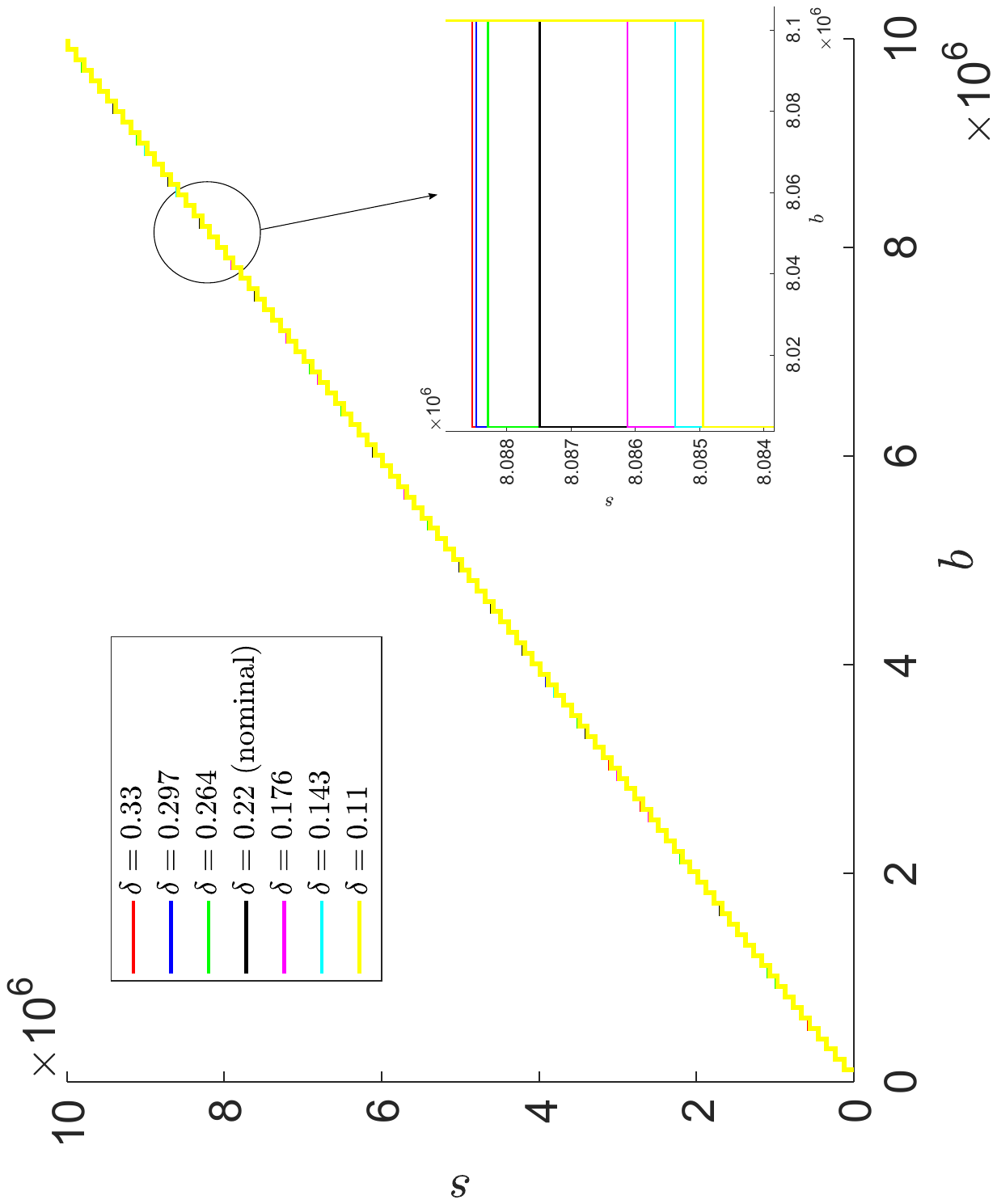}
    \caption{Numerical analysis of the sensitivity to changes in the parameter $\delta$, antibiotic decay rate. Increasing $\delta$ shrinks the success set, because it reduces the average antibiotic concentration. $I_1 = 0$ and $I_2 = 3.6264 \cdot 10^3$.}
    \label{fig:changingdelta}
\end{figure}

\begin{figure}[p]
    \centering
        \centering
        \includegraphics[width=0.9\textwidth]{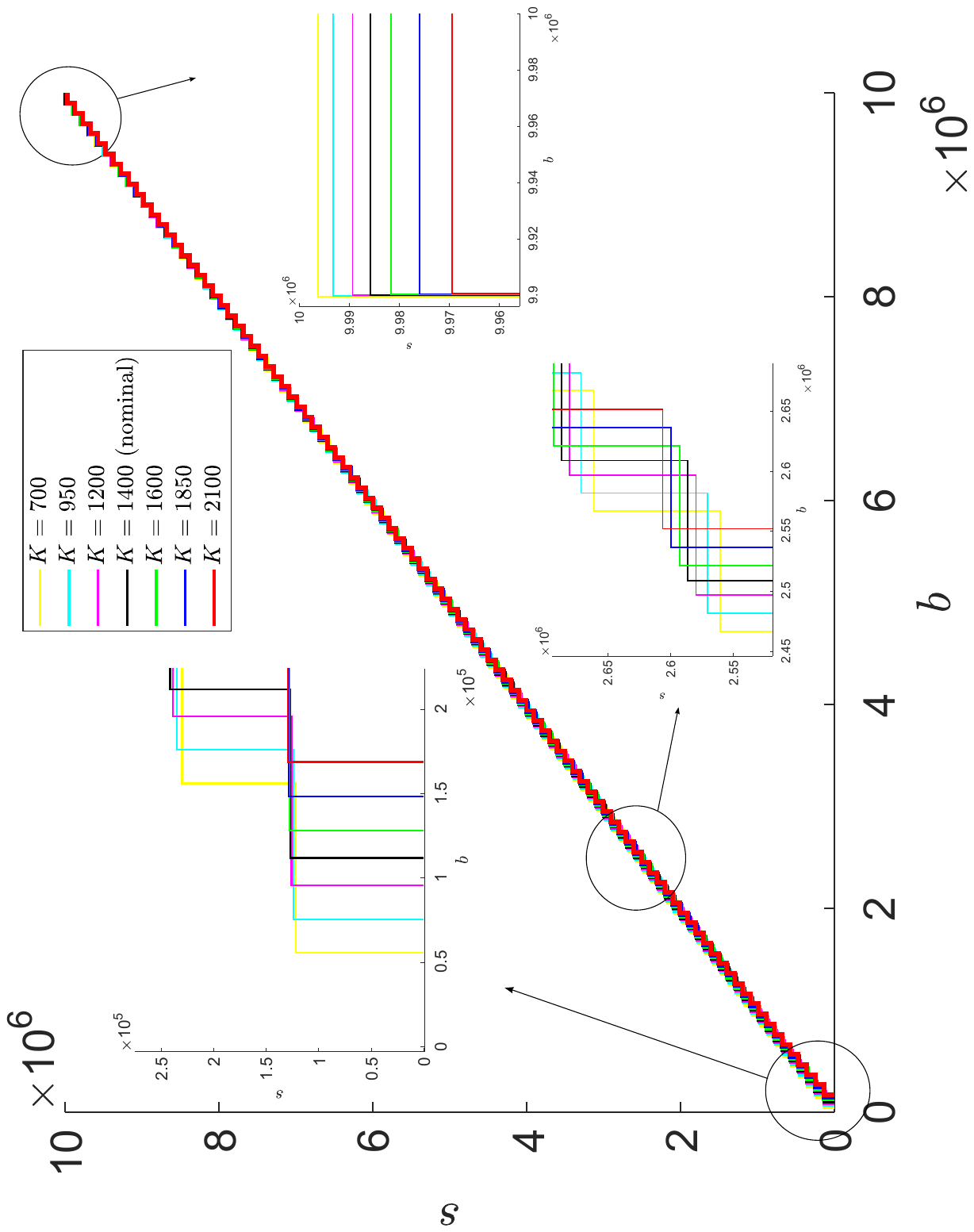}
    \caption{Numerical analysis of the sensitivity to changes in the parameter  $K$, total bacterial load yielding half-maximal killing rate due to the immune system. Increasing $K$ increases $b_1^*$ and expands the success set, as it is associated with a stronger immune system response. $I_1 = 1.1317 \cdot 10^5$ and $I_2 = 4.8308 \cdot 10^4$.}
    \label{fig:changingK}
\end{figure}

\begin{figure}[p]
    \centering
        \centering
        \includegraphics[width=0.9\textwidth]{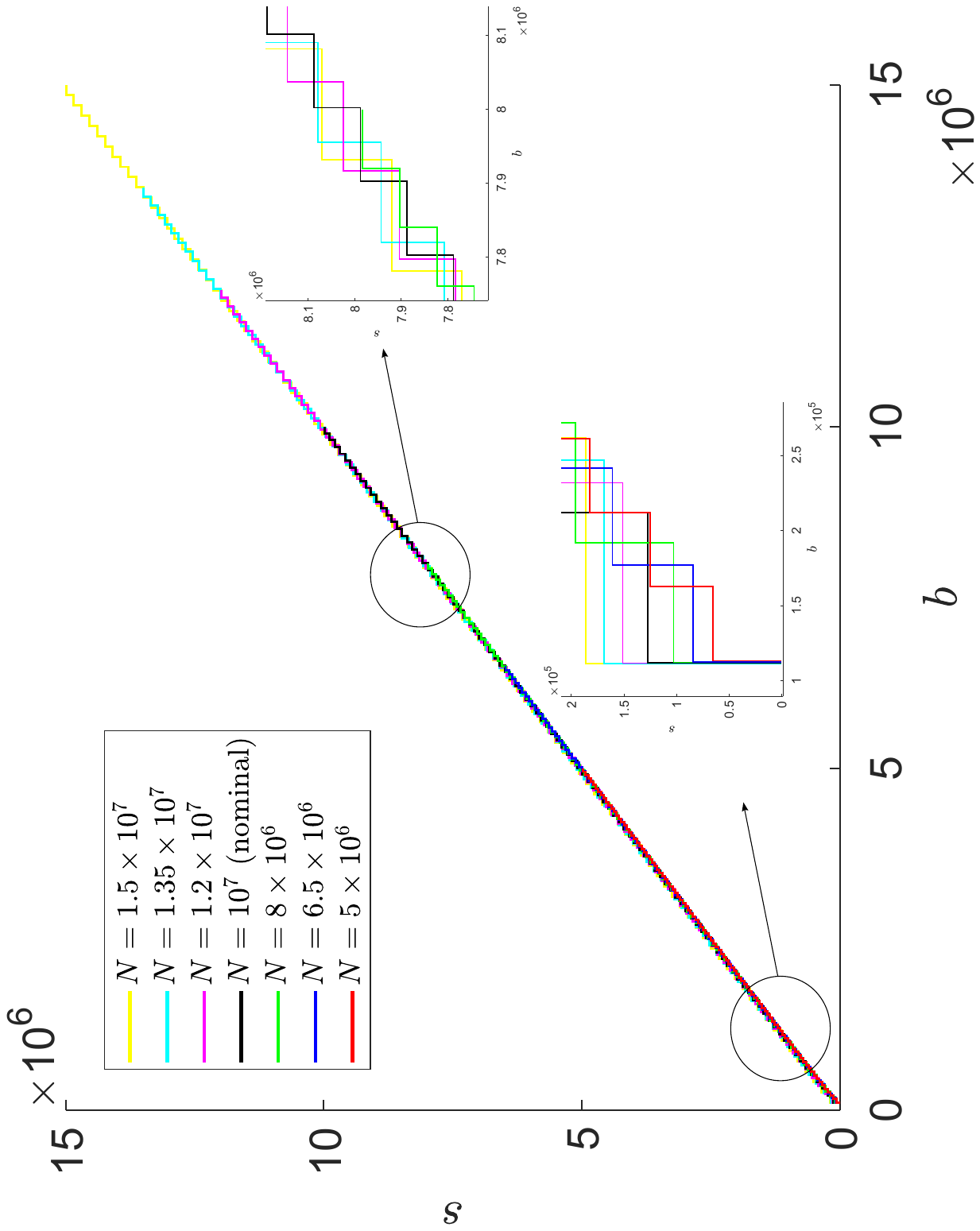}
    \caption{Numerical analysis of the sensitivity to changes in the parameter $N$. Increasing $N$ increases $b_1^*$. Note that the function $g(b)$ can only be computed for $b \in [b_1^*, N]$. $I_1 = 1.7574 \cdot 10^3$.}
    \label{fig:changingN}
\end{figure}

\begin{figure}[h!]
    \centering
        \centering
        \includegraphics[scale = 0.5]{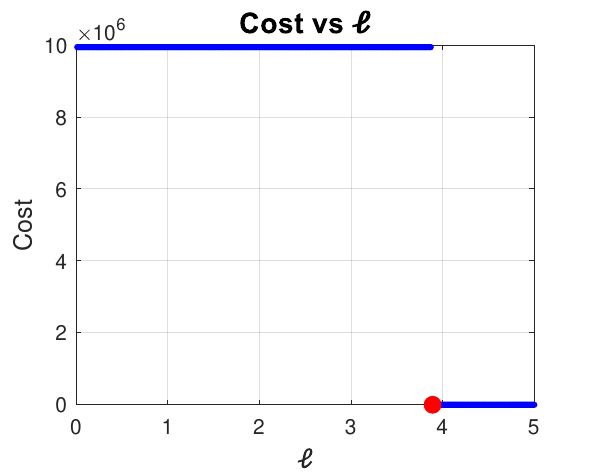}
        \includegraphics[scale = 0.4]{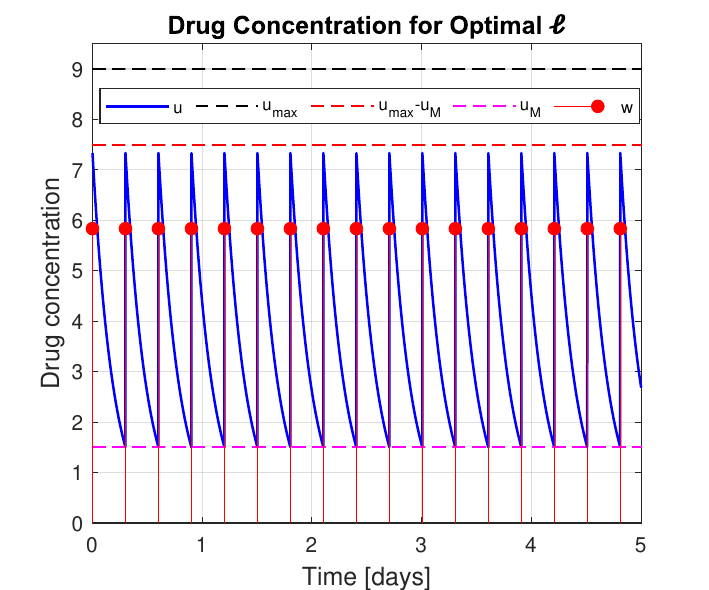}
        \includegraphics[scale = 0.5]{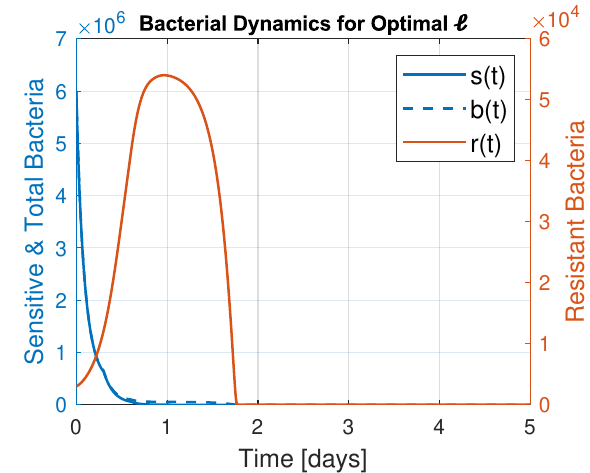}
    \caption{Optimal control problem \eqref{eq:optimalcontrolV_SP} with $\zeta_2=1$ and $\zeta_1=\zeta_3=0$ and initial condition IC1 in the success set ($s_0= 6\cdot 10^6$ and $r_0 = 3\cdot10^3$): computation of the cost $J$ as a function of $\ell$ (left) and then, with the selected optimal treatment, time evolution of the drug concentration and administered drug doses (middle), and time evolution of the state variables $b$, $s$ and of $r=b-s$ (right). The parameters are as listed in Table~\ref{tb:parameters}, apart from $\beta = 10$. The minimum $J=0$ is not unique; to avoid misuse of antibiotic, we select the optimal solution corresponding to the smallest dose, obtained for $\ell = 3.90$. }
    \label{fig:suppl_OCP1}
\end{figure}

\begin{figure}[h!]
    \centering
        \centering
        \includegraphics[scale = 0.5]{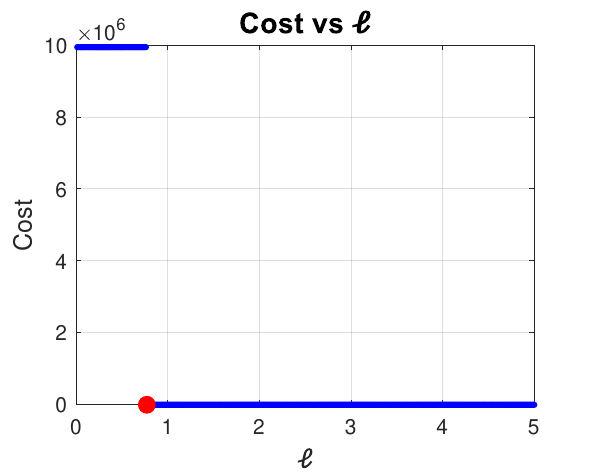}
        \includegraphics[scale = 0.4]{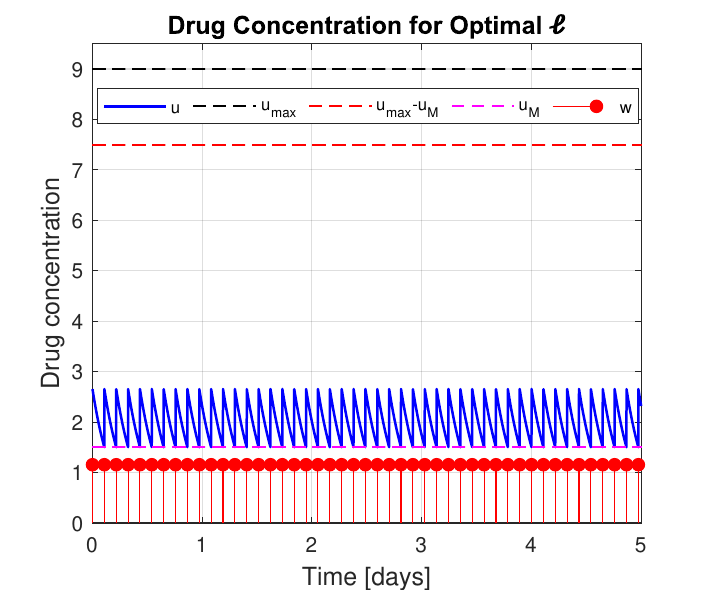}
        \includegraphics[scale = 0.5]{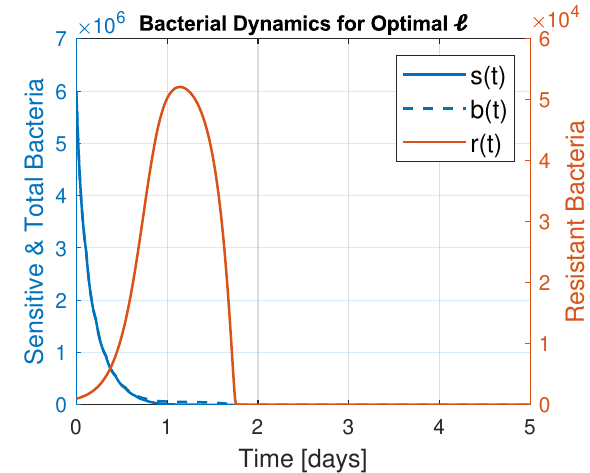}
    \caption{Optimal control problem \eqref{eq:optimalcontrolV_SP} with $\zeta_2=1$ and $\zeta_1=\zeta_3=0$ and initial condition IC2 in the success set ($s_0= 6.002\cdot 10^6$ and $r_0 = 10^3$): computation of the cost $J$ as a function of $\ell$ (left) and then, with the selected optimal treatment, time evolution of the drug concentration and administered drug doses (middle), and time evolution of the state variables $b$, $s$ and of $r=b-s$ (right). The parameters are as listed in Table~\ref{tb:parameters}, apart from $\beta = 10$. The minimum $J=0$ is not unique; to avoid misuse of antibiotic, we select the optimal solution corresponding to the smallest dose, obtained for $\ell = 0.77$. }
    \label{fig:suppl_OCP3}
\end{figure}

\begin{figure}[h!]
    \centering
        \centering
        \includegraphics[scale = 0.5]{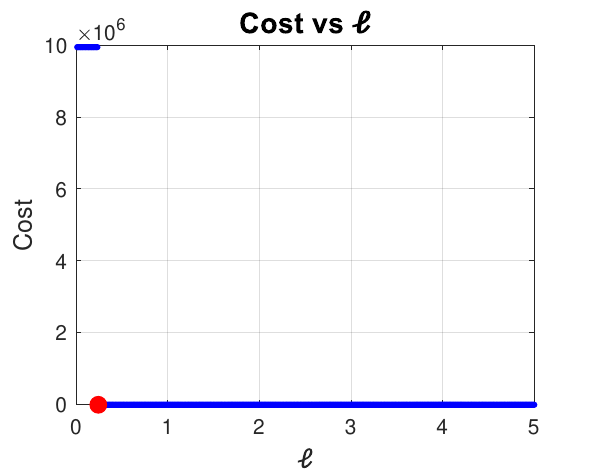}
        \includegraphics[scale = 0.4]{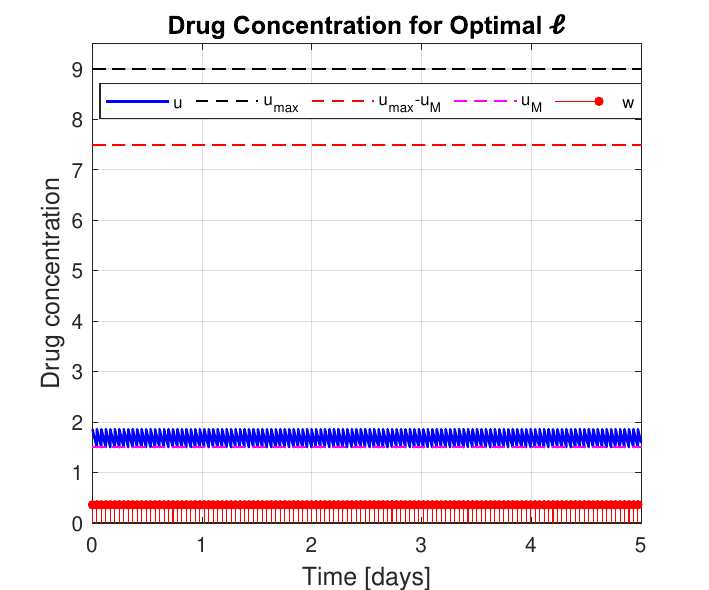}
        \includegraphics[scale = 0.5]{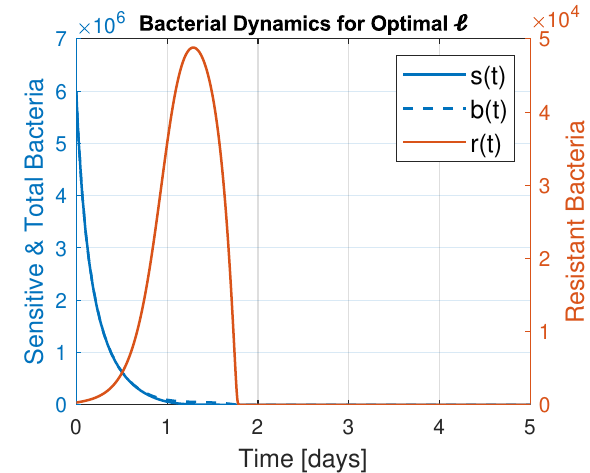}
    \caption{Optimal control problem \eqref{eq:optimalcontrolV_SP} with $\zeta_2=1$ and $\zeta_1=\zeta_3=0$ and initial condition IC3 in the success set ($s_0= 6\cdot 10^6$ and $r_0 = 3\cdot10^2$): computation of the cost $J$ as a function of $\ell$ (left) and then, with the selected optimal treatment, time evolution of the drug concentration and administered drug doses (middle), and time evolution of the state variables $b$, $s$ and of $r=b-s$ (right). The parameters are as listed in Table~\ref{tb:parameters}, apart from $\beta = 10$. The minimum $J=0$ is not unique; to avoid misuse of antibiotic, we select the optimal solution corresponding to the smallest dose, obtained for $\ell = 0.24$. }
    \label{fig:suppl_OCP2}
\end{figure}

\begin{figure}[h!]
    \centering
        \centering
        \includegraphics[scale = 0.5]{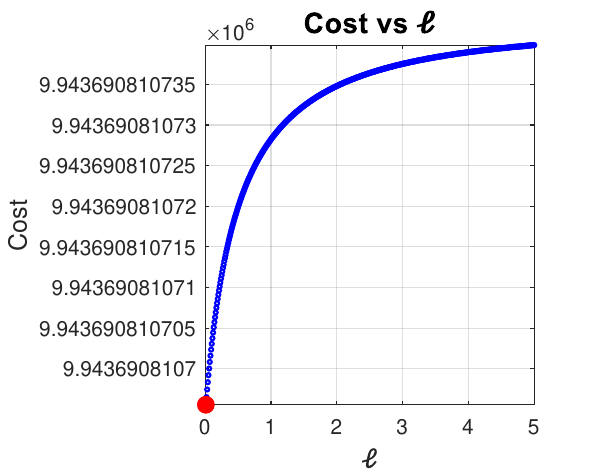}
        \includegraphics[scale = 0.4]{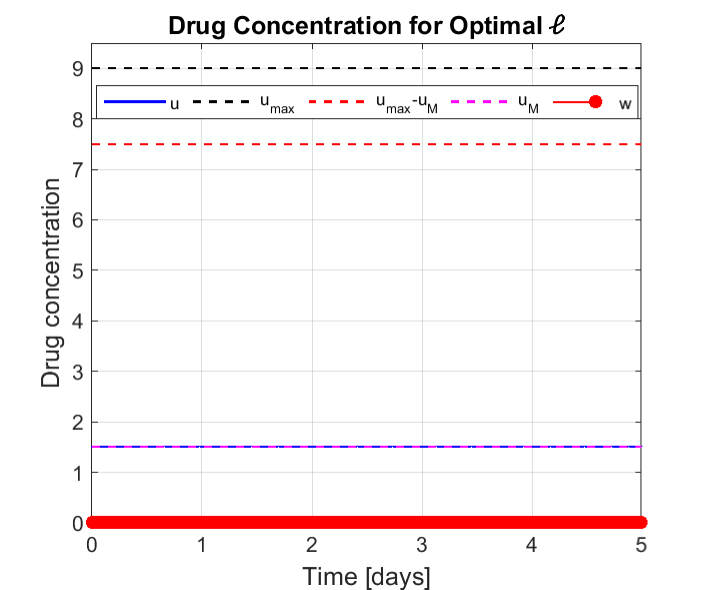}
        \includegraphics[scale = 0.5]{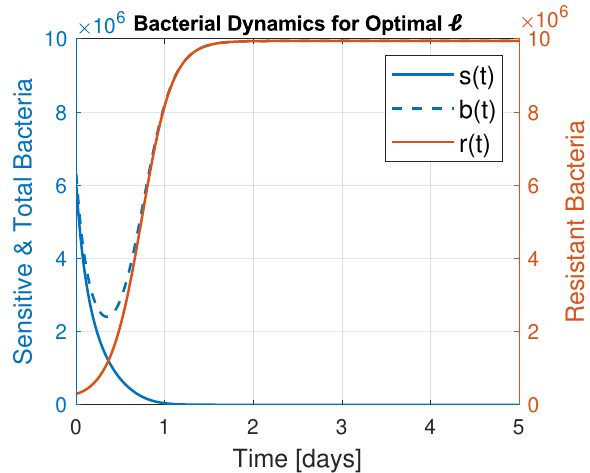}
    \caption{Optimal control problem \eqref{eq:optimalcontrolV_SP} with $\zeta_2=1$ and $\zeta_1=\zeta_3=0$ and initial condition outside the success set ($s_0= 6\cdot 10^6$ and $r_0 = 3\cdot10^5$): computation of the cost $J$ as a function of $\ell$ (left) and then, with the selected optimal treatment, time evolution of the drug concentration and administered drug doses (middle), and time evolution of the state variables $b$, $s$ and of $r=b-s$ (right). The parameters are as listed in Table~\ref{tb:parameters}, apart from $\beta = 10$. The therapy is not effective, and $b \to b_2^*$.}
    \label{fig:suppl_OCP4}
\end{figure}

\begin{figure}[h!]
    \centering
        \centering
        \includegraphics[scale = 0.5]{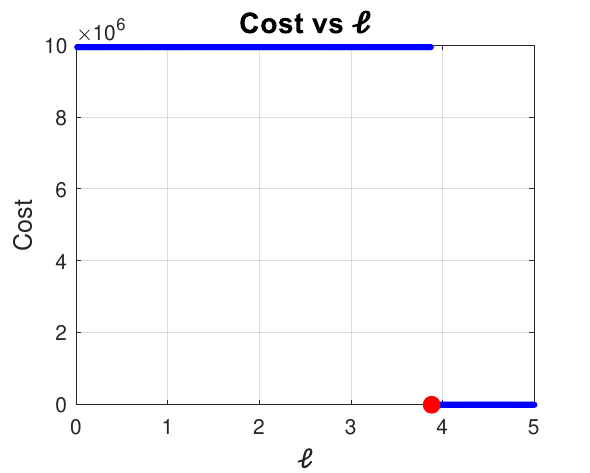}
        \includegraphics[scale = 0.4]{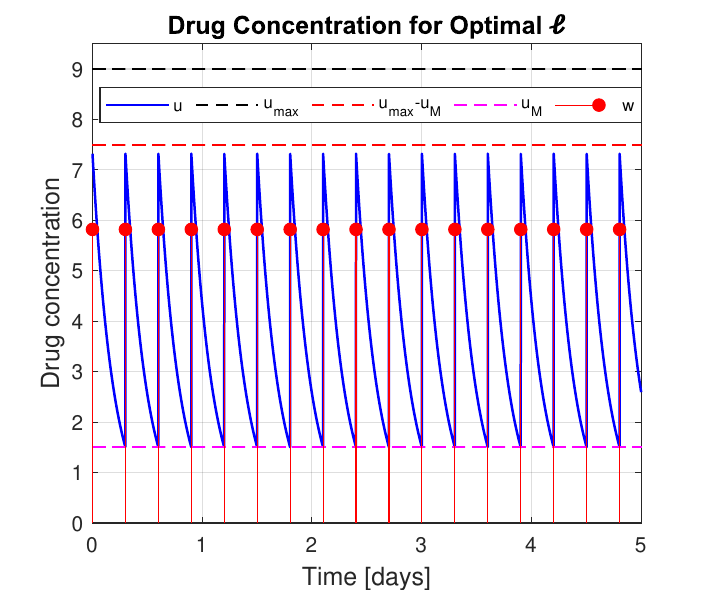}
        \includegraphics[scale = 0.5]{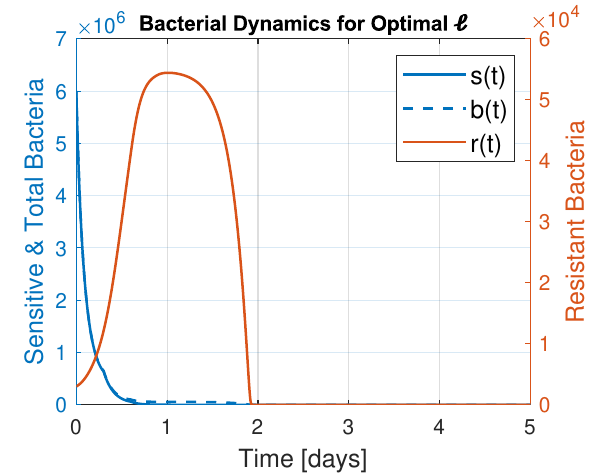}
    \caption{Optimal control problem \eqref{eq:optimalcontrolV_SP} with $\zeta_1=\zeta_2=1$ and $\zeta_3=0$ and initial condition IC1 in the success set ($s_0= 6\cdot 10^6$ and $r_0 = 3\cdot10^3$): computation of the cost $J$ as a function of $\ell$ (left) and then, with the selected optimal treatment, time evolution of the drug concentration and administered drug doses (middle), and time evolution of the state variables $b$, $s$ and of $r=b-s$ (right). The parameters are as listed in Table~\ref{tb:parameters}, apart from $\beta = 10$. The optimal solution is obtained for $\ell = 3.88$.}
    \label{fig:suppl_OCP5}
\end{figure}

\begin{figure}[h!]
    \centering
        \centering
        \includegraphics[scale = 0.5]{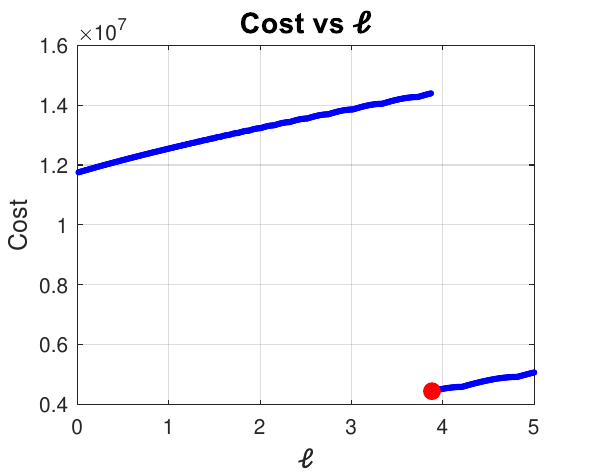}
        \includegraphics[scale = 0.4]{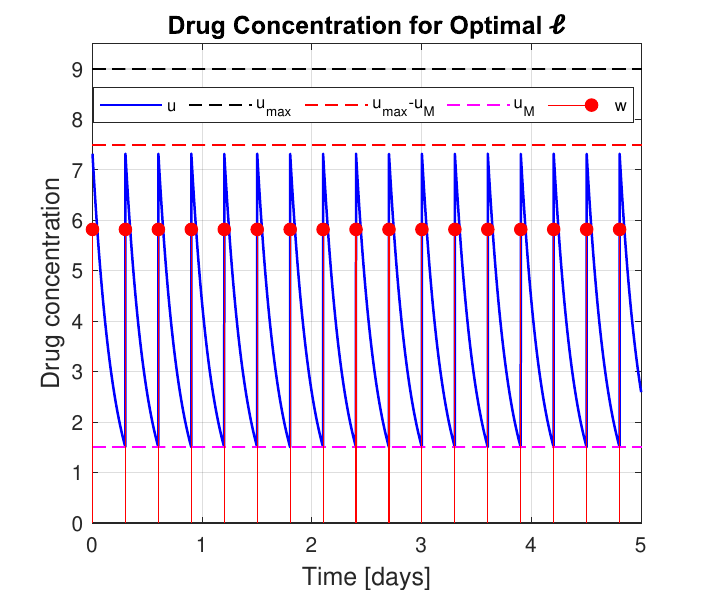}
        \includegraphics[scale = 0.5]{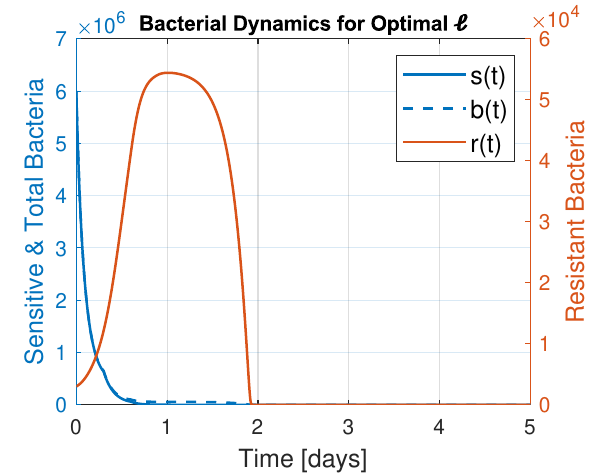}
    \caption{Optimal control problem \eqref{eq:optimalcontrolV_SP} with $\zeta_1=10^4$, $\zeta_2=1$ and $\zeta_3=0$ and initial condition IC1 in the success set ($s_0= 6\cdot 10^6$ and $r_0 = 3\cdot10^3$): computation of the cost $J$ as a function of $\ell$ (left) and then, with the selected optimal treatment, time evolution of the drug concentration and administered drug doses (middle), and time evolution of the state variables $b$, $s$ and of $r=b-s$ (right). The parameters are as listed in Table~\ref{tb:parameters}, apart from $\beta = 10$. The optimal solution is obtained for $\ell = 3.89$.}
    \label{fig:suppl_OCP6}
\end{figure}

\begin{figure}[h!]
    \centering
        \centering
        \includegraphics[scale = 0.5]{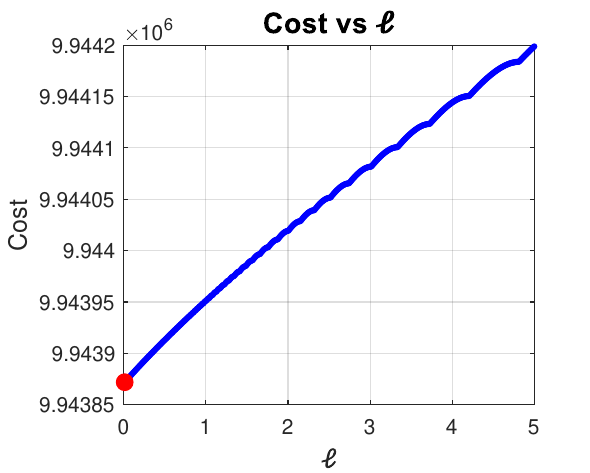}
        \includegraphics[scale = 0.4]{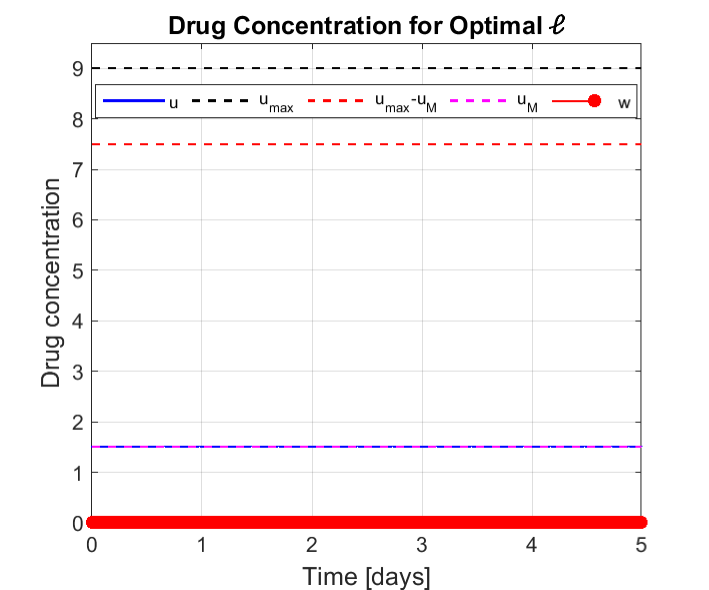}
        \includegraphics[scale = 0.5]{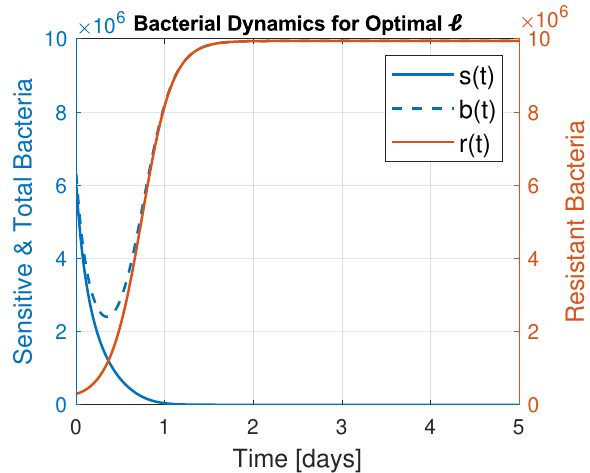}
    \caption{Optimal control problem \eqref{eq:optimalcontrolV_SP} with $\zeta_1=\zeta_2=1$ and $\zeta_3=0$ and initial condition outside the success set ($s_0= 6\cdot 10^6$ and $r_0 = 3\cdot10^5$): computation of the cost $J$ as a function of $\ell$ (left) and then, with the selected optimal treatment, time evolution of the drug concentration and administered drug doses (middle), and time evolution of the state variables $b$, $s$ and of $r=b-s$ (right). The parameters are as listed in Table~\ref{tb:parameters}, apart from $\beta = 10$. The therapy is not effective, and $b \to b_2^*$.}
    \label{fig:suppl_OCP11}
\end{figure}

\begin{figure}[h!]
    \centering
        \centering
        \includegraphics[scale = 0.5]{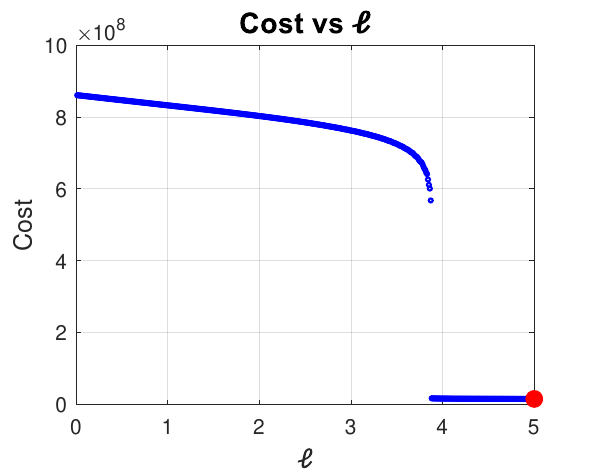}
        \includegraphics[scale = 0.4]{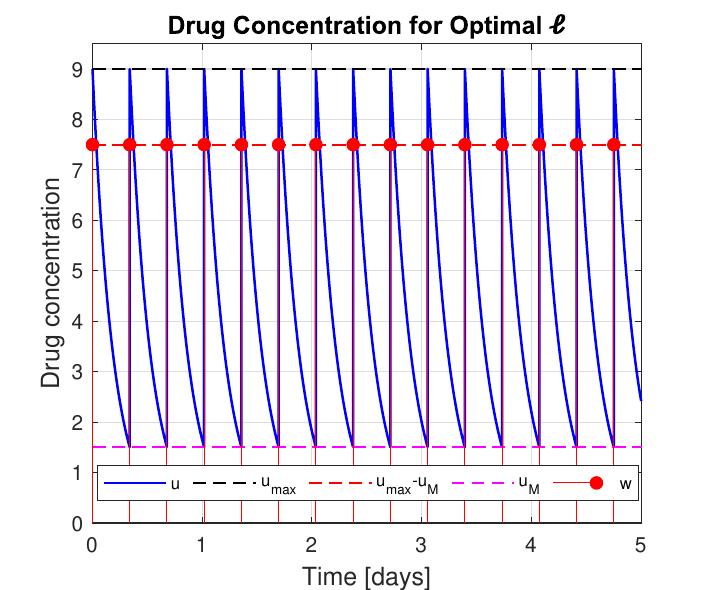}
        \includegraphics[scale = 0.5]{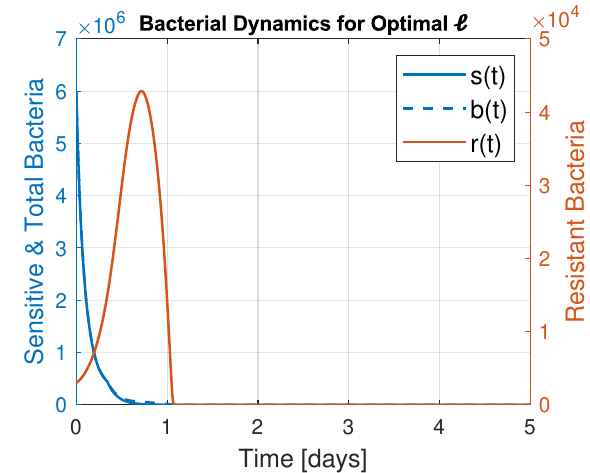}
    \caption{Optimal control problem \eqref{eq:optimalcontrolV_SP} with $\zeta_1=\zeta_2=\zeta_3=1$ and initial condition IC1 in the success set ($s_0= 6\cdot 10^6$ and $r_0 = 3\cdot10^3$): computation of the cost $J$ as a function of $\ell$ (left) and then, with the selected optimal treatment, time evolution of the drug concentration and administered drug doses (middle), and time evolution of the state variables $b$, $s$ and of $r=b-s$ (right). The parameters are as listed in Table~\ref{tb:parameters}, apart from $\beta = 10$. The optimal solution is obtained for $\ell = 5$.}
    \label{fig:suppl_OCP7}
\end{figure}

\begin{figure}[h!]
    \centering
        \centering
        \includegraphics[scale = 0.5]{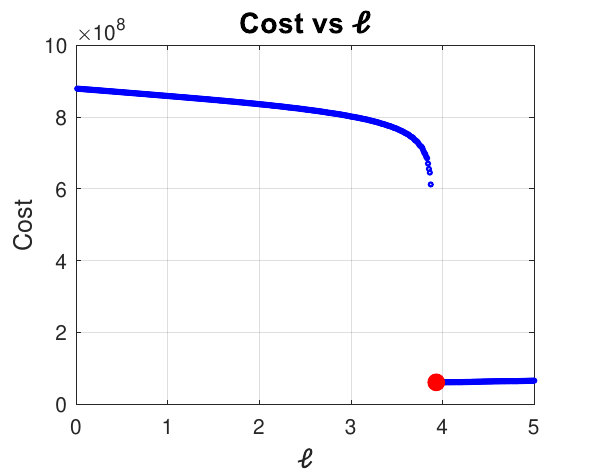}
        \includegraphics[scale = 0.4]{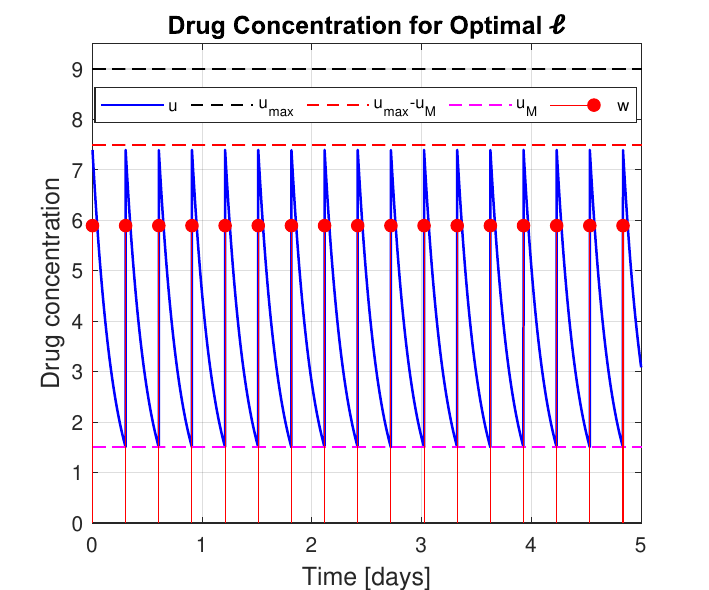}
        \includegraphics[scale = 0.5]{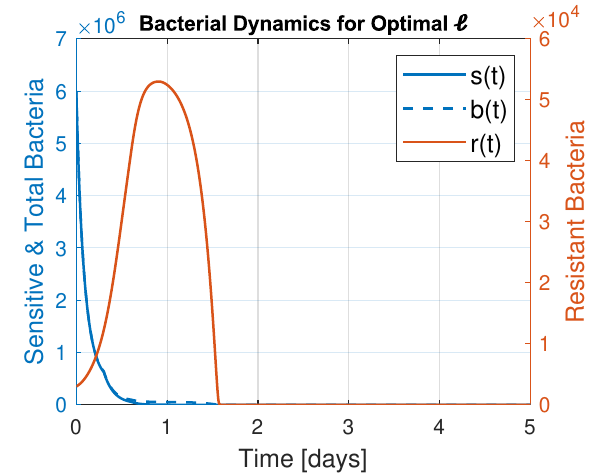}
    \caption{Optimal control problem \eqref{eq:optimalcontrolV_SP} with $\zeta_1=10^5$ and $\zeta_2=\zeta_3=1$ and initial condition IC1 in the success set ($s_0= 6\cdot 10^6$ and $r_0 = 3\cdot10^3$): computation of the cost $J$ as a function of $\ell$ (left) and then, with the selected optimal treatment, time evolution of the drug concentration and administered drug doses (middle), and time evolution of the state variables $b$, $s$ and of $r=b-s$ (right). The parameters are as listed in Table~\ref{tb:parameters}, apart from $\beta = 10$. The optimal solution is obtained for $\ell = 3.93$.}
    \label{fig:suppl_OCP8}
\end{figure}

\begin{figure}[h!]
    \centering
        \centering
        \includegraphics[scale = 0.5]{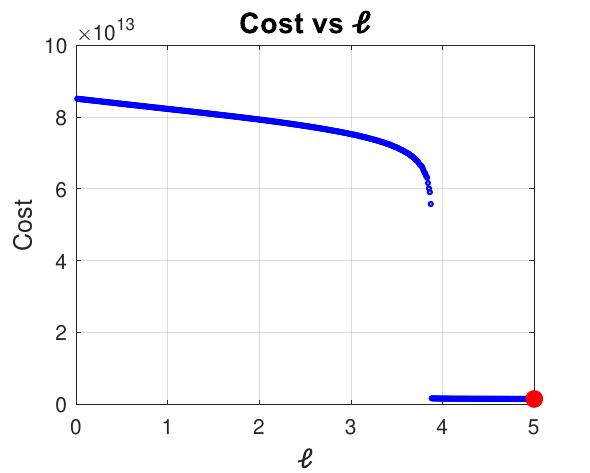}
        \includegraphics[scale = 0.4]{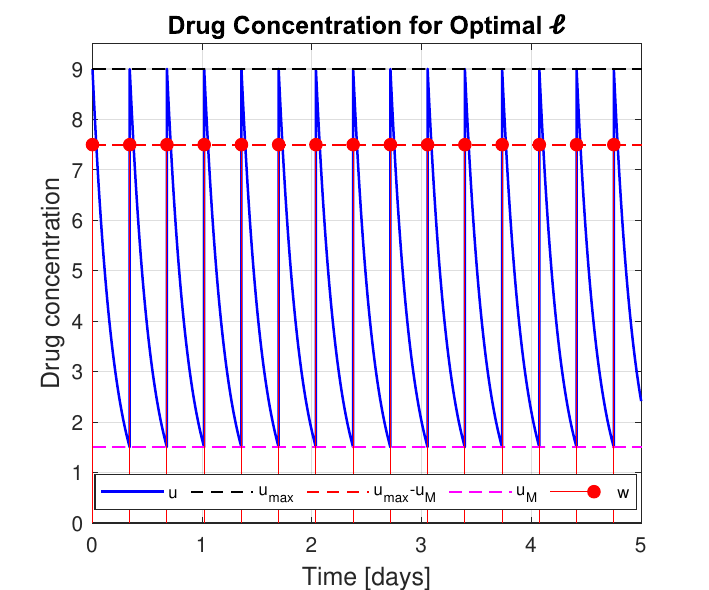}
        \includegraphics[scale = 0.5]{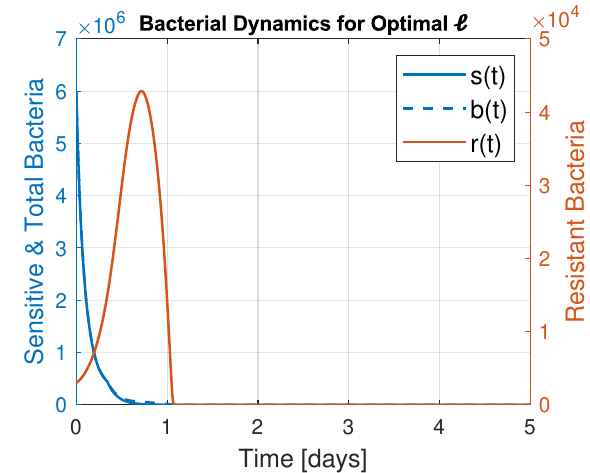}
    \caption{Optimal control problem \eqref{eq:optimalcontrolV_SP} with $\zeta_1= 10^5$, $\zeta_2 =1$ and $\zeta_3 = 10^5$ and initial condition IC1 in the success set ($s_0= 6\cdot 10^6$ and $r_0 = 3\cdot10^3$): computation of the cost $J$ as a function of $\ell$ (left) and then, with the selected optimal treatment, time evolution of the drug concentration and administered drug doses (middle), and time evolution of the state variables $b$, $s$ and of $r=b-s$ (right). The parameters are as listed in Table~\ref{tb:parameters}, apart from $\beta = 10$. The optimal solution is obtained for $\ell = 5$.}
    \label{fig:suppl_OCP9}
\end{figure}

\begin{figure}[h!]
    \centering
        \centering
        \includegraphics[scale = 0.5]{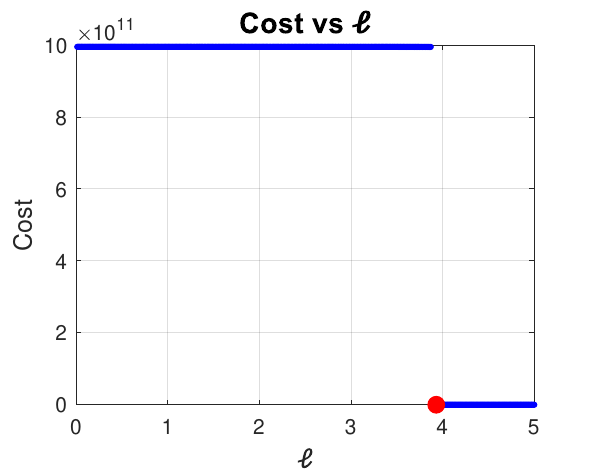}
        \includegraphics[scale = 0.4]{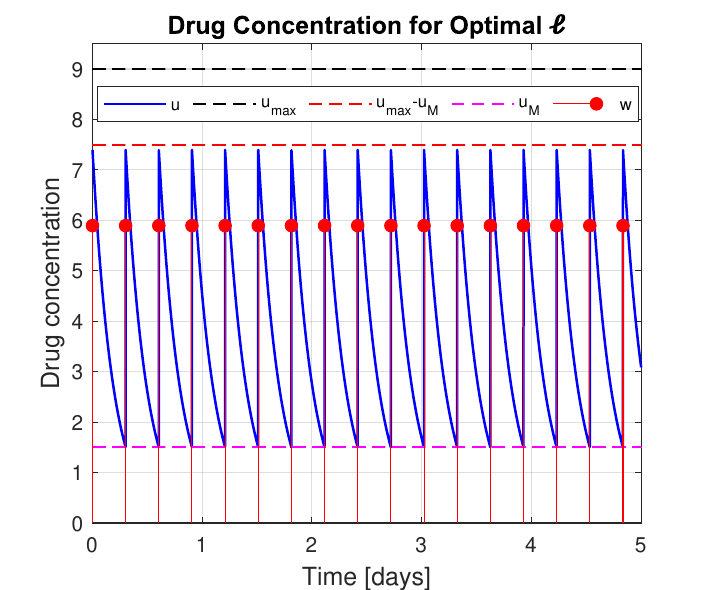}
        \includegraphics[scale = 0.5]{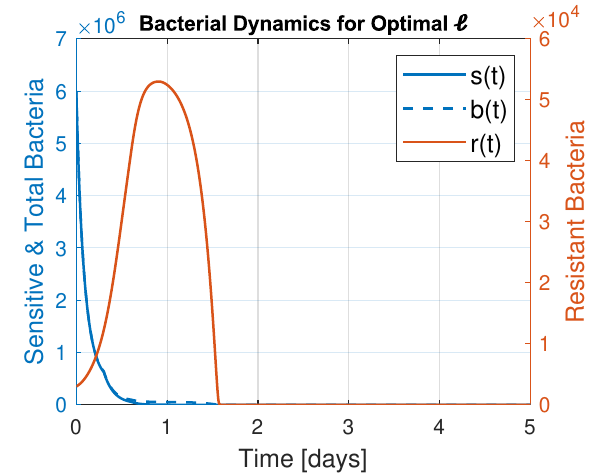}
    \caption{Optimal control problem \eqref{eq:optimalcontrolV_SP} with $\zeta_1= 10^5$, $\zeta_2 =10^5$ and $\zeta_3 = 1$ and initial condition IC1 in the success set ($s_0= 6\cdot 10^6$ and $r_0 = 3\cdot10^3$): computation of the cost $J$ as a function of $\ell$ (left) and then, with the selected optimal treatment, time evolution of the drug concentration and administered drug doses (middle), and time evolution of the state variables $b$, $s$ and of $r=b-s$ (right). The parameters are as listed in Table~\ref{tb:parameters}, apart from $\beta = 10$. The optimal solution is obtained for $\ell = 3.93$.}
    \label{fig:suppl_OCP10}
\end{figure}

\begin{figure}[h!]
    \centering
        \centering
        \includegraphics[scale = 0.5]{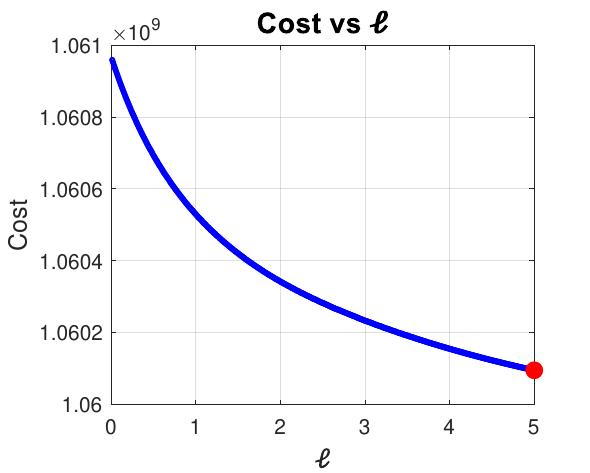}
        \includegraphics[scale = 0.4]{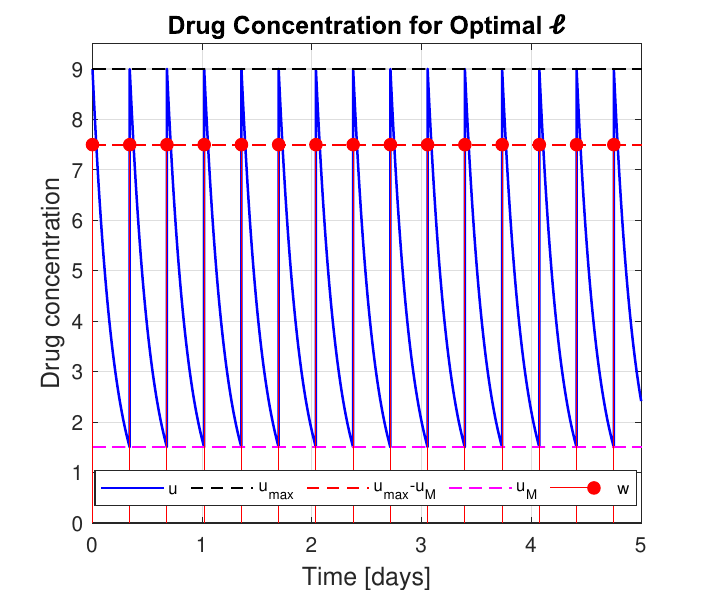}
        \includegraphics[scale = 0.5]{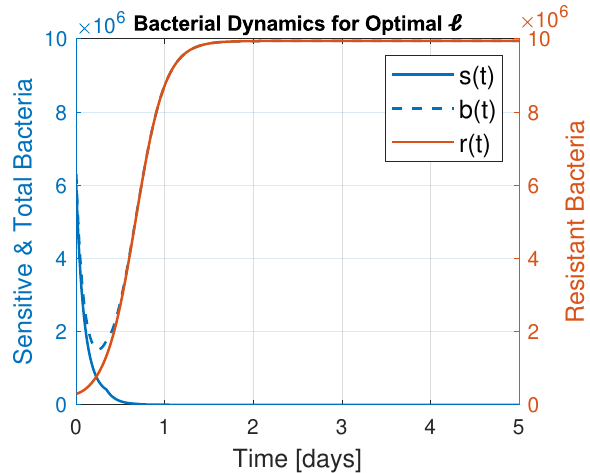}
    \caption{Optimal control problem \eqref{eq:optimalcontrolV_SP} with $\zeta_1=\zeta_2=\zeta_3=1$ and initial condition outside the success set ($s_0= 6\cdot 10^6$ and $r_0 = 3\cdot10^5$): computation of the cost $J$ as a function of $\ell$ (left) and then, with the selected optimal treatment, time evolution of the drug concentration and administered drug doses (middle), and time evolution of the state variables $b$, $s$ and of $r=b-s$ (right). The parameters are as listed in Table~\ref{tb:parameters}, apart from $\beta = 10$. The therapy is not effective, and $b \to b_2^*$.}
    \label{fig:suppl_OCP12}
\end{figure}

\begin{figure}[t]
    \centering
    \includegraphics[scale=0.5]{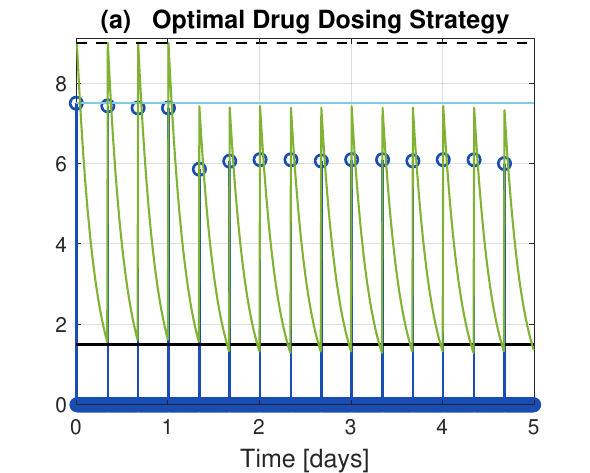}
    \includegraphics[scale=0.5]{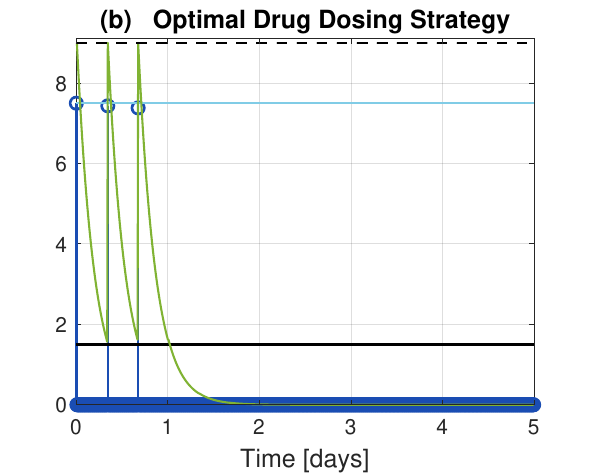}
    \includegraphics[scale=0.5]{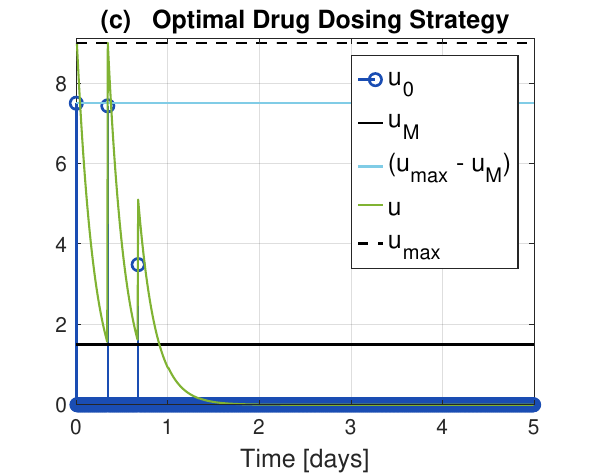}
    \includegraphics[scale=0.5]{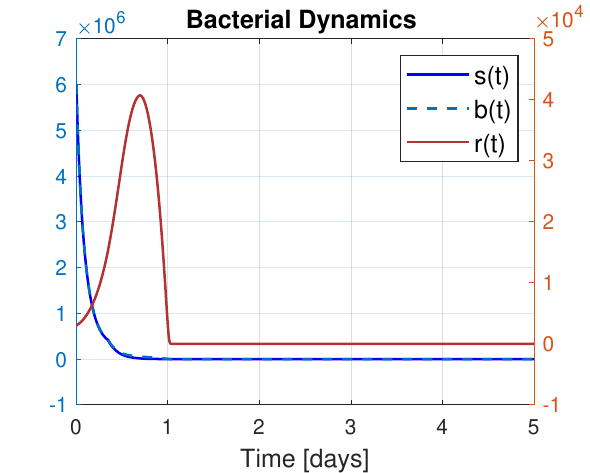} 
    \includegraphics[scale=0.5]{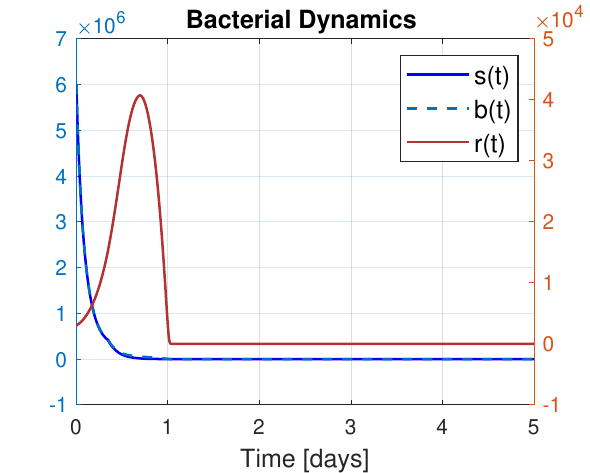}
    \includegraphics[scale=0.5]{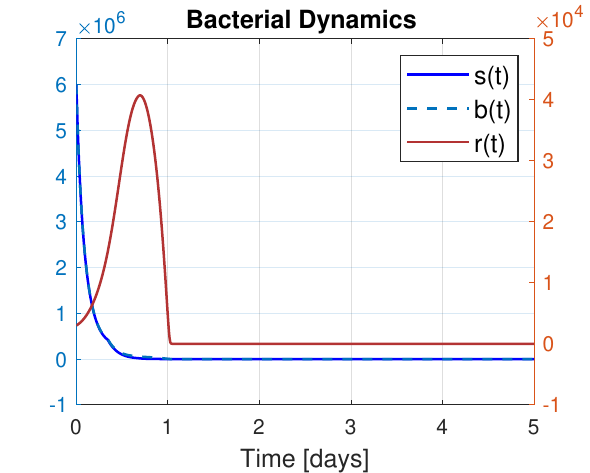}
    \caption{Optimal control problem \eqref{eq:ocp} with (a) $\zeta_3=10^3$ and $\zeta_1=1$, (b) $\zeta_3=\zeta_1=1$, (c) $\zeta_3=1$ and $\zeta_1=10^3$, and with initial condition IC1 in the success set ($s(0)=6\cdot 10^6$ and $r(0)=3 \cdot 10^3$): time evolution of the drug concentration and administered drug doses (first row) and time evolution of the state variables $b$, $s$ and of $r=b-s$ (second row). The parameters are as listed in Table~\ref{tb:parameters}, apart from $\beta = 10$.}
    \label{fig:ocp_test}
\end{figure}

\end{document}